\documentclass[twocolumn,a4paper,longbibliography,accepted=2025-02-17]{quantumarticle}
\pdfoutput=1
\usepackage[T1]{fontenc}
\usepackage{inputenc}
\usepackage{color}
\usepackage{babel}
\usepackage{physics}
\usepackage{mathtools}
\usepackage{bbding}
\usepackage{pifont}
\usepackage{textcomp}
\usepackage{wasysym}
\usepackage{amsthm}
\usepackage[normalem]{ulem}
\usepackage{nicefrac}
\usepackage{wasysym}
\usepackage{dsfont}
\usepackage[numbers]{natbib}
\usepackage{authblk}
\usepackage{amsmath}
\usepackage{amssymb}
\usepackage{graphicx}
\usepackage {hyperref}

\usepackage{float}
\hypersetup{
     colorlinks   = true,
     linkcolor    = blue,
     citecolor    = blue,
     urlcolor     = blue,
     }
\usepackage{appendix}
\renewcommand{\selectlanguage}[1]{}
\makeatletter
\@ifundefined{textcolor}{}
{%
	\definecolor{BLACK}{gray}{0}
	\definecolor{WHITE}{gray}{1}
	\definecolor{RED}{rgb}{1,0,0}
	\definecolor{GREEN}{rgb}{0,1,0}
	\definecolor{BLUE}{rgb}{0,0,1}
	\definecolor{CYAN}{cmyk}{1,0,0,0}
	\definecolor{MAGENTA}{cmyk}{0,1,0,0}
	\definecolor{YELLOW}{cmyk}{0,0,1,0}
}
\theoremstyle{plain}

\theoremstyle{plain}

\ifx\proof\undefined
\newenvironment{proof}[1][\protect\proofname]{\par
	\normalfont\topsep6\p@\@plus6\p@\relax
	\trivlist
	\itemindent\parindent
	\item[\hskip\labelsep
	\scshape
	#1]\ignorespaces
}{%
	\endtrivlist\@endpefalse
}
\providecommand{\proofname}{Proof}
\fi
\theoremstyle{plain}

\providecommand{\lemmaname}{Lemma}
\providecommand{\definitionname}{Definition}
\providecommand{\propositionname}{Proposition}

\usepackage{babel}
\usepackage{txfonts}
\usepackage{colortbl}\definecolor{myurlcolor}{rgb}{0,0,0.7}
\newcommand{\id}{{\operatorname{id}}}

\def\ket#1{| #1 \rangle}
\def\bra#1{\langle  #1 |}

\def\proj#1{| #1 \rangle\!\langle #1 |}

\newcommand{\haH}

\newtheorem{theorem}{Theorem}

\newtheorem{lemma}{Lemma}
\newtheorem{remark}{Remark}
\newtheorem{observation}{Observation}
\newtheorem{corollary}{Corollary}

\newtheorem{definition}{Definition}

\definecolor{orange}{RGB}{255,127,0}

\def\proj#1{| #1 \rangle\!\langle #1 |}

\begin{document}

\title{Squashed quantum non-Markovianity: a measure of genuine quantum non-Markovianity in states}

\author{Rajeev Gangwar}
\affiliation{Department of Physical Sciences, Indian Institute of Science Education and Research (IISER), Mohali, Punjab 140306, India}

\author{Tanmoy Pandit}
\affiliation{Fritz Haber Research Center for Molecular Dynamics, Hebrew University of Jerusalem, Jerusalem 9190401, Israel}

\author{Kaumudibikash Goswami}\email{goswami.kaumudibikash@gmail.com}
\affiliation{QICI Quantum Information and Computation Initiative, Department of Computer Science, The University of Hong Kong, Pokfulam Road, Hong Kong}

\author{Siddhartha Das}\email{das.seed@iiit.ac.in}
\affiliation{Center for Security, Theory and Algorithmic Research (CSTAR), Centre for Quantum Science and Technology (CQST), International Institute of Information Technology, Hyderabad, Gachibowli, Telangana 500032, India}
 
\author{Manabendra Nath Bera}
\email{mnbera@gmail.com}
\affiliation{Department of Physical Sciences, Indian Institute of Science Education and Research (IISER), Mohali, Punjab 140306, India}

\begin{abstract}
Quantum non-Markovianity in tripartite quantum states $\rho_{ABC}$ represents a correlation between systems $A$ and $C$ when conditioned on the system $B$ and is known to have both classical and quantum contributions. However, a systematic characterization of the latter is missing. To address this, we propose a faithful measure for non-Markovianity of genuine quantum origin called squashed quantum non-Markovianity (sQNM). It is based on the quantum conditional mutual information and is defined by the left-over non-Markovianity after squashing out all non-quantum contributions. It is lower bounded by the squashed entanglement between non-conditioning systems in the reduced state and is delimited by the extendibility of either of the non-conditioning systems. We show that the sQNM is monogamous, asymptotically continuous, convex, additive on tensor-product states, and generally super-additive. We characterize genuine quantum non-Markovianity as a resource via a convex resource theory after identifying free states with vanishing sQNM and free operations that do not increase sQNM in states. We use our resource-theoretic framework to bound the rate of state transformations under free operations and to study state transformation under non-free operations; in particular, we find the quantum communication cost from Bob ($B$) to Alice ($A$) or Charlie ($C$) is lower bounded by the change in sQNM in the states. The sQNM finds operational meaning; in particular, the optimal rate of private communication in a variant of conditional one-time pad protocol is twice the sQNM. Also, the minimum deconstruction cost for a variant of quantum deconstruction protocol is given twice the sQNM of the state.
\end{abstract}

\maketitle
	\normalfont
	\section{Introduction}
Understanding different kinds of quantum correlations has been a long-standing research quest in the field of quantum information and computation from both fundamental and technological aspects~\cite{EPR35,Bbook04,DiV95,DM03,CBTW17,DBWH21}. Quantum correlations are essential or useful to perform many of the desired quantum information and computation tasks that otherwise would not be feasible using only classical resources, such as quantum teleportation~\cite{BBC+93,FSB+98}, quantum key distribution~\cite{Eke91,DBWH21,PGT+23}, quantum computing~\cite{Sim97,BGK18}, true random number generation~\cite{PAM+10,RA11, Bera17}, etc. In the past couple of decades, resource-theoretic frameworks~\cite{CG19,TR19,Lam24} for several quantum correlations and properties, e.g., nonlocality~\cite{Bbook04,BLM+05,GWAN12}, steering~\cite{WJD07,GA15}, entanglement~\cite{HHHH09,SSD+15}, unextendibility~\cite{DPS02,DPS04,KDWW19,KDWW21}, etc.~have been introduced providing deeper insights into these and their viable applications. Another important correlation in composite quantum systems is quantum non-Markovianity: a correlation between two parties conditioned by a third one.  Although it is widely studied, a consistent resource-theoretic framework to characterize genuine quantum non-Markovianity as a quantum resource is still lacking.

An arbitrary tripartite quantum system is in a quantum non-Markovian state if it is not a quantum Markov state~\cite{hayden2004,FR15}, a quantum analog of a Markov chain in classical systems. A quantum state $\rho_{ABC}$ is a quantum Markov state $(A-B-C)$ if and only if the quantum conditional mutual information (QCMI) $I(A;C|B)_{\rho}$ vanishes~\cite{hayden2004}, i.e.,
\begin{align*}  
I(A;C|B)_\rho & \coloneqq S(AB)_{\rho}+S(BC)_{\rho}-S(ABC)_{\rho}-S(B)_{\rho}\nonumber\\ &=0, 
\end{align*}
where $S(B)_{\rho}\coloneqq -\Tr[\rho_{B}\log_2\rho_B]$ is the von Neumann entropy of a reduced state $\rho_{B}=\Tr_{AC}[\rho_{ABC}]$. It is always non-negative for quantum states. This property is equivalent~\cite{Rus02} to the strong subadditivity property of the von Neumann entropy~\cite{Lieb03, NP05} and the monotonicity of quantum relative entropy between states under the action of quantum channels~\cite{Ume62,HP91}, the quantum data-processing inequality. The QCMI is shown to have operational meaning in the quantum communication tasks of quantum state redistribution~\cite{Devetak08, Yard09}, quantum state deconstruction and conditional erasing~\cite{Berta_2018_PRL, berta2009singleshot}, and quantum one-time pad~\cite{Conditional_otp}. It is one of the widely used entropic quantities in different areas of quantum physics such as characterization of quantum correlations~\cite{Brandao11,DHW16,Kaur17,Kim12a}, quantum cryptography~\cite{RW03,DW05,WDH22,KHD22}, quantum error correction~\cite{CC97,hayden2004,PEW17}, characterization of memory effects in quantum dynamics~\cite{Bus14,BDW16,HG22,PRF+18}, analyzing properties of states and dynamics in many-body quantum systems~\cite{Kim12a,KKB20,KSKA22}.

Here we primarily focus on the resource theory of genuine quantum non-Markovianity for quantum states where the free states are (an extended class of) quantum Markov chain states~\cite{hayden2004,SFR16}. This is different from the notion of non-Markovian and Markovian memory effects in quantum dynamics~\cite{Rivas2014,Allen17,BGG+24}, which we discuss only towards the end of the paper in Section~\ref{sec:Disc}. The connections between quantum Markov chain states and quantum Markovian dynamics are discussed in~\cite{Bus14,BDW16}. In a recent work~\cite{BGG+24}, the notion of genuine non-Markovianity in quantum states developed in this work is applied to characterize memory effects in the revival of information in quantum dynamics.

The condition $I(A;C|B)_{\rho}=0$ signifies that the state $\rho_{ABC}$ can be perfectly recovered from $\rho_{AB}$ (or $\rho_{BC}$) via a quantum channel acting on $B$. The set of states with vanishing QCMI is not convex, i.e., a quantum non-Markovian state can be produced by taking a classical mixture of some quantum Markov states. That implies that non-Markovian behavior in many quantum non-Markovian states is not of quantum origin. The structure of quantum non-Markovian states is complicated. There exist quantum states with small QCMI that can be far, in terms of trace distance, from all quantum Markov chain states in the form~\eqref{eq:MarkovStructure}~\cite{FR15}. It is now well-understood that an approximate quantum Markov state (with very small QCMI) $\rho_{ABC}$ is approximately recoverable from its bipartite reduced states $\rho_{AB}$ (or $\rho_{BC}$)~\cite{FR15}. There are recent attempts to devise a resource-theoretic framework of quantum non-Markovianity in states~\cite{Wakakuwa2017, Wakakuwa21}. Still, they do not have all the desirable properties (see Table~\ref{tab:Comparison}) and are unable to quantify \textit{genuine} quantum non-Markovianity, i.e., quantum non-Markovian correlations in $\rho_{ABC}$ that are not of classical origin.

To tackle the issues above, we introduce a measure of genuine quantum non-Markovianity, which we call squashed quantum non-Markovianity (sQNM). It is based on quantum conditional mutual information and quantifies non-Markovinaity after squashing out all possible non-Markovian correlations of a non-genuinely quantum origin. This measure is in a similar spirit of squashed entanglement~\cite{Christandl04}, a faithful entanglement monotone. The genuine quantum non-Markovianity can also be quantified in terms of generalized divergence of recovery. With these in hand, we can formulate a resource-theoretic framework to characterize genuine quantum non-Markovianity in quantum states as a resource and discuss its role or connection in some of the operational quantum information processing tasks. 

We start with the properties of sQNM. The states with vanishing sQNM form a convex set. These states are considered free states. The measure satisfies various desired information-theoretic properties, such as convexity, super-additivity, monogamy, asymptotic continuity, and faithfulness. We also characterize the quantum operations that do not increase or create sQNM in states. These operations also form a convex set and are considered free operations. With well-defined free states and operations, we study the conditions for state transformations under free operations and derive various bounds on the rate of transformations. 

We explore the interrelation between quantum communication from the conditioning system to the non-conditioning system and sQNM and demonstrate that the quantum communication cost is lower bounded by the change in sQNM in the states. We also compare the sQNM with the squashed entanglement and observe that the sQNM of a state $\rho_{ABC}$ is lower bounded by the squashed entanglement of the reduced state $\rho_{AC}$ on a non-conditioning system. Furthermore, genuine quantum non-Markovianity is delimited by the extendibility in non-conditioning systems. We study the possible implications of our results in two quantum information-theoretic tasks. For conditional quantum one-time pad~\cite{Conditional_otp}, we show that the amount of fully secured quantum key that can be shared using a tripartite state is given by the sQNM of the state. For quantum deconstruction of a state ~\cite{Berta_2018_PRL, Berta_2018}, we find that, on average, the minimum number of unitaries required is given by the sQNM of the state.

The paper is organized as follows. Section~\ref{sec:sQNM} introduces the measure of genuine quantum non-Markovianity in states, outlines the properties of squashed quantum non-Markovianity (sQNM), and classifies the set of operations that do not create or increase sQNM in states. Section~\ref{sec:StateTrans} deals with state transformations under free operations and explores the interrelation between sQNM and quantum communications costs in state transformations. Two quantum information theoretic tasks are reconsidered in Section~\ref{sec:QuantTasks} to provide an operational meaning to sQNM. Finally, we discuss the possible implications of our results and conclude in Section~\ref{sec:Disc}.

\begin{table}
\centering
\resizebox{\linewidth}{!}{
\begin{tabular}{|l|l|l|l|}
\hline
     & $I(A;C|B)_\rho$~\cite{Wakakuwa2017}  & $M_f(A;C|B)_\rho$~\cite{Wakakuwa21}   & $N_{\mathrm{sq}}(A;C|B)_\rho$  \\ 
\hline
Convexity            & \texttimes   & \texttimes / \checkmark \footnote{Convexity is satisfied only when the conditional system extends classically. This is termed as conditional convexity~\cite{Wakakuwa21}.}   & \checkmark \\
Super-additivity     & \texttimes   & \texttimes        & \checkmark  \\ 
Monogamy             & \texttimes   & \texttimes        & \checkmark  \\ 
Continuity           & \checkmark   & \checkmark        & \checkmark   \\   
Faithfulness         & \checkmark   & \checkmark        & \checkmark   \\
\hline
\end{tabular}}
\caption{The table provides a comparison between basic properties of the measures of quantum non-Markovianity in quantum states: QCMI $I(A;C|B)_\rho$~\cite{Wakakuwa2017}, non-Markovianity of formation $M_f(A;C|B)_\rho$~\cite{Wakakuwa21}, and the sQNM $N_{\mathrm{sq}}(A;C|B)_\rho$~\eqref{eq:genQNM}. The checkmark (\checkmark) denotes that the corresponding property holds, and the cross (\texttimes) denotes that the corresponding property does not hold. Due to differences in the properties of the measures used to define the set of free states, the corresponding resource theories are also different. For instance, the frameworks proposed in~\cite{Wakakuwa2017} and~\cite{Wakakuwa21} are non-convex and conditionally convex, respectively, while the one we formulate here is convex.}
\label{tab:Comparison}
\end{table}

\section{Characterizing genuine quantum non-Markovianity \label{sec:sQNM}}
The QCMI for an arbitrary state $\rho_{ABC}$ can also be written as
\begin{align}
    I(A;C|B)_{\rho}\coloneqq I(AB;C)_{\rho}-I(B;C)_{\rho}\geq 0,
\end{align}
where $I(B;C)_{\rho}\coloneqq S(B)_{\rho}+S(C)_{\rho}-S(BC)_{\rho}$ is the quantum mutual information between $B,C$ in the state $\rho_{BC}$, and it also always holds that $I(A;C|B)_{\rho}=I(C;A|B)_{\rho}$. For any quantum state $\rho_{ABC}$, there exists a universal recovery map $\mathcal{R}_{B\to BC}$ such that~\cite{SFR16,JRS+18} 
\begin{align}\label{eq:qcmi-fid}
    I(A;C|B)_{\rho}\geq -\log_2 F(\rho_{ABC},\mathcal{R}_{B\to BC}(\rho_{AB})),
\end{align}
where $F(\rho,\sigma)\coloneqq \norm{\sqrt{\rho}\sqrt{\sigma}}_1^2$ is the fidelity with $\norm{M}_1\coloneqq \Tr [\sqrt{M^\dag M}]$ denoting the trace-norm of an operator $M$. A universal recovery map $\mathcal{R}_{B\to BC}$ in inequality~\eqref{eq:qcmi-fid} is a quantum channel, i.e., a completely positive, trace-preserving (CPTP) map, depending only on $\rho_{BC}$ (not on $\rho_{ABC}$) and $\mathcal{R}_{B\to BC}(\rho_{B})=\rho_{BC}$. From inequality~\eqref{eq:qcmi-fid}, we see that for any state $\rho_{ABC}$ there exists a universal recovery map $\mathcal{R}_{B\to BC}$ such that~\cite{FR15,SFR16,JRS+18}
\begin{align}
    F(\rho_{ABC},\mathcal{R}_{B\to BC}(\rho_{AB}))\geq 2^{-I(A;C|B)_{\rho}}.
\end{align}
A state $\rho_{ABC}$ is a quantum Markov state ($A-B-C$), i.e., $I(A;C|B)_{\rho}=0$, if and only if there exists a universal map $\mathcal{R}_{B\to BC}$ such that $\rho_{ABC}=\mathcal{R}_{B\to BC}(\rho_{AB})$. An approximate Markov state $\rho_{ABC}$ is approximately recoverable from $\rho_{AB}$ (or $\rho_{BC}$) via an universal recovery map $\mathcal{R}_{B\to BC}$ (or $\mathcal{R}_{B\to AB}$) as they have small QCMI, i.e., $I(A;C|B)_{\rho}\leq \varepsilon$ for some appropriately small $\varepsilon$.

A quantum state $\rho_{ABC}$ for which the QCMI is non-zero, $I(A;C|B)_\rho \neq 0$, is a non-Markov state. However, the Markov states do not form a convex set. There are many examples of quantum Markov states ${\sigma^x_{ABC}}_{x\in\mathcal{X}}$, i.e., $I(A;C|B)_{\sigma^x}=0$ for all $x\in\mathcal{X}$, whose non-trivial probabilistic mixtures $\sigma_{ABC}=\sum_{x\in\mathcal{X}}p_X(x)\sigma^x_{ABC}$ for $0\leq p_x\leq 1$ and $\sum_xp_X(x)=1$ yield non-zero QCMI, $I(A;C|B)_{\sigma}>0$. It can be easily checked by assuming $ABC$ as a three-qubit system with $\sigma_{ABC}= \proj{000}_{ABC}$ and $\tau_{ABC}= \proj{101}_{ABC}$. Mixing a set of states following a probability distribution is a classical operation. Thus, it is clear that the non-Markovian nature of a quantum state may have a classical origin. Beyond this, some states have non-Markov properties of genuine quantum origin and cannot be created by classical processing like the one mentioned above. On the other hand, as we show later, a probabilistic mixture of genuine quantum non-Markov states may result in a Markov state (see Observation~\ref{obs:QuantExt}). Hence, non-Markovianity in a quantum state may have both classical and quantum contributions.

\subsection{Squashed quantum non-Markovianity}

We aim to characterize states having non-Markovianity of genuine quantum origin. Here, we introduce a faithful measure quantifying genuine quantum non-Markovianity in a state. We call this measure as {\it squashed quantum non-Markovianity (sQNM)}. 
\begin{definition}[Squashed quantum non-Markovianity]\label{defi:QNM}
The squashed quantum non-Markovianity (sQNM) $N_{\mathrm{sq}}$ of an arbitrary tripartite quantum state $\rho_{ABC}$ is defined by a genuine quantum correlation between two subsystems $A$ and $C$ when conditioned on the subsystem $B$, as
\begin{align}\label{eq:QNMMeas} 
N_{\mathrm{sq}}(A;C|B)_\rho:=\frac{1}{2}\inf_{\rho_{ABCE}} I(A;C|BE)_\rho,    
\end{align}
where the infimum is taken over the set of all state extensions $\rho_{ABCE}$ of the state $\rho_{ABC}$, i.e., $\rho_{ABC}=\Tr_E[\rho_{ABCE}]$. 
\end{definition}

It trivially follows from the above definition that $N_{\mathrm{sq}}(A;C|B)_{\rho}=N_{\operatorname{sq}}(C;A|B)_{\rho}$ for any state $\rho_{ABC}$. A priori $|E|$ is unbounded (cf.~\cite{Christandl04,BG14}). As a direct consequence of the above definition, we see that for an arbitrary pure state $\psi_{ABC}$ of arbitrary dimensional systems $ABC$, 
\begin{align}
    N_{\mathrm{sq}}(A;C|B)_{\psi}=\frac{1}{2}I(A;C|B)_{\psi} = \frac{1}{2}I(A;C)_{\psi}.
\end{align}
For example, let us consider two important classes of pure entangled states for three-qubit systems: Greenberger-Horne-Zeilinger (GHZ)~\cite{GHZ89, mermin90} and W states~\cite{DVC00}. The three-qubit GHZ state $\op{\mathrm{GHZ}}$ and W state $\op{W}$ are given by
\begin{align}
&\ket{\mathrm{GHZ}}_{ABC}=\frac{1}{\sqrt{2}}\bigg(\ket{000}+\ket{111}\bigg)_{ABC}, \nonumber \\
&\ket{\mathrm{W}}_{ABC}=\frac{1}{\sqrt{3}}\bigg(\ket{001}+\ket{010}+\ket{100}\bigg)_{ABC}\nonumber \\
&\qquad =\sqrt{\frac{2}{3}}\ket{0}_A\ket{\psi^+}_{BC}+\sqrt{\frac{1}{3}}\ket{1}_A\ket{00}_{BC},  \end{align}
where $\ket{\psi^+}=(\ket{01}+\ket{10})/\sqrt{2}$ corresponds to an EPR-Bell state. The sQNM for three-qubit GHZ and W states are $N_{\mathrm{sq}}(A;C|B)_{\mathrm{GHZ}}=0.5$ and $N_{\mathrm{sq}}(A;C|B)_{\mathrm{W}}=0.4591$, respectively. More generally, for the following families $\{\mathrm{GHZ}(p)\}_{0\leq p\leq 1}$ and $\{\mathrm{W}(p)\}_{0\leq p\leq 1}$ of three-qubit GHZ and W states, where $\mathrm{GHZ}(p)=\op{\mathrm{GHZ}(p)}$ and $\mathrm{W}(p)=\op{\mathrm{W}(p)}$ such that
\begin{align}
&\ket{\mathrm{GHZ}(p)}_{ABC}:=\sqrt{p}\ket{000}_{ABC}+\sqrt{1-p}\ket{111}_{ABC},  \label{eq:GHZ} \\
&\ket{\mathrm{W}(p)}_{ABC} :=\sqrt{p}\ket{0}_A\ket{\psi^+}_{BC}+\sqrt{1-p}\ket{1}_A\ket{00}_{BC},\label{eq:wstate}\end{align}
we have $N_{\mathrm{sq}}(A;C|B)_{\mathrm{GHZ}(p)}=h_2(p)/2=N_{\mathrm{sq}}(A;C|B)_{\mathrm{W}(p)}$ with $h_2(p):=-(1-p)\log_2(1-p)-p\log_2p$ being the binary Shannon entropy.

In contrast, if the systems $B$ and $A$ are swapped in $\mathrm{W}_{ABC}$~\eqref{eq:wstate} above to get $\widetilde{\mathrm{W}}_{ABC}$, we have another family $\{\widetilde{\mathrm{W}}(p)\}_{0\leq p\leq 1}$ of states for which $N_{\mathrm{sq}}(A;C|B)_{\widetilde{\mathrm{W}}(p)}=h_2(p/2)-h_2(p)/2$, a monotonically increasing function in $p$. In the limiting case of $p=1$, we have $\ket{\widetilde{\mathrm{W}}(1)}=\ket{0}_B\ket{\psi^+}_{AC}$, and $N_{\mathrm{sq}}(A;C|B)_{\widetilde{\mathrm{W}}(p)}=1$. This is a special case of the following observation. If the pure state $\psi_{ABC}$ is a product state between $AC$ and $B$, i.e., of the form $\psi_{ABC}=\varphi_{AC}\otimes\phi_{B}$, then
\begin{align}
     N_{\mathrm{sq}}(A;C|B)_{\psi}=\frac{1}{2} I(A;C)_{\varphi}=S(A)_{\varphi}=S(C)_{\varphi}.
\end{align}

The definition of $N_{\mathrm{sq}}$ shares the same spirit of squashed entanglement ($E_{\mathrm{sq}}$), where the latter is defined as~\cite{Christandl04} 
\begin{align}
E_{\mathrm{sq}} (A;C)_\sigma :=\frac{1}{2} \inf_{\sigma_{ACF}} I(A;C|F)_\sigma,
\end{align}
for an arbitrary bipartite state $\sigma_{AC}$, where the infimum is carried out over the set of all state extensions $\sigma_{ACF}$, i.e., $\sigma_{AC}=\Tr_F [\sigma_{ACF}]$. It is a faithful measure of quantum entanglement and vanishes if and only if $\sigma_{AC}$ is a separable state in the partition $A|C$~\cite{Brandao11}.

Any finite-dimensional quantum state $\rho_{ACF}$ with $I(A;C|F)_\rho=0$ satisfies a Markov chain structure~\cite{hayden2004},
\begin{align}
\rho_{ACF}=\oplus_i p_i \ \rho_{AF^L_i} \otimes \rho_{F^R_iC}, \label{eq:MarkovStructure} 
\end{align}
where the Hilbert space of system $F$ has a decomposition $F=\oplus_{i}F^L_i\otimes F^R_i$. It can be easily seen that the reduced state $\rho_{AC}$ is separable. While the measures $N_{\mathrm{sq}}$ and $E_{\mathrm{sq}}$ share several interesting properties, they represent qualitatively different aspects of a quantum state. The $E_{\mathrm{sq}}$ signifies genuine bipartite quantum correlation, that is, quantum entanglement. In contrast, for a tripartite system, $N_{\mathrm{sq}}$ represents a genuine bipartite quantum correlation with the condition of having complete knowledge of the third system. The $E_{\mathrm{sq}}$ and $N_{\mathrm{sq}}$ satisfy the relationship mentioned in the lemma below. We refer to Appendix~\ref{app:Lemma-sQNMsEnt} for the proof.

\begin{lemma}\label{lem:sQNMsEnt}
For any finite-dimensional tripartite system $ABC$ in an arbitrary state $\rho_{ABC}$ and any isometry $V: B \to B_L \otimes B_R$ transforming $\rho_{ABC} \to \sigma_{AB_LB_RC}$, i.e., $V(\rho_{ABC})V^\dagger=\sigma_{AB_LB_RC}$,
\begin{align}
&E_{\mathrm{sq}}(A;C)_\rho\leq N_{\mathrm{sq}}(A;C|B)_\rho  \leq\nonumber\\
 &\min \left\{E_{\mathrm{sq}}(AB_L;B_RC)_\sigma, E_{\mathrm{sq}}(AB;C)_\rho, E_{\mathrm{sq}}(A;BC)_\rho\right\}.
\end{align}
\end{lemma}
Clearly, if a state $\rho_{ABC}$ is separable in any of the partitions $A|BC$, $AB|C$, or $\sigma_{AB_LB_RC}$ is separable in the partition $AB_L|B_RC$, then the sQNM vanishes. Separability between $A$ and $C$ does not necessarily imply that the state has vanishing sQNM, e.g., a tripartite GHZ state~\eqref{eq:GHZ} for which $0=E_{\rm sq}(A;C)_{\rm GHZ(\frac{1}{2})}<N_{\rm sq}(A;C|B)_{\rm GHZ(\frac{1}{2})}=1$. For a tripartite state $\psi_{AC}\otimes\omega_B$, $E_{\mathrm{sq}}(A;C)=N_{\mathrm{sq}}(A;C|B)$ if $\psi_{AC}$ is pure. The measure $N_{\rm sq}$ is a genuine tripartite quantum correlation that can be considered in-between bipartite and tripartite entanglement.

Now, we dwell on various important properties satisfied by the measure sQNM and state them in the following lemma. See Appendix~\ref{app:lem-Qprop1} for the proof of (P1)-(P5).  

\begin{lemma}\label{lem:Qprop1}
The squashed quantum non-Markovianity $N_{\mathrm{sq}}$ of arbitrary finite-dimensional systems satisfies the following properties:

\noindent
(P1) Convexity: For any state $\rho_{ABC}= \sum_m p_m \  \rho_{ABC}^m $, where $0 \leq p_m \leq 1$ for all $m$ and  $\sum_m p_m=1$,
\begin{align}
N_{\mathrm{sq}}(A;C|B)_\rho \leq \sum_m p_m \ N_{\mathrm{sq}}(A;C|B)_{\rho^m}.     
\end{align}
\noindent
(P2) Super-additivity: For every quantum state $\rho_{AA^\prime BB^\prime CC^\prime}$, the sQNM is super additive,
\begin{align}
    &N_{\mathrm{sq}}(AA^\prime;CC^\prime|BB^\prime)_{\rho}\nonumber \\ &\hspace{1cm} \geq N_{\mathrm{sq}}(A;C|B)_{\rho} + N_{\mathrm{sq}}(A^\prime;C^\prime|B^\prime)_{\rho},
\end{align}
and, for a state with tensor product structure $\rho_{AA^\prime BB^\prime CC^\prime} = \rho_{A B C} \otimes \rho_{A^\prime B^\prime C^\prime}$, it is additive,
\begin{align}
 &N_{\mathrm{sq}}(AA^\prime;CC^\prime|BB^\prime)_{\rho}\nonumber \\
 & \hspace{1cm}= N_{\mathrm{sq}}(A;C|B)_{\rho} + N_{\mathrm{sq}}(A^\prime;C^\prime|B^\prime)_{\rho}.   
\end{align}

\noindent
(P3) Monogamy: For every state $\rho_{AA^\prime B C}$ the sQNM satisfies a monogamy relation
\begin{align}
 &N_{\mathrm{sq}}(AA^\prime ; C|B)_\rho \nonumber \\ 
 &\hspace{1cm}\geq N_{\mathrm{sq}}(A ; C|B)_\rho + N_{\mathrm{sq}}(A^\prime ; C|B)_\rho. 
\end{align}

\noindent
(P4) Asymptotic continuity: For any two states $\rho_{ABC}$ and $\sigma_{ABC}$ with trace-distance $\frac{1}{2}\norm{\rho_{ABC} - \sigma_{ABC}}_1 \le\varepsilon$ with $0 \leq \varepsilon \leq 1$,
 \begin{align}
  \abs{  N_{\mathrm{sq}}(A;C|B)_\rho - N_{\mathrm{sq}}(A;C|B)_\sigma } \le \frac{f(\varepsilon)}{2},  
 \end{align}
where $f(\varepsilon) \to 0$ as $\varepsilon \to 0$.

 \noindent
 (P5) Faithfulness: For any state $\rho_{ABC}$, squashed quantum non-Markovianity is zero, i.e., $N_{\mathrm{sq}}(A;C|B)_{\rho}=0$, if and only if the state $\rho_{ABC}$ is not genuinely quantum non-Markovian. Meaning that, $N_{\mathrm{sq}}(A;C|B)_{\rho}=0$ if and only if there exists a universal recovery channel $\mathcal{R}_{BE\to BCE}$ such that 
\begin{align}
\mathcal{R}_{BE\to BCE}(\rho_{ABE})=\rho_{ABCE},  \label{faith_ful}
\end{align}
there exits state extension $\rho_{ABCE}$ of $\rho_{ABC}$, such that  Eq.~\eqref{faith_ful} is true. Also see Remark~\ref{rem:gendivfaith}.
\end{lemma}

\begin{figure}
\centering
\includegraphics{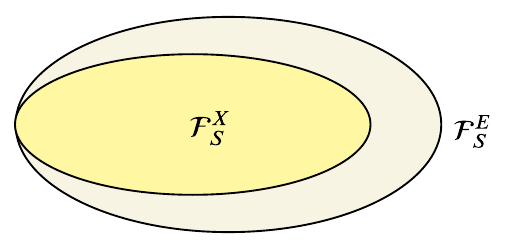}
\caption{{\it{Set of states with vanishing sQNM.}} Here, $\mathcal{F}^{X}_S$ refers to the convex set of free states with the optimal extension for $N_{\mathrm{sq}}{=}0$ achieved through a classical extension, see Eq.~\eqref{eq:ClassExt}. Outside of $\mathcal{F}^{X}_S$, there exist free states for which $N_{\mathrm{sq}}{=}0$  is achieved via extensions beyond classical registers, see Observation~\ref{obs:QuantExt}. In general, we represent the convex set of free states with the most general extension by $\mathcal{F}^{E}_S$, with the set $\mathcal{F}^{X}_S$ contained within it, i.e.,  $\mathcal{F}^{X}_S\subsetneq \mathcal{F}^{E}_S$ holds.
\label{fig:SetOfStates}}
\end{figure}

Any quantum property, such as sQNM, quantified after squashing out all the non-quantum contributions should satisfy the convexity property (P1). This, in other words, implies that one cannot increase sQNM on average by classically mixing any set of states. As consequence, all the states with $N_{\mathrm{sq}}=0$ (i.e., {\it free states}), form a convex set $\mathcal{F}_S^E$ (see Fig.~\ref{fig:SetOfStates} for details). Note that a convex mixture of the states in this set may result in a state with non-zero QCMI. However, a classical extension of the state and considering the extended system as part of the condition leads to a vanishing QCMI. For example, consider a set $\{\rho^x_{ABC}\}_{x\in\mathcal{X}}$ of quantum states such that $I(A;C|B)_{\rho^x}=0$ for each $x\in\mathcal{X}$. In general, for nontrivial convex combination $\rho^X_{ABC}\coloneqq\sum_{x}p_X(x)\rho^x_{ABC}$, $I(A;C|B)_{\rho^X}\neq 0$, unless states $\rho^x_{ABC}$ have special structure. If we take state extension 
\begin{align}
\rho_{ABCX}=\sum_xp_X(x)\rho^x_{ABC}\otimes\op{x}_X,  \label{eq:ClassExt}  
\end{align}
given that $\operatorname{Tr}_X\{\rho_{ABCX}\}=\sum_xp(x)\rho^x_{ABC}$, then $I(A;C|BX)_{\rho}=0$ and hence there exists a universal map $\mathcal{R}_{BX\to BCX}$ where $\mathcal{R}_{BX\to BCX}(\rho_{ABX})=\rho_{ABCX}$. This is how the classical contribution to the non-vanishing QCMI is squashed out. Nevertheless, these states form a convex set, which we denote as $\mathcal{F}^{X}_S$. The convex set $\mathcal{F}^E_S$ of states with vanishing sQNM is proper super-set of $\mathcal{F}^{X}_S$, that is $ \mathcal{F}^X_S\subsetneq \mathcal{F}^{E}_S$ (see Fig.~\ref{fig:SetOfStates}).

In a more general setting, there are states $\rho_{ABC}$ with respective state extension $\rho_{ABCE}$, i.e., $\Tr_{E}[\rho_{ABCE}]=\rho_{ABC}$, such that $I(A;C|B)_{\rho}\neq 0$ but do have $I(A;C|BE)_{\rho}=0$, where the system $E$ need not be a classical register. For each such states $\rho_{ABC}$, there exists a universal recovery map $\mathcal{R}_{BE\to BCE}$ so that $\mathcal{R}_{BE\to BEC}(\rho_{ABE})=\rho_{ABCE}$. In this regard, we highlight the following observation.

\begin{observation}\label{obs:QuantExt}
A probabilistic mixture of genuine quantum non-Markov states may result in a state with vanishing sQNM. One such example is a set of states $\{ \ketbra{i}{i}_B \otimes \ketbra{\phi_i}{\phi_i}_{AC} \}_i$ where $\{\proj{\phi_i}_{AC}\}_i$ are the four Bell states shared between $A$ and $C$. While each state $\ketbra{i}{i}_B \otimes \ketbra{\phi_i}{\phi_i}_{AC}$ has $N_{\mathrm{sq}}(A;C|B)=1$, a probabilistic mixture $\sigma_{ABC}=\frac{1}{4} \sum_{i} \ketbra{i}{i}_B \otimes \ketbra{\phi_i}{\phi_i}_{AC}$ has vanishing sQNM. More importantly, the extension of $\sigma_{ABC}$ for which we have $I(A;C|BE)_\sigma =0$, is a four-partite pure state $\sigma_{ABCE}=\proj{\psi}_{ABCE}$ with $\ket{\psi}_{ABCE}=\frac{1}{2}\sum_i \ \ket{ii}_{BE} \otimes \ket{\phi_i}_{AC}$. Here, the state extension is genuinely quantum, and having access to that extension as part of the condition leads to vanishing QCMI.
\end{observation}
  
Generally, for any form of classical correlation in a multipartite state, the global amount of that property cannot exceed the sum of local ones. However, on the contrary, this is not true for quantum correlations such as (squashed) entanglement. The sQNM satisfies super-additivity (P2) and quantifies a correlation of genuine quantum nature. The sQNM also satisfies monogamy property (P3), signifying that, for a multipartite system, the more a quantum system shares a genuine quantum non-Markovian correlation with another, then the less genuine quantum non-Markovian correlation it shares with the others. The asymptotic continuity (P4)  is another desirable property satisfied by the sQNM measure of quantum non-Markovianity. As expected, any two states, asymptotically close in trace distance, are also asymptotically close in their sQNM. These properties, along with faithfulness (P5), indicate that the sQNM may be considered a quantum resource, and a systematic characterization would prepare the base for a resource theory of genuine quantum non-Markovianity. It is worth mentioning that there have been attempts to study quantum non-Markovianity as a resource. For instance, the QCMI is used to quantify quantum non-Markovianity in~\cite{Wakakuwa2017}. In~\cite{Wakakuwa21}, non-Markovianity of formation is introduced from a resource theoretic approach. However, these quantifiers do not satisfy all the properties satisfied by sQNM. See Table~\ref{tab:Comparison} for a comparison. Apart from that, the set $\mathcal{F}_S^W$ of free states considered in~\cite{Wakakuwa2017} and~\cite{Wakakuwa21} are the Markov states, and, thus, it is a subset of $\mathcal{F}_S^E$, i.e., $\mathcal{F}_S^W \subsetneq \mathcal{F}_S^E$. 

The extendibility of quantum states in non-conditioning systems limits the quantum non-Markovianity, and it depicts the monogamous behavior of quantum non-Markovianity, as shown below. A state $\rho_{AC}$ is said to be $k$-extendible in $C$ for a $k\geq 2, k\in \mathbb{N}$, if there exists a state extension $\sigma_{AC_1C_2\ldots C_k}$, where $C_1\simeq C_2\simeq \cdots\simeq C_k\simeq C$, such that $\sigma_{AC_i}\coloneqq \Tr_{C_1\ldots C_k\setminus {C_i}}[\sigma_{C_1\ldots C_k}]=\rho_{AC}$ for all $i\in \{1,2,\ldots, k\}$~\cite{DPS02}.

\begin{theorem}
For an arbitrary quantum state $\rho_{ABC}$ that is $k$-extendible in $C$, the squashed quantum non-Markovianity is upper bounded by 
\begin{align}\label{eq:qnm-ext}
N_{\mathrm{sq}}(A;C|B)_{\rho}\leq \frac{\log_2|A|}{k}.
\end{align}
\end{theorem}
 \begin{proof}
Let $\rho_{ABC}$ be $k$-extendible in $C$. Using properties of QCMI and monogamy property of squashed quantum non-Markovianity, we have
\begin{align}
\log_2|A| &\geq N_{\mathrm{sq}}(A;C_1C_2\ldots C_{k}|B)_{\rho} \nonumber \\
&\geq N_{\mathrm{sq}}(A;C_1|B)_{\rho}+\cdots +N_{\mathrm{sq}}(A;C_{k}|B)_{\rho} \nonumber \\
& = k \ N_{\mathrm{sq}}(A;C|B)_{\rho},
\end{align}
which yields bound in Eq.~\eqref{eq:qnm-ext}.
\end{proof}

It was shown in~\cite{Li18} that for any state $\rho_{AC}$ that is $k$-extendibile in $C$ has $E_{\mathrm{sq}}(A;C)_\rho \leq \frac{\log_2|A|}{k}$, and we have $ E_{\mathrm{sq}}(A;C)_\rho \leq N_{\mathrm{sq}}(A;C|B)_{\rho}$ from Lemma~\ref{lem:sQNMsEnt}. Based on the theorem above, we can see that the bound~\eqref{eq:qnm-ext} provides a tighter relation for $E_{sq}$ in the sense that for all states $\rho_{ABC}$ that is $k$-extendible in $C$, we have
\begin{align}
    E_{\mathrm{sq}}(A;C)_\rho \leq N_{\mathrm{sq}}(A;C|B)_{\rho}\leq \frac{\log_2|A|}{k}.
\end{align}

Now, we turn to the allowed operations that do not increase sQNM. These operations can also be considered the {\it free operations} for a resource theory characterizing quantum non-Markovianity. 

\begin{figure}[H]
\centering
\includegraphics[scale=0.6]{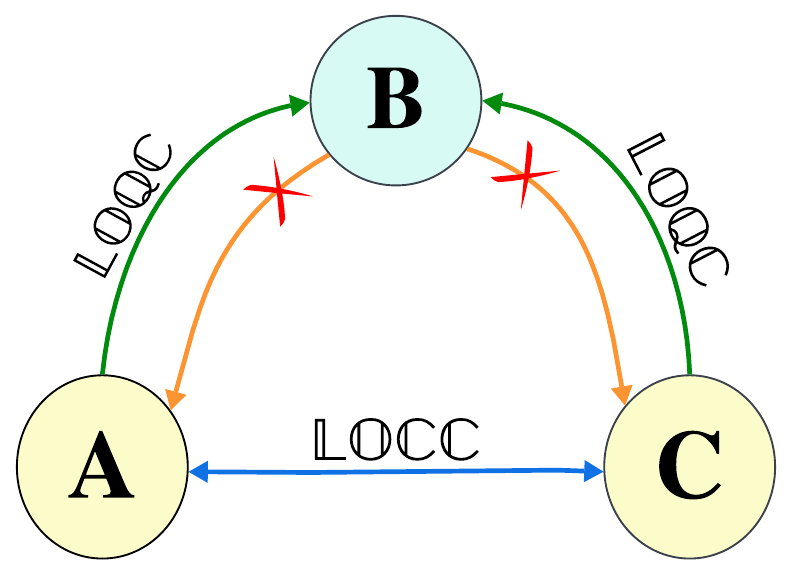}
\caption{The figure depicts the free operations under which the sQNM is monotonically non increasing. Here, we assume Alice, Bob, and Charlie possess the systems $\mathbf{A}$, $\mathbf{B}$, and $\mathbf{C}$, respectively, and Bob is used as the conditioning system. The arrows represent the direction of communication while implementing the operations. The free operations are local operation and classical communication between Alice and Charlie $\mathbb{LOCC}_{\mathbf{A} \leftrightarrow \mathbf{C}}$, involving both-way classical communications; and local operation and quantum communication $\mathbb{LOQC}_{\mathbf{A} \to \mathbf{B}}$ ($\mathbb{LOQC}_{\mathbf{C} \to \mathbf{B}}$) involving one-way quantum communication from Alice (Charlie) to Bob, while no classical communication is allowed from Bob to Alice (Charlie). See Lemma~\ref{prop:FreeOps} for more details.}
\label{fig:FreeOps}
\end{figure}
For notational convenience, whenever there is no ambiguity, we denote Alice, Bob, and Charlie's laboratories with $\mathbf{A}$, $\mathbf{B}$, and $\mathbf{C}$, respectively.
\begin{lemma}[Free operations]\label{prop:FreeOps}
The measure $N_{\mathrm{sq}}(A;C|B)_\rho$ of any state $\rho_{ABC}$ is a monotone, i.e., non-increasing, under the following operations (see Fig.~\ref{fig:FreeOps}): 

\noindent    
(O1) $\mathbb{LOCC}_{\mathbf{A}\leftrightarrow \mathbf{C}}$~\textemdash~Local quantum operations and classical communication (LOCC) channels~\cite{CLM+14} between Alice and Charlie.

\noindent
(O2) $\mathbb{LOQC}_{\mathbf{A}\to \mathbf{B}}$ and $\mathbb{LOQC}_{\mathbf{C}\to \mathbf{B}}$~\textemdash~Local quantum operation and (one-way) quantum communication (LOQC)~\cite{Datta_2016} channel, involving quantum communication from Alice to Bob and from Charlie to Bob respectively. It includes local quantum operations (channels) $\mathbb{LO}_{\mathbf{B}}$ on Bob.
\end{lemma}

The lemma is proved in Appendix~\ref{App:Prop3}. The sQNM is a monotone under local quantum channels on Alice and Charlie ($\mathbb{LO}_{\mathbf{A}}$ and $\mathbb{LO}_{\mathbf{C}}$). Moreover, it is a strong monotone under selective operations (or quantum instruments) on Alice and Charlie. For a set of instruments $\{ \mathcal{E}_{m,A} \}_m$, where each $\mathcal{E}_{m,A}$ is a completely positive map and their sum is a trace-preserving map $\mathcal{E}_{A}$, applied on a state $\rho_{ABC}$ with $\sigma^m_{ABC}=\mathcal{E}_{m,A} \otimes \id_{BC}(\rho_{ABC})/p_m$ and $p_m=\Tr[\mathcal{E}_{m,A} \otimes \id_{BC}(\rho_{ABC})]$, the strong monotonicity implies 
\begin{align}
N_{\mathrm{sq}}(A;C|B)_{\rho} \geq \sum_m \ p_m \  N_{\mathrm{sq}}(A;C|B)_{\sigma^m}.   
\end{align}
For the resultant state $\sigma_{ABC}=\mathcal{E}_A \otimes \id_{BC}(\rho_{ABC})$, we have $\sum_m \ p_m \  N_{\mathrm{sq}}(A;C|B)_{\sigma^m} \geq N_{\mathrm{sq}}(A;C|B)_\sigma$ as the consequence of convexity of sQNM. Hence, we find the monotonicity relation 
\begin{align}
N_{\mathrm{sq}}(A;C|B)_{\rho} \geq N_{\mathrm{sq}}(A;C|B)_\sigma.    
\end{align}

A similar relation also holds for selective operations applied to Charlie. Together with convexity, the strong monotonicity of operations on Alice and Charlie renders that the sQNM is monotonically non-increasing under $\mathbb{LOCC}_{\mathbf{A}\leftrightarrow \mathbf{C}}$, similar to what we have seen for squashed entanglement~\cite{Vidal00, Christandl04}; the sQNM $N_{sq}(A;C|B)_{\rho}$ could be deemed as a measure of the conditional entanglement in $\rho_{ABC}$ quantifying entanglement between $A$ and $C$ when conditioned on $B$. The $\mathbb{LOCC}_{\mathbf{A}\leftrightarrow \mathbf{C}}$ may involve implementing one-way or multiple rounds of classical communication between Alice and Charlie. 

The sQNM, $N_{\mathrm{sq}}(A;C|B)_\rho$ for a state $\rho_{ABC}$, remains invariant under a local unitary (or isometry) operation by Bob. In general, it is monotonically non-increasing under local quantum channel $\mathbb{LO}_{\mathbf{B}}$ on Bob's system. Any non-unitary operation introduces noise in Bob's system, which should not increase a genuine quantum non-Markovianity as expected. It is, however, not a strong monotone under selective operation (or instruments) on Bob. As a consequence, we see that the local operations implemented via one-way classical communication from Alice to Bob (and from Charlie to Bob), i.e., $\mathbb{LOCC}_{\mathbf{A} \to \mathbf{B}}$ (and $\mathbb{LOCC}_{\mathbf{C} \to \mathbf{B}}$), are the operations that do not increase sQNM. Any operation that involves local operations and classical communications from Bob to Alice or Charlie may increase the sQNM. For instance, with two initial Bell states shared between Alice and Bob and Bob and Charlie, local Bell measurements by Bob and classical communication to Charlie followed by a local unitary operation by Charlie can implement entanglement swapping, leading to the maximally entangled state between Alice and Charlie. Thus, $\mathbb{LOCC}_{\mathbf{B} \to \mathbf{A}}$ and $\mathbb{LOCC}_{\mathbf{B} \to \mathbf{C}}$ are not allowed, and this is a consequence of the fact that sQNM is not a strong monotone under local operations on Bob.

It can be seen from the additivity property of sQNM that attaching an additional bipartite system in an arbitrary state and sharing it with Alice and Bob (or with Charlie and Bob) does not increase sQNM. These states could be maximally entangled states, and, utilizing these, the operations $\mathbb{LOCC}_{\mathbf{A} \to \mathbf{B}}$ (and $\mathbb{LOCC}_{\mathbf{C} \to \mathbf{B}}$) enable perfect quantum communications $\mathbb{QC}_{\mathbf{A}\to \mathbf{B}}$ (and $\mathbb{QC}_{\mathbf{C}\to \mathbf{B}}$) from Alice (and Charlie) to Bob without increasing sQNM. This can also seen from the weak chain rule of sQNM. That is, for any four partite state $\rho_{AMBC}$,
\begin{align*}
N_{\mathrm{sq}}(AM;C|B)_\rho &\geq N_{\mathrm{sq}}(A;C|BM)_\rho + N_{\mathrm{sq}}(M;C|B)_\rho \\ 
&\geq N_{\mathrm{sq}}(A;C|BM)_\rho,  
\end{align*} 
the first inequality directly follows from the chain rule of QCMI and the definition of sQNM, and the second inequality is because $N_{\mathrm{sq}}(M;C|B)_\rho \geq 0$. The strong monotonicity under selective operations by Alice and Charlie and monotonicity under $\mathbb{QC}_{\mathbf{A}\to \mathbf{B}}$ and $\mathbb{QC}_{\mathbf{C}\to \mathbf{B}}$ together imply that the sQNM is monotonically non-increasing under (one-way) $\mathbb{LOQC}_{\mathbf{A}\to \mathbf{B}}$ and $\mathbb{LOQC}_{\mathbf{C}\to \mathbf{B}}$.  

Note that the set of free operations, denoted by $\mathcal{F}_O^E$, characterized above form a convex set. The set of free operations ($\mathcal{F}_O^{W_1}$) considered in~\cite{Wakakuwa2017} is a subset of the free operations ($\mathcal{F}_O^{W_2}$) considered in~\cite{Wakakuwa21}. As expected, both these sets are subsets of the set of free operations we characterize, and $\mathcal{F}_O^{W_1} \cup \mathcal{F}_O^{W_2} \subsetneq \mathcal{F}_O^E$. It is worth mentioning that intrinsic information~\cite{Christandl_2007,Schneider_2002}, a correlation in tripartite systems, is different from sQNM. For instance, the former is not convex and is not monotonically non-increasing under $\mathbb{LOCC}_{\mathbf{A}\leftrightarrow \mathbf{C}}$.

With the measure of sQNM and the free states and operations, we have the necessary ingredients to formulate a resource-theoretic framework to characterize genuine quantum non-Markovianity as a resource. One aim of this framework would be to study transformations between arbitrary states when restricted to free operations. This we shall consider in Section~\ref{sec:StateTrans}. Before that, we provide other measures of genuine quantum non-Markovianity based on the generalized divergence and the recoverability of quantum states. We also discuss the applicability of some of the measures of genuine QNM for infinite-dimensional systems. 

\subsection{Generalized squashed quantum non-Markovianity and fidelity of recovery}
We now introduce other measures to quantify squashed quantum non-Markovianity for arbitrary quantum systems applicable to both discrete and continuous variable quantum systems. Different distinguishability quantifiers have different operational meanings in terms of discrimination tasks, so it is pertinent to have measures based on different distinguishability measures. Depending on physical or analytical requirements, one can choose from appropriate quantifiers or measures accordingly.

A generalized divergence $\mathbf{D}(\rho\Vert\sigma)$ between states $\rho,\sigma$ is a monotone functional for all states under the action of quantum channels~\cite{SW13}, i.e., $\mathbf{D}(\rho_{A'}\Vert\sigma_{A'})\geq \mathbf{D}(\mathcal{N}_{A'\to A}(\rho_{A'})\Vert\mathcal{N}_{A'\to A}(\sigma_{A'}))$, where $\mathcal{N}$ is an arbitrary quantum channel and $\rho,\sigma$ are arbitrary states. Consequently, generalized divergences remain invariant under the action of isometric operations. Some examples of generalized divergences are the trace distance, sandwiched R\'enyi relative entropy~\cite{MDSFT13,WWY14}, hypothesis testing relative entropy~\cite{BD10,WR12}, relative entropy~\cite{Ume62}, negative of fidelity, etc. Based on the generalized divergence of recovery, in the same spirit as the fidelity of recovery~\cite{SW15} (see also ~\cite{BT16}), we now define generalized divergence of squashed quantum Markovianity and discuss specific properties with an example of the fidelity of squashed quantum non-Markovianity.
\begin{definition}\label{def:genQNM}
Generalized divergence $\mathbf{D}_{\mathrm{sqnm}}$ of squashed quantum non-Markovianity of an arbitrary state $\rho_{ABC}$ is defined as
\begin{align}\label{eq:genQNM}
&\mathbf{D}_{\mathrm{sqnm}}(A;C|B)_{\rho} \coloneqq  \nonumber\\
&\inf_{\rho_{ABCE}}\inf_{\mathcal{R}_{BE\to BCE}}
\mathbf{D}(\rho_{ABCE}\Vert\mathcal{R}_{BE\to BCE}(\rho_{ABE})),
\end{align}
where $\rho_{ABCE}$ is an arbitrary state extension of $\rho_{ABC}$ and $\mathcal{R}_{BE\to BCE}$ is a universal recovery map. Generally, this definition holds for arbitrary dimensional quantum systems, finite or infinite.
\end{definition}

We now consider an example of $\mathbf{D}_{\mathrm{sqnm}}$ which is obtained by taking generalized divergence $\mathbf{D}(\cdot\Vert\cdot)$ to be the negative fidelity $-F(\cdot, \cdot)$. Fidelity between arbitrary quantum states is non-decreasing under the action of quantum channels.
\begin{definition}\label{def:fidrec}
The fidelity of squashed quantum non-Markovianity $F_{\mathrm{sqnm}}$ for an arbitrary state $\rho_{ABC}$ is defined as
\begin{align}
     &   F_{\mathrm{sqnm}}(A;C|B)_{\rho} \coloneqq \nonumber\\
     &  \sup_{\rho_{ABCE}}\sup_{\mathcal{R}_{BE\to BCE}}
       F(\rho_{ABCE} \Vert \mathcal{R}_{BE\to BCE}(\rho_{ABE})),
\end{align}
where $\rho_{ABCE}$ is an arbitrary state extension of $\rho_{ABC}$ and $\mathcal{R}_{BE\to BCE}$ is a universal recovery map.
\end{definition}

As $0\leq F(\rho,\sigma)\leq 1$ between any two states defined on the same Hilbert space, $0\leq F_{\mathrm{sqnm}}(A;C|B)_{\rho}\leq 1$ holds for an arbitrary state $\rho_{ABC}$. We note that $ F_{\mathrm{sqnm}}(A;C|B)_{\rho}$ is the supremum of the fidelity of recovery $F(A;C|BE)_{\rho}$ for all possible state extensions $\rho_{ABCE}$, with recovery map being a universal recovery map. It holds that $ F_{\mathrm{sqnm}}(A;C|B)_{\rho} \leq F^{\mathrm{sq}}(A;C)_{\rho}$, where $F^{\mathrm{sq}}(A;C)_{\rho}$ is defined in~\cite[Eq.~(23)]{SW15} with recovery map being a universal recovery map.

\begin{remark}\label{rem:gendivfaith}
Consider the generalized divergences $\mathbf{D}$ that are faithful, meaning that, the generalized divergence $\mathbf{D}(\rho\Vert \sigma)$ between states $\rho$ and $\sigma$ is vanishing, $\mathbf{D}(\rho\Vert \sigma)=0$, if and only the states are equal, $\rho=\sigma$. Then, the generalized divergence of squashed quantum non-Markovianity (Definition~\ref{def:genQNM}) is a faithful measure of genuine quantum non-Markovianity in states whenever the associated generalized divergence is faithful. That is, for any faithful generalized divergence $\mathbf{D}$ in Eq.~\eqref{eq:genQNM}, $\mathbf{D}_{\mathrm{sqnm}}(A;C|B)_{\rho}=0$ for a state $\rho_{ABC}$ if and only if there exists a universal recovery map $\mathcal{R}_{BE\to BCE}$ such that $\rho_{ABCE}=\mathcal{R}_{BE\to BCE}(\rho_{ABE})$ and $\Tr_{E}(\rho_{ABCE})=\rho_{ABC}$. In this case, $\mathbf{D}_{\mathrm{sqnm}}(A;C|B)_{\rho}=0$ if and only if $N_{\mathrm{sq}}(A;C|B)_{\rho}=0$. For an instance, $-\log_2F_{\mathrm{sqnm}}(A;C|B)_{\rho}=0$ if and only if $N_{\mathrm{sq}}(A;C|B)_{\rho}=0$, as $-\log_2 F(\cdot,\cdot)$ is a faithful generalized divergence.
\end{remark}

We note that the definitions for the quantifiers of squashed quantum non-Markovianity (Definitions~\ref{defi:QNM} and~\ref{def:genQNM}) also apply to continuous-variable quantum systems as long as the associated measures (appearing on the right-hand sides of Eqs.~\eqref{eq:QNMMeas} and \eqref{eq:genQNM}) are well-defined. For instance, the fidelity and the trace distance are well-defined and always finite for both finite and infinite-dimensional quantum systems. On the other hand, QCMI is always finite for finite-dimensional systems and is defined also for infinite-dimensional systems with some subtleties involved (cf.~\cite{Shirokov2017}). For the sake of simplicity, we have limited our discussions in the other parts of this paper to finite-dimensional systems.

\section{Quantum non-Markovianity and state transformations \label{sec:StateTrans}}
We now discuss the state transformations under the free operations and find out the conditions in terms of the squashed quantum non-Markovianity. As discussed earlier, all states with vanishing sQNM form a convex set. All these states are free states. On the other hand, the operations mentioned in Lemma~\ref{prop:FreeOps} are free operations as these operations neither create nor increase sQNM.

As expected, all free states can be created using free operations. As we see below, free operations such as $\mathbb{LOCC}_{\mathbf{A} \leftrightarrow \mathbf{C}}$, $\mathbb{QC}_{\mathbf{A} \to \mathbf{B}}$, $\mathbb{QC}_{\mathbf{C} \to \mathbf{B}}$, and partial tracing on Bob's system are sufficient for this. For instance, consider any free state $\sigma_{ABC}$ with $N_{\mathrm{sq}}(A;C|B)_\sigma =0$. Further, consider the extension $\sigma_{ABCE}$ for which $I(A;C|BE)_\sigma=0$, and the state assumes the structure $\sigma_{ABCE} = \oplus_i p_i \ \sigma_{A{(BE)}^L_i} \otimes \sigma_{{(BE)}^R_iC}$, up to a local isometry by Bob. This is a separable state with the partition $A(BE)^L|(BE)^RC$. Say, initially, Alice and Charlie perform $\mathbb{LOCC}_{\mathbf{A} \leftrightarrow \mathbf{B}}$ between them and create a state $\oplus_i p_i \ \sigma_{A{(BE)}^L_i} \otimes \sigma_{{(BE)}^R_iC}$, where Alice and Charlie possess the systems $A(BE)^L$ and $(BE)^RC$ respectively. After that, both Alice and Charlie transfer the states corresponding to $(BE)^L$ and $(BE)^R$ to Bob via quantum communications $\mathbb{QC}_{\mathbf{A} \to \mathbf{B}}$ and $\mathbb{QC}_{\mathbf{C} \to \mathbf{B}}$ respectively followed by tracing out of $E^LE^R$ by Bob -- resulting in the desired state $\sigma_{ABC}$. This implies that the preparation of free states does not require any resources.   

A resource state, i.e., a state with genuine QNM, cannot be created by the action of free operations on the free states. However, a resource state can be generated by the action of resource (non-free) operations on the free states. For a tripartite system $ABC$, the maximal sQNM state is a maximally entangled state between Alice and Charlie $\phi_{AC}^+=\frac{1}{d}\sum_{n,n'}\ket{nn}\bra{n'n'}$ of Schmidt-rank $d=\min\{|A|,|C|\}$ and $N_{\mathrm{sq}}(A;C|B)_{\phi^+}=\log_2 d$. With the help of states with maximal sQNM and free operations, we provide examples of preparing arbitrary state $\rho_{ABC}$ shared between Alice, Bob, and Charlie, with $d=\min\{\abs{A}, \abs{C}\}$. 
Before the preparation starts, Alice and Charlie share a maximally entangled state $\phi^+_{A^\prime C^\prime}$. Without loss of generality, we assume $ d = \abs{C}$. The protocol for the preparation involves the following steps: (a) Alice locally prepares the state $\rho_{ABC}$; (b) Alice sends the local state of $B$ to Bob using noise-less quantum communication $\mathbb{QC}_{\mathbf{A} \to \mathbf{B}}$; (c) finally, Alice transfers the state of $C$ to Charlie via quantum teleportation exploiting the shared entangled state $\phi^+_{A^\prime C^\prime}$ and $\mathbb{LOCC}_{A \rightarrow C}$.

Thus, an arbitrary tripartite quantum state can be created using free operations upon consuming a state with maximal sQNM. However, arbitrary tripartite operations cannot generally be simulated by exploiting the states with maximal sQNM and free operations. We discuss one such example later in Section~\ref{sec:QComCost}.

\subsection{Transformation between arbitrary states}
A significant problem in any quantum resource theory is to investigate the possibility of transforming one resource state into another using free operations defined by the theory. Here, we explore the criteria for conversion between quantum states in the context of the resource theory of quantum non-Markovianity. Particularly,  we are exploring asymptotic convertibility, i.e., a scenario where we have large copies ($n$) of the initial state $(\rho_1)_{ABC}$, and we want to convert to $m_n$ copies of the final state $(\rho_2)_{ABC}$ using a free operation $\Lambda_{\rm{free}}$, such that
\begin{align}
    \frac{1}{2}\norm{\Lambda_{\rm{free}} (\rho_1^{\otimes n}) - \rho_2^{\otimes m_n}}_1 \le \varepsilon.
\end{align}
In problems of resource interconversion, a quintessential quantity of interest is the rate of transformation, $R{=}\lim_{\varepsilon \to 0} \lim_{n\to \infty}(m_n/n)$. We present two important upper bounds of the rate $R$: (1) the \emph{converse bound}, $R_{C}$, which gives an upper bound allowed by the resource theory, and (2) the \emph{achievable bound}, $R_{A}$, that we could achieve using a specific quantum protocol. Our protocol is generally not optimal, and the inequality $R_A{\le}R_{C}$ holds. We will discuss the question of optimality later shortly. In the following, we present both bounds and defer the proof to Appendix~\ref{App:State_transformation}. Let us first present the converse bound.

\begin{theorem}[Converse bound]
The maximum rate $R$, at which a state $(\rho_1)_{ABC}$ can be asymptotically converted to the state $(\rho_2)_{ABC}$, respects the upper bound
\begin{align}
R \le R_{C} = \frac{N_{\mathrm{sq}}(A;C|B)_{\rho_1}}{N_{\mathrm{sq}}(A;C|B)_{\rho_2}},  \end{align}
where $N_{\mathrm{sq}}(A;C|B)_{\rho}$ is the squashed quantum non-Markovianity of the state $\rho_{ABC}$. \label{Lem:interconversion:necessary}
\end{theorem}
See Appendix~\ref{subsec:converse} for proof of the converse bound. Next, we present the achievable bound.
\begin{theorem}[Achievable bound]
There exists a free operation that allows us to asymptotically convert state $(\rho_1)_{ABC}$ to $(\rho_2)_{ABC}$ at a rate 
\begin{align}
R \le R_{A}=\frac{\max\Big\{I(A\rangle C)_{\rho_1},I(C\rangle A)_{\rho_1}\Big\}}{\min\Big\{S(A)_{\rho_2},S(C)_{\rho_2}\Big\}}.       \end{align}
where $I(A_1\rangle A_2)_{\rho}:=S(A_2)_{\rho}-S(A_1A_2)_{\rho}$ is the coherent information of the state between the subsystems $A_1$ and $A_2$ of $\rho$. 
\label{Lem:interconversion:sufficient}
\end{theorem}
See Appendix~\ref{subsec:achievable} for detailed proof of the achievable bound and a detailed protocol description. 

Here, we briefly describe the protocol. Bob discards his system, and Alice and Charlie perform $\mathbb{LOCC}_{\mathbf{A}\leftrightarrow \mathbf{C}}$ to generate shared Bell states at a rate $\max\Big\{I(A\rangle C)_{\rho_1},I(C\rangle A)_{\rho_1}\Big\}$. Subsequently, Alice/ Charlie locally prepares the target state $\rho_2$, sends Bob's system using $\mathbb{QC}_{\mathbf{A}\to \mathbf{B}}$ and $\mathbb{QC}_{\mathbf{C}\to \mathbf{B}}$, and teleports (or performs state splitting of) Charlie's/ Alice's systems using the previously generated Bell states.  
As a sanity check,  we can ensure $R_{A}{\le}R_{C}$ exploiting the following inequality:
\begin{corollary}
For any state $\rho_{ABC}$,
\begin{align}
  \min \Big \{S(A)_{\rho},S(C)_{\rho}\Big\} & \ge  {N_{\mathrm{sq}}(A;C|B)_{\rho}}\\
   & \ge  \max\Big\{ I(A\rangle C)_{\rho},I(C\rangle A)_{\rho} \Big\} \nonumber.
\end{align}    
\label{Cor:nec_vs_suf}
\end{corollary}
\noindent We prove the above inequalities in Appendix~\ref{subsec:Cor:nec_vs_suf}. 

Next, we discuss the optimality of our achievability protocol with a few examples:

\noindent \emph{Example 1:} If $(\rho_1)_{ABC}$ and $(\rho_2)_{ABC}$  are pure bipartite states in $AC$, i.e., of the form $(\rho_i)_{ABC}{=}\ketbra{\psi_i}{\psi_i}_{AC}{\otimes}(\gamma_i)_B$, we have $R_A{=}R_C$ and the achievability protocol reduces to interconversion between bipartite entangled states proposed by Nielsen~\cite{Nielsen_quantum_transformation} and a local state preparation by Bob. 
However, we generally have a persisting gap between the achievable and converse rates. The gap raises the question of whether a better protocol can yield a higher achievable rate. 

\noindent \emph{Example 2:} We answer this positively for interconversion between a specific class of states $\proj{\psi_1}{\to}\proj{\psi_2}$ where
\begin{align}
\ket{\psi_i}_{ABC}
    &=V_{1,\ A_1A_2 \to A}^{(i)}{\otimes}V_{2,\ B_1B_2 \to B}^{(i)}{\otimes}V_{3,\ C_1C_2 \to C}^{(i)} \nonumber \\
    &\qquad \left(\ket{\alpha_i}_{A_1C_1}\otimes\ket{\beta_i}_{A_2B_1}\otimes\ket{\gamma_i}_{C_2B_2}\right).
\end{align}
For these states, we have $N_{\mathrm{sq}}(A;C|B)_{\psi_i}{=}S(A_1)_{\alpha_i}$, i.e., all the resources for quantum non-Markovianity are concentrated in the state $\proj{\alpha_i}$. Hence to convert $\proj{\psi_1}\to\proj{\psi_2}$, Alice and Charlie can first uncouple $\proj{\alpha_1}_{A_1C_1}$ by applying local isometries $V_1^{(1)}{^{\dagger}}$ and $V_3^{(1)}{^{\dagger}}$ respectively. The rest of the state is a free state and can be discarded. Alice and Charlie now share the pure bipartite state $\proj{\alpha_1}_{A_1C_1}$, which can be asymptotically converted to $\proj{\alpha_2}_{A_1C_1}$ at a rate 
\begin{align}
R=\frac{S(A_1)_{\alpha_1}}{S(A_1)_{\alpha_2}}=\frac{N_{\mathrm{sq}}(A;C|B)_{\psi_1}}{N_{\mathrm{sq}}(A;C|B)_{\psi_2}}.   
\end{align}
Subsequently, Alice and Charlie and Bob prepare free states $\proj{\beta_2}_{A_2B_1}$ and $\proj{\gamma_2}_{C_2B_2}$ and apply the local isometries $V_1^{(2)}, V_2^{(2)}$, and $V_3^{(2)}$ to get the state $\proj{\psi_2}^{ABC}$. Note this approach of uncoupling $\proj{\alpha_1}$ first and then converting to $\proj{\alpha_2}$ is better than discarding Bob's system in  $\proj{\psi_1}_{ABC}$ and performing an entanglement distillation on Alice and Charlie's subsystems as suggested in our achievability protocol. In this case, we have a strict inequality: $N_{\mathrm{sq}}(A;C|B)_{\psi_i}>\max\{I(A\rangle C)_{\psi_i}, I(C\rangle A)_{\psi_i}\}$. This indicates, for this specific example, we can find a protocol to achieve the rate of transformation $R_{\psi_1\to \psi_2}$ satisfying $R_A{<}R_{\psi_1\to \psi_2}{=}R_C$. 

However, in general, for a tripartite state, uncoupling a bipartite state between Alice and Charlie containing the entire resource of quantum non-Markovianity of the original tripartite state is not possible using free operations, as it requires either $\mathbb{LOCC}_{\mathbf{B}\to \mathbf{C}}$ or $\mathbb{LOCC}_{\mathbf{B}\to \mathbf{A}}$~\cite{ent_combing1,ent_combing2},  which are not free operations for us. We leave the topic of how to narrow the difference between the achievable and converse bounds for the interconversion between generic tripartite states an open question.

\subsection{Quantum non-Markovianity and communication cost \label{sec:QComCost}} 
For any tripartite state of the parties Alice, Bob, and Charlie with Bob as the condition, the maximally non-Markov state is any bipartite maximally entangled state shared between Alice and Charlie. For instance, $\phi^+_{AC}$ is a maximally entangled state of Schmidt rank $d$. We assume that Bob does not possess any state. One may create state $\phi^+_{AC}$ in the following steps. First, Alice and Bob share a maximally entangled state $\phi^+_{AB}$. Note this is allowed as it does not increase the sQNM. In the second step, Bob transfers his share of state to Charlie via a perfect ($\log_2 d$)-qubit of quantum communication to Charlie. As a result, Alice and Charlie share a maximally entangled state. Note that performing quantum communication from Bob to Charlie is not free and requires a quantum communication cost. We refer to an operation involving a perfect quantum communication of ($\log_2 d$)-qubit from Bob to Charlie (or Alice) as a maximally resourceful operation that can potentially create a maximally resourceful state. We highlight here that while a maximally resourceful operation can create a maximally resourceful state by consuming a free state, the converse, i.e., implementing a maximally resourceful operation enabling perfect quantum communication from Bob to Charlie (or Alice) consuming a maximally resourceful state ($\phi^+_{AC}$) and free states, is not possible. This is contrary to what we see in the resource theory of entanglement, at least for the bipartite cases. A maximally entangling operation can yield a maximally entangled state while applied to a non-entangled state. In reverse, using a maximally entangled state and local operation and classical communications (LOCC), a free operation in the resource theory of entanglement, one can implement arbitrary bipartite operations, including maximally entangling operations. Thus, not satisfying the inter-convertibility between maximally resourceful states and operations within the framework developed for sQNM may indicate two possibilities. First, the set of operations proposed in Lemma~\ref{prop:FreeOps} as the free operations is not a complete set of operations that do not increase sQNM. Second, there is fundamental irreversibility in the inter-conversion between sQNM resources present in the states and operations. We conjecture the second possibility. With this, we now study the quantum communication cost required to implement an operation resulting in the desired state transformation.             
 
Consider Alice, Bob, and Charlie want to prepare a quantum state $\rho_{ABC}$ using free operations and quantum communication. 
Initially, Alice locally prepares the state $\rho_{AA'A''}$ in her lab, where $A'\simeq B$ and $A''\simeq C$.  She then distributes systems $A'$ and $A''$ to Bob and Charlie, respectively. Note, the latter can be done in two steps: (i) Alice sends $A'A''$ to Bob via noiseless quantum channel, i.e., the identity channel, $\id_{A'A'' \to BB'}$, where $A''{\simeq} B'{\simeq} C$; then, (ii) Bob sends $B'$ to Charlie via $\id_{B' \to C}$. 



Both steps exploit the protocol of quantum state redistribution. Step (i) can be implemented using a free operation and does not consume any resources. But, step (ii) cannot be executed using free operations, and there is a cost of quantum communication to realize it. 

\begin{lemma}
The minimum quantum communication cost $Q_c$ to prepare an arbitrary quantum state $\rho_{ABC}$ via quantum state redistribution protocol is 
\begin{align}
    Q_c \geq N_{\mathrm{sq}}(A;C|B)_\rho.
\end{align}
Here, Alice and Bob initially have access to systems $A$ and $BB'$, respectively with $B'\simeq C$. Bob uses a quantum communication cost of $Q_c$ to transfer $B'$ to Charlie $C$ via the identity channel $\id_{B'\to C}$, $B'\simeq C$. This results in Alice, Bob, and Charlie having access to systems $A$, $B$, and $C$, respectively.
\end{lemma}
\begin{proof}
The proof relies on the quantum state redistribution protocol~\cite{Yard09}. In particular, here, for an initial state $\rho_{ABB'}$, we quantify the cost of quantum communication for the distribution of $B'\simeq C$ from Bob to Charlie, i.e., $\id_{B'\to C}$, where $A$ shared with Alice remains undisturbed. For that, we consider the purified state $\ket{\psi}_{ABCR}$ of $\rho_{ABC}$, where $R$ represents a reference system required for the extension. The purification is not unique. There may be infinitely many such purifications interrelated through unitary rotations on the reference $R$. The state redistribution using this purified state must not alter the systems $A$ and $R$. 

In general, implementation of $\id_{B'\to C}$ for a state $\rho_{ABB'}$ requires a quantum communication cost of $Q_c= I(AR{;}C)_{\psi}/2$, where $R$ is the purifying extension with the corresponding pure state $\op{\psi}$. As $B'\simeq C$, in the following entropic inequalities, we will use them interchangeably for the notational convenience. 

\begin{align}
  Q_c=\frac{1}{2}I(AR{;}C)_{\psi} & \ge \frac{1}{2}I(A{;}C)_{\rho} \nonumber \\ 
& \ge \frac{1}{2} \inf_{\rho_{ABCE}}I(A{;}C|{BE})_{\rho} \nonumber\\ 
& = N_{\mathrm{sq}}(A;C|B)_{\rho},    
\end{align}
the first inequality is due to data-processing inequality, and the second inequality is due to the definition of sQNM.
Note that the second inequality saturates when $\rho_{ABC}$ is a pure state.

Above, we have considered quantum communication from Bob to Charlie. Can there be less quantum communication cost if we consider quantum channels from Alice to Charlie instead? We answer this negatively. 

Consider for an initial state $\rho_{AA'B}$, with $A' \simeq C$, the state redistribution protocol to achieve $\id_{A'\to C}$ without loss of generality. This protocol requires pre-shared entanglement cost of $I(A;C)_{\rho}/2$ and quantum communication cost of $I(C;BR)_\rho/2$ where $R$ is the purifying extension of $\rho_{ABC}$~\cite{Yard09}. As preparing shared entanglement between Alice and Charlie is not a free operation unlike the case between Bob and Charlie, we have to invest $I(A;C)_{\rho}/2$ amount of quantum communication cost as well, i.e., the total quantum communication cost is $I(A;C)_{\rho}/2{+}I(C;BR)_{\rho}/2{=}S(C)_{\rho}$. In general, $S(C)_{\rho}{\ge}N_{\mathrm{sq}}(A;C|B)_{\rho}$ (Corollary~\ref{Cor:nec_vs_suf}), with the equality is achieved when $\rho_{ABC}$ is a pure bipartite state in $AC$. In contrast, the minimum quantum communication cost for the previously discussed protocol implementing $\id_{B'\to C}$, for a state $\rho_{ABB'}$  when $\rho_{ABC}$ is pure tripartite state, in which case $Q_c{=}I(A;C)_{\rho}/2{\le}S(C)_{\rho}$. This means quantum communication from Bob to Charlie requires lesser cost. 
\end{proof} 

Above, we have considered preparing an arbitrary quantum state from scratch. Now, we delve into the transformation between quantum states and associated quantum communication costs. 

\begin{lemma}
Consider a state transformation, from $\rho_{ABC}$ to $\sigma_{A'B'C'}$, with $N_{\mathrm{sq}}(A';C'|B')_\sigma {\ge} N_{\mathrm{sq}}(A;C|B)_\rho$ via reversible local operations and perfect quantum communication. Then, the quantum communication cost $Q_c$ satisfies  
\begin{align}
Q_c \geq N_{\mathrm{sq}}(A';C'|B')_{\sigma} - N_{\mathrm{sq}}(A;C|B)_{\rho}.     
\end{align}
\end{lemma}
\begin{proof}
We first recall that the sQNM remains invariant under reversible local transformation. Let us start with an extension $\rho_{ABCE}$ of $\rho_{ABC}$ so that $2 \ N_{\mathrm{sq}}(A;C|B)_{\rho}=I(A;C|BE)_{\rho}$. Let us say the desired transformation can be executed in the steps involving: (i) application of an isometry $V_B: B \to B'B_1B_2$ leading to $\rho_{ABCE} \to \rho_{AB'B_1B_2CE}$; (ii) perfect quantum communication of $B_1$ and $B_2$ from Bob to Alice and Charlie, via channel $\id_{B_1\to A_1}$ and $\id_{B_2\to C_1}$  respectively; and (iii) application of isometries $V_A: AA_1 \to A'$ and $V_C: CC_1 \to C'$ resulting in $\rho_{AA_1B'CC_1E} \to \sigma_{A'B'C'E}$. Following the chain rule of QCMI,
\begin{align}
  I(A;C|BE)_{\rho} &= I(A;C|B'B_1B_2E)_{\rho}\nonumber \\
                         &= I(AA_1;CC_1|B'E)_{\rho} - 2 Q_c,
\end{align}
where $Q_c= (I(A_1;CC_1|B'E)_{\rho} + I(A;C_1|B'E)_{\rho})/2$ is the quantum communication cost for the state redistribution $A|B'B_1B_2E|C \to AA_1|B'B_2E|C \to AA_1|B'E|CC_1$.
The notation A|B|C signifies a partitioned arrangement where the first segment belongs to Alice's lab, the second to Bob, and the third to Charlie. Now with $I(AA_1;CC_1|B'E)_{\rho} = I(A';C'|B'E)_{\sigma} \geq 2 \ N_{\mathrm{sq}}(A';C'|B')_{\sigma}$, we have $Q_c \geq N_{\mathrm{sq}}(A';C'|B')_{\sigma} - N_{\mathrm{sq}}(A;C|B)_{\rho}$. Note, for pure states $\rho_{ABC}$ and $\sigma_{A'B'C'}$, the equality is attained.
\end{proof}

We now consider a transformation to more general cases. Consider an initial state $\rho_{ABC}$ is transformed to $\sigma_{ABC}$ by utilizing free operations and quantum communication from Bob to Alice and Charlie. Then, the cost of quantum communication is given by the following theorem.
\begin{lemma} \label{Thm:necessary_Q_cost}
If a state $\rho_{ABC}$ is transformed to a final state $\sigma_{ABC}$, with $N_{\mathrm{sq}}(A;C|B)_\sigma {\ge} N_{\mathrm{sq}}(A;C|B)_\rho$,  by exploiting free operations and quantum communications from Bob to Alice and Charlie, then the necessary amount of quantum communication $Q_c$ required for the transformation satisfies
\begin{align}
Q_c \geq N_{\mathrm{sq}}(A;C|B)_\sigma - N_{\mathrm{sq}}(A;C|B)_\rho.      
\end{align}

\end{lemma}

\begin{proof}    
Note we can restrict quantum communication from Bob to Charlie without loss of generality. Now, from the converse bound in Lemma~\ref{Lem:interconversion:necessary}, it is clear that when we want to convert $\rho^{\otimes n} {\to} \sigma^{\otimes n}$, the necessary condition is $N_{\mathrm{sq}}(A;C|B)_{\rho}{\ge}N_{\mathrm{sq}}(A;C|B)_{\sigma}$.  The scenario where the condition is not met, i.e., $N_{\mathrm{sq}}(A;C|B)_{\rho}{<}N_{\mathrm{sq}}(A;C|B)_{\sigma}$, Bob needs to perform quantum communication $\mathbb{QC}_{\mathbf{B}{\to} \mathbf{C}}$ to increase the sQNM and make the transformation possible. The optimal one qubit quantum communication $\mathbb{QC}_{\mathbf{B}{\to} \mathbf{C}}$ increases the overall sQNM by at most one qubit. The protocol is similar to the examples in Section~\ref{sec:QComCost}. Consider Alice and Bob initially share a Bell state, which is a free state. Now Bob transfers his share to Charlie via a one-qubit channel. This introduces a maximal resourceful state between Alice and Charlie and increases the overall sQNM by one qubit. Let us denote this operation by $\Phi_1$, where the subscript denotes the number of quantum communications and 
\begin{align}
\Phi_1(\rho_{ABC})=\proj{\phi^{+}}_{A'C'}\otimes \rho_{ABC},
\end{align}
where the state $\proj{\phi^{+}}_{A'C'}$ is a Bell state shared between Alice and Charlie. Similarly, when Bob utilises $k$-qubit quantum channel to increase $k$-bit sQNM by successive application of $\Phi_1$, we denote the overall operation as $\Phi_k$ where 
\begin{align}
 \Phi_k(\rho_{ABC})&=\smash[b]{ \underbrace{\Phi_1{\circ}\Phi_1{\circ}\cdots \circ \Phi_1\,}_\text{$k$ times}}(\rho_{ABC})\nonumber\\ \nonumber\\
 &=\proj{\phi^{+}}_{A'C'}^{\otimes k}\otimes \rho_{ABC}.
\end{align}
Clearly, we have 
\begin{align}
N_{\mathrm{sq}}(A;C|B)_{\Phi_k(\rho^{\otimes n})}=k + n N_{\mathrm{sq}}(A;C|B)_{\rho}. \label{Eq:ebit_equality}
\end{align}
Now consider Bob has access to $nQ_c$ quantum channels to Charlie. Is there a better protocol to increase sQNM than Bob sending all $nQ_c$ ebits simultaneously, i.e., applying $\Phi_{nQ_C}$? To answer this, we consider a general operation $\Omega$ involving total $nQ_C$ quantum communication where each operation $\Phi_1$ is applied after a free operation: 
\begin{align}
&\Omega = \Omega_{(nQ_C)}{\circ}\cdots \circ \Omega_{(2)}\circ \Omega_{(1)},  
\end{align} 
where $\Omega_{(i)}= \Phi_1{\circ}\Lambda_{\rm{free}}^{(i)}$ represents one step of quantum communication. As the operation $\Phi_1$ is being applied on a Bell state independent of $\rho^{\otimes n}$, we can apply the monotonicity of sQNM: 
\begin{align}
    N_{\mathrm{sq}}(A;C|B)_{\Phi_{nQ_c}(\rho^{\otimes n})} \ge N_{\mathrm{sq}}(A;C|B)_{\Omega(\rho^{\otimes n})}.
\end{align}
This suggests $\Phi_{nQ_c}$ is the operation causing the maximum increase in sQNM for a given number of quantum communications. Now invoking the converse bound in Lemma~\ref{Lem:interconversion:necessary}, the necessary condition for $nQ_c$ quantum channels to allow the conversion from $\rho^{\otimes n}$ to $\sigma^{\otimes n}$ is   
\begin{align}
    nQ_c + n N_{\mathrm{sq}}(A;C|B)_{\rho} &= N_{\mathrm{sq}}(A;C|B)_{\Phi_{nQ_c}(\rho^{\otimes n})}\nonumber \\
    &\ge n N_{\mathrm{sq}}(A;C|B)_{\sigma}.
\end{align}
The first equality is due to Eq.~\eqref{Eq:ebit_equality}. Rearranging the above equation and dividing by $n$, we have 
\begin{align}
    Q_c \ge N_{\mathrm{sq}}(A;C|B)_{\sigma} - N_{\mathrm{sq}}(A;C|B)_{\rho}.
\end{align}
This completes the proof.
\end{proof}

So far, we have only provided a bound for the necessary quantum communication. In other words, satisfying the bound in Theorem~\ref{Thm:necessary_Q_cost} does not guarantee the relevant state transformation. Hence, it is interesting to have a bound for quantum communication, ensuring a successful state transformation. Below,  we provide a similar bound motivated from our achievability protocol in Lemma~\ref{Lem:interconversion:sufficient}.

\begin{lemma} \label{Thm:sufficient_q_cost}
If a state $\rho_{ABC}$ is transformed to a final state $\sigma_{ABC}$, with $N_{\mathrm{sq}}(A;C|B)_\sigma {\ge} N_{\mathrm{sq}}(A;C|B)_\rho$, by exploiting free operations and quantum communications from Bob to Alice and Charlie, then it is sufficient to have a quantum communication rate of $Q_c$ satisfying
\begin{align}
Q_c \geq & \min\{S(A)_{\sigma},S(C)_{\sigma}\}\nonumber\\ &\qquad -\max\{I(A\rangle C)_\rho,I(B\rangle C)_\rho,0\}.     \end{align}
\end{lemma}
\begin{proof}
The relation $N_{\mathrm{sq}}(A;C|B)_\sigma {\ge} N_{\mathrm{sq}}(A;C|B)_\rho$, due to Corollary~\ref{Cor:nec_vs_suf},  implies $\min\{S(A)_{\sigma},S(C)_{\sigma}\}{\ge}\max\{I(A\rangle C)_\rho,I(B\rangle C)_\rho\}$.   The theorem is motivated by the achievable bound and the corresponding protocol in Lemma~\ref{Lem:interconversion:sufficient}. Consider the task is to convert $\rho^{\otimes n}$ to $\sigma^{\otimes n}$ for a large $n$, then firstly Alice and Charlie distill ${\approx} n\max\{I(A\rangle C)_\rho,I(B\rangle C)_\rho\}$ Bell pairs from $\rho^{\otimes n}$. This distillation is only successful when the relevant coherent information is non-zero. In the scenario where $\max\{I(A\rangle C)_\rho,I(B\rangle C)_\rho\}\le 0$, the parties prepare the state $\sigma^{\otimes n}$ from scratch. Now, to prepare $\sigma^{\otimes n}$, Alice and Charlie need to establish ${\approx} n\min\{S(A)_{\sigma},S(C)_{\sigma}\}$ Bell-pairs. In the scenario when $\min\{S(A)_{\sigma},S(C)_{\sigma}\}\ge  \max\{I(A\rangle C)_\rho,I(B\rangle C)_\rho\}$, the desired state transformation is not possible with free operation as  Alice and Charlie still need  $n\bigg(\min\{S(A)_{\sigma},S(C)_{\sigma}\}{-}  \max\{I(A\rangle C)_\rho,I(B\rangle C)_\rho, 0\}\bigg)$ amount of extra Bell pairs. However, the quantum communication from Bob to Charlie/ Alice can establish these extra Bell Pairs. Hence, whenever the rate of quantum communication satisfies 
    \begin{align}
        &Q_C\ge \min\{S(A)_{\sigma},S(C)_{\sigma}\} \nonumber \\ &\hspace{2cm}{-}  \max\{I(A\rangle C)_\rho,I(B\rangle C)_\rho,0\},
    \end{align}
    a successful conversion from $\rho^{\otimes n}$ to $\sigma^{\otimes n}$ is guaranteed. 
\end{proof}

To summarise,  we have provided rigorous characterizations of resourceful operations to perform transformation from a less resourceful state to a more resourceful state. Without the loss of generality, two particular resourceful operations are quantum communications from Bob to Charlie and quantum communications from Alice to Charlie.  

Firstly, we have shown that to increase the same amount of sQNM, the former resourceful operation generally consumes fewer quantum channels than the latter. Next, motivated by the converse bound in Lemma~\ref{Lem:interconversion:necessary}, in Theorem~\ref{Thm:necessary_Q_cost}, we have shown that the quantum communication cost must be greater than the extra resource possessed by the final state, given by the difference in the sQNM between the final and the initial state.  Finally, invoking the achievability bound in Lemma~\ref{Lem:interconversion:sufficient}, in Theorem~\ref{Thm:sufficient_q_cost}, we have provided a protocol and a lower bound on quantum communication cost from Bob to Alice/ Charlie that guarantees a  transformation of a less resourceful state into a more resourceful state. 

\section{Operational meaning of squashed QNM \label{sec:QuantTasks}}	
Now we consider two important quantum information theoretic tasks, i.e., conditional quantum one-time pad~\cite{Conditional_otp}, and quantum state deconstruction~\cite{Berta_2018_PRL, Berta_2018}, to provide some operational meanings to sQNM. The operational interpretations are in the same spirit as the operational interpretations for squashed entanglement provided in Refs.~\cite{oppenheim2008paradigm,Conditional_otp, Berta_2018}. 

\subsection{Conditional quantum one-time pad}
We first consider the conditional one-time pad protocol introduced in~\cite{Conditional_otp}. In this protocol, Alice sends a private message to Charlie over an ideal quantum channel accessible to both Charlie and an eavesdropper, Bob.

Suppose the parties share many copies of the state $\rho_{ABC}$, with the subsystems $A$, $B$, and $C$ possessed by Alice, Bob, and Charlie, respectively. In that case, the optimal rate at which Alice can send the message to Charlie keeping it hidden from Bob is given by the conditional quantum mutual information $I(A{;}C|B)_{\rho}$~\cite{Conditional_otp}. The protocol fails when the state $\rho_{ABC}$ is a Markov state corresponding to the Markov chain ($A-B-C$). This is because, for a Markov state, the quantum conditional mutual information vanishes, $I(A{;}C|B)_{\rho}{=}0$. Operationally, such a Markov state allows Bob to access Charlie's subsystem by applying a recovery map $\mathcal{R}_{B{\rightarrow} BC}$ on her share to fully reconstruct the state $\rho_{ABC}$: 
\begin{align}
 {\id_{A}{\otimes}\mathcal{R}_{B{\rightarrow} BC}({\rho}_{AB})=\rho_{ABC}}.  
\end{align}
\noindent Here $\id$ is the identity channel and ${\rho}_{AB}=\Tr\! _B\  \rho_{ABC} $ is the state composed of Alice's system (A),  which is accessible to Bob via the ideal quantum channel, and her system (B). In contrast, Bob cannot construct a perfect recovery map for a non-Markov state to fully reconstruct the state, resulting in a non-zero private communication rate from Alice to Charlie. 

Let us now introduce a variant of the conditional quantum one-time pad protocol where we have allowed the eavesdropper (Bob) to take the help of a powerful eavesdropper, Eve, who could provide Bob with
the system that extends $\rho_{ABC}$. In this scenario, Alice has a bigger challenge: concealing her message from Bob and Eve while successfully transferring it to Charlie.

We now give the following operational interpretation of sQNM. 

\begin{theorem}
Alice wants to share a private message with Charlie in the presence of eavesdropper Bob who
may access a powerful eavesdropper, Eve. Alice is to communicate with Charlie over a noiseless
quantum channel that is also accessible to Bob. If Alice ($A$), Bob ($B$), and Charlie ($C$ ) share
state $\rho_{ABC}$ and Eve ($E$) has access to all the remaining systems, then the (asymptotic)
optimal rate $\widetilde{R}$ at which Alice can communicate with Charlie is
\begin{align}
    \widetilde{R}=2N_{\mathrm{sq}}(A;C|B)_\rho.
\end{align}
\end{theorem}

\begin{proof}
    Consider we have many copies of the state $\rho_{ABCE}$ with Eve in possession of the extension $E$ of the state $\rho_{ABC}{=}\Tr_{E}[\rho_{ABCE}]$. In this scenario, using Ref.~\cite{Conditional_otp}, the private communication rate from Alice to Charlie is $I(A{;}C|BE)_\rho$. Considering Eve to be a smart adversary who chooses the optimal extension to reduce the private communication rate, we conclude that the optimum asymptotic rate of private communication for this variant of the conditional one-time pad is 
\begin{align}
    \widetilde{R}=\inf_{\rho_{ABCE}} I(A;C|BE)_\rho=2N_{\mathrm{sq}}(A;C|B)_{\rho}.
\end{align}
This proves the theorem. 
\end{proof}
To further explain the scenario, we consider the state $\rho_{ABC}$ to be a probabilistic mixture of different Markov states, resulting in a classical non-Markov state. In such a scenario, Bob alone cannot access Alice's secret message, as $I(A{;}C|B)_{\rho}{\neq}0$. However, Bob can always collaborate with another malicious party, Eve, who has the optimum extended subsystem $E$ such that $2N_{\mathrm{sq}}(A;C|B)_{\rho}=\inf_{\rho_{ABCE}}I(A;C|BE)_\rho{=}0$. In other words, Bob and Eve can apply the global recovery map $\mathcal{R}_{BE{\rightarrow} BCE}$ to access Charlie's subsystem. In contrast, if the state $\rho_{ABC}$ is a quantum non-Markov state ($N_{\mathrm{sq}}{\neq}0$), Eve cannot jeopardize the protocol even in the presence of another adversary. This gives an operational interpretation of our measure of the squashed quantum non-Markovianity of a state.

\subsection{State deconstruction}
Let us now consider state deconstruction protocol as introduced in~\cite{Berta_2018_PRL, Berta_2018}. In a deconstruction protocol, the goal is to decouple subsystem $A$ in a density matrix $\rho_{ABC}$ from the $BC$ system. The decoupling is done over asymptotically large copies of the quantum state i.e., $\rho^{{\otimes} n}_{A^nB^nC^n}$ with a large $n$, via the action of a quantum channel $\Lambda$ acting on the quantum systems $A^nB^n$ and an auxiliary state $\alpha_{A'}$ initially uncorrelated with the rest of the systems. The channel $\Lambda$ represents uniform randomisation over unitary operations $\{U_i\}$ with $i{\in}\{1,2,{\cdots}, M_n\}$, i.e., $\Lambda(\cdot)= 1/M_n \sum_{i=1}^{M_n} U_i(\cdot)U_i^\dagger$, 
\begin{align}
    \sigma_{A^nB^nC^nA'}=\left(\id_{C^n}{\otimes}\Lambda_{A^nB^nA'}\right) \left[\rho^{\otimes n}_{A^nB^nC^n}{\otimes} \alpha_{A'}\right].
\end{align}
In the process of deconstruction of system $A$, the marginal state ${BC}$ should remain asymptotically intact, a condition known as \textit{negligible disturbance}:
\begin{align}
    F \left(\Tr_{A^n}(\rho_{ABC}^{\otimes n}),\Tr_{A^nA'}(\sigma_{A^nB^nC^nA'})\right)\geq 1-\varepsilon.
\end{align}
The deconstruction protocol also asymptotically obeys the \emph{local recoverability} condition with respect to system $B$, i.e. there is a recovery map $\mathcal{R}_{B^n\rightarrow A^nB^nA'}$ such that 
\begin{align}
    F\big(\sigma_{A^nB^nC^nA'}, 
    &\mathcal{R}_{B^n\rightarrow A^nB^nA'} 
    \left(\Tr_{A^nA'}(\sigma_{A^nB^nC^nA'})\right)\big) \nonumber \\
    &\geq 1-\varepsilon.
\end{align}

What is the minimum number of unitaries required to achieve a successful state deconstruction? To answer this, the \emph{deconstruction cost} $D(A;C|B)_{\rho}{=}\inf_\Lambda \lim_{\varepsilon \to 0} \lim_{n\rightarrow \infty}\frac{\log_2 M_n}{n}$ have been introduced~\cite{Berta_2018_PRL,Berta_2018}.It has been shown the deconstruction cost is given by QCMI of the state $\rho$ conditioned over $B$: $D(A;C|B)_{\rho}{=}I(A;C|B)_\rho$. Clearly when $\rho$ is a Markov state, $D(A;C|B)_{\rho}{=}I(A;C|B)_\rho{=}0$---only one unitary is enough. 

We now introduce a variant of the state deconstruction protocol where Bob is allowed to have access to an arbitrary extension $\rho_{ABCE}$ of $\rho_{ABC}$. In this case, the random unitary channel for the deconstruction protocol is acting on the systems $A B E$ in the asymptotic limit and $A'$ of an auxiliary state $\alpha_{A'}$. This allows us to give the following operational interpretation of state deconstruction protocol.

\begin{theorem}
    When Bob is allowed to have access to system $E$ from any possible state extension $\rho_{ABCE}$ of the state $\rho_{ABC}$, then the minimum deconstruction cost, $\widetilde{D}(A;C|B)_\rho$, is given by twice the sQNM of the state:
    \begin{align}
     \widetilde{D}(A;C|B)_\rho=\inf_{\rho_{ABCE}}D(A;C|BE)_\rho=2N_{\mathrm{sq}}(A;C|B)_\rho.   
    \end{align}
\end{theorem}

\begin{proof}
    The theorem is the simple consequence of the fact that $D(A;C|BE)_\rho=I(A;C|BE)_{\rho}$. Hence, 
    \begin{align*}
     \widetilde{D}(A;C|B)_\rho&=\inf_{\rho_{ABCE}}D(A;C|BE)_\rho \nonumber \\
     &=\inf_{\rho_{ABCE}}I(A;C|BE)_\rho = 2N_{\mathrm{sq}}(A;C|B)_\rho.
    \end{align*} 
\end{proof}
As a consequence of the above theorem, when $\rho$ is a classically non-Markov state, we can still achieve zero deconstruction cost as there exists an extension $\rho_{ABCE}$ such that $I(A;C|BE)_\rho{=}0$.

\section{Discussion and Conclusion \label{sec:Disc}}
Quantum non-Markovianity, seen in tripartite quantum systems, is a widely discussed quantum correlation.  It finds many applications in quantum information processing. It is thus pertinent to study the quantum origin of non-Markovianity in quantum non-Markovian states, i.e., states with strictly positive quantum conditional mutual information (QCMI). In this regard, we have introduced a resource-theoretic framework to characterize genuine quantum non-Markovianity as a resource. We have shown that quantum non-Markovianity in states can have both classical and quantum contributions, and the leftover non-Markovianity can quantify the genuine quantum non-Markovianity after squashing out all non-quantum contributions. We denote this as the squashed quantum non-Markovianity (sQNM), and it satisfies several desirable information-theoretic properties. With this in hand, we have characterized the free states, i.e., states with no resources or vanishing sQNM. The set of free states forms a convex set, and hence, our resource theory of genuine quantum non-Markovianity is a convex quantum resource theory~\cite{LBT19, RBTL20}. Thus, by probabilistic mixing of the free states, one cannot create a resource state, i.e., a state with genuine quantum non-Markovianity. We have also characterized the free quantum operations, i.e., physical transformations that cannot increase or create sQNM in states. Equipped with free states and free operations, we have studied the conditions for state transformations under free operations and derived various bounds on the rate of transformations. We explored the interrelation between quantum communication and sQNM in state transformations and have demonstrated that the quantum communication cost of creating an arbitrary quantum state is lower bounded by the sQNM of the state. Further, the quantum communication cost for any arbitrary state transformation is lower bounded by the change in sQNM. We have discussed a generalization of sQNM, expressed in terms of generalized divergence and fidelity of recovery, and possible implications of our results in quantum information theoretic tasks, namely conditional quantum one-time pad and quantum state deconstruction. 

Our results, particularly the resource-theoretic framework, are expected to deepen the understanding of quantum non-Markovianity and find wide-range implications in quantum information theory. For instance, the genuine quantum non-Markovianity quantified by sQNM represents a correlation that may be understood as conditional quantum entanglement, i.e., representing entanglement between two parties while conditioning on a third one. It unravels a different kind of `bipartite' entanglement in a tripartite system. Thus, studying the genuine quantum non-Markovianity, which lies between bipartite and tripartite entanglement, may help understand qualitative differences between bipartite and multipartite entanglements. Our results find that quantum communication, a non-free (or resourceful) operation, can create or increase genuine quantum non-Markovianity in states. This substantiates our general intuition that quantum non-Markovianity and quantum communication are interconnected. However, we also note that non-Markovianity as a resource present in a state and the same in an operation are not inter-convertible. For instance, one cannot simulate quantum communication by utilizing a quantum non-Markovian state and free operations, in general. This raises the question of whether there is a fundamental irreversibility in the resource transformations. The answer to this requires further studies, which we leave to future works.

Beyond states, non-Markovianity is also widely explored in the context of operations and dynamics, represented by stochastic processes involving sequences of events. Markovianity, in general, signifies the memoryless property of such stochastic processes where the present event of a system fully determines the evolution to the next event. Non-Markovianity, on the other hand, signifies that the dynamics of a system cannot completely be characterized by ignoring the event(s) it underwent in the past. In the last decades, the non-Markovian properties have been studied extensively in quantum dynamical processes~\cite{Laine10, Luo12, Rivas2014, Bus14, Dariusz14, BD16,BDW16, Das2018,Bhattacharya21, Kuroiwa21, Berk2021, bylicka2013, anand2019}. We anticipate that, as the dynamical properties are often characterized in terms of states, the resource theoretical insights gained in studying non-Markovianity in states will find important implications to acquire improved understanding of non-Markovianity in quantum processes. Furthermore, Markov states, i.e., the states represented by short Markov chains, often appear in the causal context of quantum processes~\cite{Costa2016, Pollock2018, Felix2018, Taranto2019, Milz2021, Milz2021a, Nery2021}. In particular, it is known that the common-cause quantum processes are characterized by Markov states~\cite{Allen17}. With the notion of genuine quantum non-Markovianity in states, our results are expected to introduce a new approach to these studies, particularly exploring the genuine quantum aspects in the causal structures and their possible technological applications.

An executive summary of the main results is the following.
 \begin{itemize}
     \item Squashed quantum non-Markovianity (sQNM) is introduced to measure genuine quantum non-Markovian correlations in tripartite quantum states. It is based on quantum conditional mutual information. For finite-dimensional systems, the sQNM is convex, monogamous, continuous, faithful, and super-additive. For infinite-dimensional systems, measures of genuine quantum non-Markovianity based on generalized divergence and recoverability of states via universal recovery maps are formulated. 

     \item The relation of genuine QNM with entanglement and unextendibility, two of the important correlations in quantum information theory, is explored. In particular, it is shown that the genuine QNM is monogamous. The sQNM of an arbitrary tripartite state is always lower bounded by the squashed entanglement between non-conditioning systems. Further, it is limited by the extendibility of any of the non-conditioning systems.
     
     \item A framework for a convex resource theory is introduced to study genuine quantum non-Markovianity as a resource after characterizing (resource-)free states and operations. The resourceful states have non-zero sQNM and are genuine quantum non-Markovian states. State transformations, in relation to sQNM and quantum communication, are studied within this framework.  
     
     \item Two quantum information-theoretic tasks, quantum one-time pad and state deconstruction, are re-investigated to provide an operational meaning to sQNM.
     
 \end{itemize}

\noindent \textit{Note}.--- In a recent work~\cite{BGG+24}, the squashed quantum non-Markovianity in states has found application in characterizing the open quantum system dynamics in terms of its information revival being either due to a genuine information backflow from the environment or due to a non-causal revival. 
\medskip

\begin{acknowledgments}
We thank Mark M. Wilde, Arul Lakshminarayan, and Prabha Mandyam for comments and suggestions. R.G. thanks the Council of Scientific and Industrial Research (CSIR), Government of India, for financial support through a fellowship (File No. 09/947(0233)/2019-EMR-I). K.G. is supported by the Hong
Kong Research Grant Council (RGC) through the grant
No. 17307520, John Templeton Foundation through
grant 62312, “The Quantum Information Structure
of Spacetime” (qiss.fr). S.D. acknowledges support from the Science and Engineering Research Board, Department of Science and Technology (SERB-DST), Government of India under Grant No. SRG/2023/000217 and the Ministry of Electronics and Information Technology (MeitY), Government of India, under Grant No. 4(3)/2024-ITEA. S.D. also thanks IIIT Hyderabad for the Faculty Seed Grant. M.N.B. gratefully acknowledges financial support from SERB-DST (CRG/2019/002199), Government of India.
\end{acknowledgments}

\appendix 
\section{Preliminaries}
In this section, we briefly describe the standard notations and definitions of some basic concepts used in this paper to discuss and arrive at the main results.

A quantum system $A$ is associated with a separable Hilbert space $\mathcal{H}_A$ and $|A|:=\dim(\mathcal{H}_A)$. The Hilbert space of a composite system $AB$ is given by $\mathcal{H}_A\otimes \mathcal{H}_B$. Wherever there is no confusion, we label the Hilbert space of a system with the same label of the system. A density operator $\rho_A$ represents a state of a quantum system $A$. The density operator $\rho_A$ is positive semi-definite defined on $\mathcal{H}_A$ with a unit trace, i.e., $\Tr[\rho_A]=1$. If $\rho_{AB}$ is a joint state of a composite system $AB$, then its reduced state on system $B$ is $\rho_{B}:=\Tr_{B}[\rho_{AB}]$. An arbitrary pure state $\psi_A$ is a density operator of rank-one, where if $\psi_A = \op{\psi}_A$ then $\ket{\psi}_A\in\mathcal{H}_A$. For finite-dimensional system $A$ or $B$, a maximally entangled state $\phi^+_{AB}$ between $A$ and $B$ is of Schmidt-rank $d=\min\{|A|,|B|\}$ and denoted as $\phi^+_{AB}:= \op{\phi^+}_{AB}$, where
\begin{align}
    \ket{\phi^+}_{AB}:=\frac{1}{\sqrt{d}}\sum_{i=0}^{d-1}\ket{ii}_{AB},
\end{align}
and $\{\ket{i}\}_i$ is an orthonormal basis in the $d$-dimensional Hilbert space.

 A quantum channel $\mathcal{N}_{A\to B}:\mathcal{L}(\mathcal{H}_A)\to \mathcal{L}(\mathcal{H}_B)$ is a completely positive, trace-preserving (CPTP) map that takes input in the system $A$ and outputs in the system $B$. If a quantum channel $\mathcal{N}_{A\to A}:\mathcal{L}(\mathcal{H}_A)\to \mathcal{L}(\mathcal{H}_A)$, then we simply denote it as $\mathcal{N}_{A}$. We denote the identity operator on $\mathcal{H}_A$ as $\mathds{1}_A$. Let $\id_{A\to B}$ denote the identity quantum channel with input in system $A$ and output in system $B$. 
 
 A quantum instrument is relevant in the context of a general evolution which takes an input quantum state $\rho_A$ and gives rise to both quantum and classical outputs, and is defined by a set of complete positive (CP), trace non-increasing maps $\{p_x{\Lambda}^x_{A\to B}\}$, such that $\Tr(p_x{\Lambda}^x_{A\to B}(\rho_{A})){\le}\Tr(\rho_{A})$, and the sum map $\sum_{x}p_x \Lambda_x$ is a quantum channel. Each CP map, $p_X \Lambda^x$, is associated with a classical outcome ${\proj{x}_X}$ resulting in an overall quantum channel $\mathcal{E}_{A{\to}XB}$ transforming the quantum state $\rho_{A}$ to 
\begin{align}
    \mathcal{E}_{A\to XB}(\rho_A) = \sum_x p_x  \proj{x}_X \otimes {\Lambda^x}_{A \to B}(\rho_A).
\end{align}

The fidelity $F(\rho, \sigma)$ between two quantum states $\rho$ and $\sigma$ is $F(\rho,\sigma)\coloneqq \norm{\sqrt{\rho}\sqrt{\sigma}}_1^2=\left(\Tr[\sqrt{\sqrt{\rho}\sigma\sqrt{\rho}}]\right)^2$, where $\norm{M}_1:=\Tr[\abs{M}]=\Tr[\sqrt{M^\dagger M}]$ is the trace-norm of an operator $M$. The fidelity between the two states is always between 0 and 1. $F(\rho,\sigma)=0$ if and only if $\rho$ and $\sigma$ are orthogonal and $F(\rho,\sigma)=0$ if and only if $\rho=\sigma$. The trace-distance between two states $\rho$ and $\sigma$ is given by $\norm{\rho-\sigma}_1/2$. The trace distance between any two states is always between 0 and 1. $\norm{\rho-\sigma}_1/2=1$ if and only if $\rho$ and $\sigma$ are orthogonal and $\norm{\rho-\sigma}_1/2=0$ if and only if $\rho=\sigma$. The following relation between the trace distance and fidelity between arbitrary states $\rho$ and $\sigma$ holds~\cite{FV99}:
\begin{align}
    1-\sqrt{F(\rho,\sigma)}\leq \frac{1}{2}\norm{\rho-\sigma}_1\leq \sqrt{1-F(\rho,\sigma)}.
\end{align}
The von Neumann entropy $S(A)_{\rho}$ for a quantum state $\rho_{A}$ is defined as
\begin{align}
    S(A)_{\rho} := S(\rho_A)= -\Tr(\rho_{A}\log_2(\rho_{A})).
\end{align}
For any pure pure state $\psi_A$, $S(A)_{\psi}=0$. $0\leq S(A)_{\rho}\leq \log_2 |A|$ always holds for arbitrary system $A$. When $|A|<\infty$, the maximum is attained for a maximally mixed state, i.e., $S(A)_{\rho}=\log_2|A|$ if and only if $\rho_A =\mathds{1}_A/|A|$. If a bipartite state $\psi_{AB}$ is pure then $S(A)_{\psi}=S(B)_{\psi}$.

Quantum conditional entropy $S(A|B)_{\rho}$ of an arbitrary bipartite state $\rho_{AB}$ when conditioned on $B$ is given as
\begin{align}
 S(A|B)_{\rho} :=  S(AB)_{\rho} - S(B)_{\rho}.
\end{align}
For a bipartite state $\rho_{AB}$, it quantifies the average uncertainty of quantum state $\rho_A$ when the state $\rho_B$ is known. The negative quantum conditional entropy $S(AB)_{\rho}$, i.e., $S(AB)_{\rho}<0$, implies that the state $\rho_{AB}$ is entangled, but the converse need not be true.

Quantum mutual information $I(A;B)_{\rho}$ quantifies the total correlations between $A,B$ when they are in the joint state $\rho_{AB}$,
\begin{align}
    &I(A;B)_{\rho}\nonumber \\ &\hspace{0.5cm}:= S(A)_{\rho} + S(B)_{\rho} - S(AB)_{\rho} = I(B;A)_{\rho}.
\end{align}
Quantum conditional mutual information $I(A;C|B)_{\rho}$ for a tripartite state $\rho_{ABC}$ quantifies the total correlation between $A,C$ when conditioned on $B$,
\begin{align}
    I(A;C|B)_{\rho} & := S(AB)_{\rho} + S(BC)_{\rho} - S(ABC)_{\rho} - S(B)_{{\rho}} \nonumber\\
    &= I(AB;C)_{\rho}-I(C;B)_{\rho}\nonumber\\& = I(C;A|B)_{\rho}.
\end{align}
The coherent information $I(A\rangle B)_{\rho}$ of a bipartite state $\rho_{AB}$ between subsystems $A$ and $B$ is
\begin{align}
    I(A \rangle B)_{\rho} :=S(B)_{\rho} - S(AB)_{\rho} = - S(A|B)_{\rho}.
\end{align}

\section{Proof of Lemma~\ref{lem:sQNMsEnt} \label{app:Lemma-sQNMsEnt}}
Here, we prove the Lemma~\ref{lem:sQNMsEnt} elaborating on the inter-relation between squashed quantum non-Markovianity ($N_{\mathrm{sq}}$) and squashed entanglement ($E_{\mathrm{sq}}$). The squashed entanglement for a state $\rho_{AC}$ is defined as,
\begin{align}
E_{\mathrm{sq}}(A;C)_\rho=\inf_{\rho_{ABF}} \frac{1}{2}I(A;C|F)_{\rho}\label{eq19},
\end{align}
where the infimum is carried out over all extensions of $\rho_{AB}=\Tr_F [\rho_{ABF}]$. For any state $\rho_{ABC}$ with $\rho_{AC}=\Tr_B[\rho_{ABC}]$, we can easily see that
\begin{align*}
\inf_{\rho_{ABCE}}\frac{1}{2} I(A;C|BE)_\rho = N_{\mathrm{sq}}(A;C|B)_\rho\geq E_{\mathrm{sq}}(A;C)_\rho.   
\end{align*}
Because the infimum taken to find out $N_{\mathrm{sq}}$ utilizes a more restricted set of extended states than that for $E_{\mathrm{sq}}$. From this, it is clear that $E_{\mathrm{sq}}(A;C)_\rho >0$ implies $N_{\mathrm{sq}}(A;C|B)_\rho >0$, but the reverse statement is not necessarily true. 

For any state $\rho_{ABC}$, we can easily see that
\begin{align}
E_{\mathrm{sq}}(AB;C)_\rho \geq N_{\mathrm{sq}}(A;C|B)_\rho,    
\end{align}
by using the definitions of squashed entanglement and sQNM, and $I(AB;C)_\rho \geq I(A;C|B)$, a consequence of the chain rule of QCMI. Similarly, we have $E_{\mathrm{sq}}(A;BC)_\rho \geq N_{\mathrm{sq}}(A;C|B)_\rho$. 

Consider an isometry $V: B \to B_L \otimes B_R$, leading to a transformation $\rho_{ABC} \to \sigma_{AB_LB_RC}$. Note this operation does not alter $E_{\mathrm{sq}}(AB;C)_\rho$, $E_{\mathrm{sq}}(A;BC)_\rho$, and $N_{\mathrm{sq}}(A;C|B)_\rho$ of the initial state $\rho_{ABC}$. Nevertheless, it can be easily checked that
\begin{align}
E_{\mathrm{sq}}(AB_L;B_RC)_\sigma \geq N_{\mathrm{sq}}(A;C|B_LB_R)_\sigma,    
\end{align}
following the definitions of squashed entanglement, sQNM, and applying the chain rule of QCMI. Again, $N_{\mathrm{sq}}(A;C|B_LB_R)_\sigma = N_{\mathrm{sq}}(A;C|B)_\rho$. Now, combining the equations, we have
\begin{align}
\min \left\{ E_{\mathrm{sq}}(AB_L;B_RC)_\sigma, \ E_{\mathrm{sq}}(AB;C)_\rho, \ E_{\mathrm{sq}}(A;BC)_\rho \right\} \nonumber \\
 \geq N_{\mathrm{sq}}(A;C|B)_\rho\geq E_{\mathrm{sq}}(A;C)_\rho.
\end{align}

\section{Proof of Lemma~\ref{lem:Qprop1} \label{app:lem-Qprop1}}

(P1) {\it Convexity}: For any state $\rho_{ABC}=\sum_m p_m \ \rho_{ABC}^m $,
\begin{align}
N_{\mathrm{sq}}(A;C|B)_\rho \leq \sum_m p_m \ N_{\mathrm{sq}}(A;C|B)_{\rho^m}. \label{eq:app-conv}    
\end{align}
To prove the convexity, consider an extension $\rho_{ABCE}=\sum_m p_m \ \rho_{ABCE}^m $ for which $N_{\mathrm{sq}}(A;C|B)_{\rho^m} = I(A;C|BE)_{\rho^m}$ for all $m$. By adding an ancilla $F$, we have a state
\begin{align}
\rho_{ABCEF}=\sum_m p_m \ \proj{m}_F \otimes \rho_{ABCE}^m,     
\end{align}
where $\{ \ket{m}_F \}$ are the orthonormal bases of $F$. The QCMI of the state $\rho_{ABCEF}$ satisfies
\begin{align}
 I(A;C|BEF)_\rho & = \sum_m p_m \ I(A;C|BE)_{\rho^m}, \nonumber \\
 & = \sum_m p_m \ 2 N_{\mathrm{sq}}(A;C|B)_{\rho^m}.   
\end{align}
Now with $I(A;C|BEF)_\rho \geq 2 N_{\mathrm{sq}}(A;C|B)_{\rho}$, we recover the convexity relation~\eqref{eq:app-conv}. As a corollary, the states with $N_{\mathrm{sq}}=0$, i.e., Markov states and states with CNM, form a convex set. \\

(P2) {\it Super-additive:} For every state $\rho_{AA^\prime BB^\prime CC^\prime}$, the sQNM is super-additive
\begin{align}
   &N_{\mathrm{sq}}(AA^\prime;CC^\prime|BB^\prime)_{\rho}\nonumber \\ &\hspace{1cm}\geq N_{\mathrm{sq}}(A;C|B)_{\rho} + N_{\mathrm{sq}}(A^\prime;C^\prime|B^\prime)_{\rho},
\end{align}
and additive on the tensor product, i.e.,
\begin{align}
 &N_{\mathrm{sq}}(AA^\prime;CC^\prime|BB^\prime)_{\rho}\nonumber \\
 &\hspace{1cm}=N_{\mathrm{sq}}(A;C|B)_{\rho} + N_{\mathrm{sq}}(A^\prime;C^\prime|B^\prime)_{\rho}, 
\end{align}
for $\rho_{AA^\prime BB^\prime CC^\prime} = \rho_{A B C} \otimes \rho_{A^\prime B^\prime C^\prime}$. 

We prove the super-additivity property of sQNM, from which the additivity is easily followed. Consider a state $\rho_{AA^\prime BB^\prime CC^\prime}$ and its no-signaling extension $\rho_{AA^\prime BB^\prime CC^\prime E}$ such that $2N_{\mathrm{sq}}(AA^\prime;CC^\prime|BB^\prime)_\rho = I(AA^\prime;CC^\prime|BB^\prime E)_{\rho}$. Now, using the chain rule, monotonicity of QCMI under partial tracing, and definition of $N_{\mathrm{sq}}$, we have
\begin{align}
        &2N_{\mathrm{sq}}(AA^\prime;CC^\prime|BB^\prime)_\rho  
        = I(AA^\prime;CC^\prime|BB^\prime E)_{\rho}\nonumber \\
        &\hspace{0.5cm}= I(A;CC^\prime|BB^\prime E A^\prime)_{\rho} 
        + I(A^\prime;CC^\prime|BB^\prime E)_{\rho}\nonumber  \\
        &\hspace{0.5cm}\geq I(A;C|BB^\prime E A^\prime)_{\rho} 
        + I(A^\prime;C^\prime|BB^\prime E)_{\rho}\nonumber  \\
        &\hspace{0.5cm}\geq 2N_{\mathrm{sq}}(A;C|B)_{\rho} 
        + 2N_{\mathrm{sq}}(A^\prime;C^\prime|B^\prime)_{\rho}.
\end{align}

(P3) {\it Monogamy:} For every state $\rho_{AA^\prime B C}$, the measure of quantum non-Markovianity, the quantum non-Markovianity satisfies a monogamy relation
\begin{align}
 &N_{\mathrm{sq}}(AA^\prime ; C|B)_\rho \nonumber\\
 &\hspace{1cm}\geq N_{\mathrm{sq}}(A ; C|B)_\rho + N_{\mathrm{sq}}(A^\prime ; C|B)_\rho.
\end{align}

To prove the relation, we consider an extension $\rho_{AA^\prime B E C}$  of the state for which $N_{\mathrm{sq}}(AA^\prime ; C|B)_\rho = I(AA^\prime ; C|BE)_\rho$. Now, using the chain rule of QCMI
\begin{align}
 N_{\mathrm{sq}}(AA^\prime ; C|B)_\rho & = I(AA^\prime ; C|BE)_\rho,\nonumber \\
 & = I(A ; C|BE A^\prime)_\rho + I(A^\prime ; C|BE)_\rho, \nonumber\\
 & \geq N_{\mathrm{sq}}(A ; C|B)_\rho + N_{\mathrm{sq}}(A^\prime ; C|B)_\rho.
\end{align}
This monogamy relation can be generalized to $n$-party systems possessed by Alice.

(P4) {\it Continuity}: For any two states $\rho_{ABC}$ and $\sigma_{ABC}$ with $\frac12\norm {\rho_{ABC} - \sigma_{ABC} }_1 {\le} \varepsilon$,
\begin{align}
| \  N_{\mathrm{sq}}(A;C|B)_\rho - N_{\mathrm{sq}}(A;C|B)_\sigma \ | \le \frac{ f(\varepsilon)}{2}, \label{eq:app-Lem:continuity} 
\end{align}
where $f(\varepsilon) \to 0$, for $\varepsilon \to 0$. 

Our proof is similar to the continuity bound of squashed entanglement in Ref.~\cite{Christandl04}. We first use the continuity bound for state purification. For two states, $\rho_{ABC}$ and $\sigma_{ABC}$, with a trace distance $\nicefrac{\norm{\rho{-}\sigma}_1}{2} \leq\varepsilon$, we can find respective purifications, $\ket{\psi_{ABCR}}$ and $\ket{\phi_{ABCR}}$, for which the trace distance has the upper bound 
\begin{align}
	\frac{\norm{ \ \proj{\psi}-\proj{\phi}\  } _1}{2}\leq 2\sqrt{\varepsilon}.
\end{align}
We obtain the above inequality utilizing the Uhlman theorem~\cite{uhlmann_transition_1976, jozsa_fidelity_1994} and Fuchs-Van de Graaf inequality~\cite{fuchs_cryptographic_1999}. Alternatively, one can directly use the continuity of Stinespring dilations~\cite{kretschmann_information-disturbance_2006,kretschmann_continuity_2007}. Note, without loss of generality, we have considered identical purifying systems ($R$) for both $\rho$ and $\sigma$ as we can always consider the smaller purifying system a subspace of the larger one. 

Next, we introduce the continuity bound of arbitrary extensions of $\rho$ and $\sigma$.  We obtain an arbitrary extension of a state by applying a CPTP map $\Lambda_{R\to E}$ acting on the purifying extension $R$ of the state, i.e.,
\begin{align}
	\rho_{ABCE}&= \Big[\id_{ABC}\otimes\Lambda_{R{\to}E}\Big]\left(\proj{\psi}_{ABCR}\right),  \label{55}\\
	\sigma_{ABCE}&=\Big[\id_{ABC}\otimes\Lambda_{R{\to}E}\Big]\left(\proj{\phi}_{ABCR}\right) \label{56}.
 \end{align}
Next, we use the contractivity of trace distance~\cite{P_rez_Garc_a_2006} to obtain
 \begin{align}	
 \frac{\norm{\rho_{ABCE}-\sigma_{ABCE}}_1}{2}\leq2\sqrt{\varepsilon}. \label{Eq:trace_distance_extension}
\end{align}
Having established Eq.~\eqref{Eq:trace_distance_extension}, we can now use the tight uniform continuity bound for quantum conditional mutual information (See Corollary 1 of~\cite{Shirokov2017}, also~\cite{Wilde_2016}): 
\begin{align}
    & |I(A;C|BE)_\rho - I(A;C|BE)_\sigma|  \nonumber \\
    &{\leq} 4\sqrt{\varepsilon}\log_2 \min \{\abs{A},\abs{C}\} {+}2  (1{+}2\sqrt{\varepsilon}) \ h_2\left(\frac{2\sqrt{\varepsilon}}{1{+}2\sqrt{\varepsilon}}\right)\nonumber \\
    &= f(\varepsilon). \label{Eq:QCMI_continuity}
\end{align}
Here, $h_2(.)$ is the binary Shannon entropy. Note this is a tighter bound compared to the one introduced in~\cite{Christandl04} by combining the continuity of von Neumann entropy~\cite{Alicki_2004,fannes1973continuity} and continuity of conditional von Neumann entropy.

Now, we are ready to prove the continuity bound of the measure of quantum non-Markovianity $N_{\mathrm{sq}}$. Without loss of generality, we consider $N_{\mathrm{sq}}(A;C|B)_\rho \geq N_{\mathrm{sq}}(A;C|B)_\sigma$. Let us assume the CPTP map $\widetilde{\Lambda}$ achieves the infimum value of $I(A;C|BE)$ for $\sigma$, i.e., $\inf_{\widetilde{\Lambda}} I(A;C|BE)_\sigma=I(A;C|BE)_{\widetilde{\Lambda}(\proj{\phi})}$ . Then, one can write,
\begin{align}
	&2 \ |N_{\mathrm{sq}}(A;C|B)_\rho - N_{\mathrm{sq}}(A;C|B)_\sigma|\nonumber \\
 &= \inf _{\widetilde{\Lambda}}I(A;C|BE)_{\rho} - I(A;C|BE)_{\widetilde{\Lambda}(\proj{\phi})}\nonumber \\ 
 &\leq I(A;C|BE)_{\widetilde{\Lambda}(\proj{\psi})}-I(A;C|BE)_{\widetilde{\Lambda}(\proj{\phi})}\label{EQ:Lemma1_step3}\nonumber \\
 &\leq f(\varepsilon).
\end{align} 
Here, Eq.~\eqref{EQ:Lemma1_step3} indicates the state $\widetilde{\Lambda}(\proj{\psi})$ may not achieve the infimum value of QCMI, $I(A;C|BE)_\rho$. The final inequality is due to Eq.~\eqref{Eq:QCMI_continuity}. This completes the proof.

(P5) {\it Faithfulness}: It follows from combining Eqs.~\eqref{eq:qcmi-fid}~and~\eqref{eq:QNMMeas} (see also Remark~\ref{rem:gendivfaith}) as the fidelity of recovery is vanishing if and only if the state is perfectly recoverable via a universal recovery map.

\section{Proof of Lemma~\ref{prop:FreeOps} \label{App:Prop3}}

\subsection{Monotonicity under local quantum channel on Bob ($\mathbb{LO}_\mathbf{B}$)}
\label{App:MonoOnB}
Here, we aim to demonstrate the monotonicity under local quantum channel local trace-preserving and completely positive operations (i.e., CPTP maps), also known as the local quantum channel applied to Bob's system. In previous works~\cite{Wakakuwa21, Wakakuwa2017}, the authors established that only reversible operations on Bob's system are allowable operations that do not increase non-Markovianity. However, their measure of non-Markovianity includes both classical and quantum aspects. In contrast, using the squashed quantum non-Markovianity measure $N_{\mathrm{sq}}$, we establish that any CPTP map will not increase quantum non-Markovianity. 

First, let us establish that the application of the partial tracing-out operation on the conditioned system possessed by Bob does not lead to an increase in quantum non-Markovianity. We consider two states, $\rho_{ABC}$ and $\rho_{ABB^\prime C}$, where $\rho_{ABC}=\Tr_{B'}[\rho_{ABB^\prime C}]$. Now, the quantum non-Markovianity is expressed as follows:
\begin{align}
&2N_{\mathrm{sq}}(A;C|B)_{\rho}=\inf_{\rho_{ABCE}} I(A;C|BE)_{\rho},\label{Eq:A1} \\
&2N_{\mathrm{sq}}(A;C|BB')_{\rho} =\inf_{\rho_{{ABB'CF}}} I(A;C|BB'F)_{\rho}, \label{Eq:A2}    
\end{align}
the infimums are carried out on the set of all state extensions $\rho_{ABCE}$ and $\rho_{ABB'CF}$, respectively. A noticeable distinction between Eq.~\eqref{Eq:A1} and Eq.~\eqref{Eq:A2} is that in Eq.~\eqref{Eq:A2}, due to the presence of system $B'$, the conditional mutual information is subjected to a more constrained infimum over $F$ than the one over $E$. This leads us to conclude that
\begin{align}
    N_{\mathrm{sq}}(A;C|BB')_{\rho}\geq N_{\mathrm{sq}}(A;C|B)_{\rho}. \label{eq:MonPtraceB}
\end{align}
Thus, a partial tracing-out operation on the conditioning system does not increase quantum non-Markovianity. 

Next, we prove monotonicity under the local CPTP operation applied to Bob's system $B$. Let us consider a state $\rho_{ABC}$ and its extension $\rho_{ABCE}$ for which $2N_{\mathrm{sq}}(A;C|B)_{\rho}=I(A;C|BE)_{\rho}$. Now we introduce an ancillary system $B'$ in the state $\proj{0}$, and composite state becomes $\rho_{ABCEB'}=\rho_{ABCE}\otimes\proj{0}_{B'}$. Then any CPTP operation $\Lambda$ on Bob, leading to $\sigma_{A\widetilde{B}CE}=\Lambda_{B\to\widetilde{B}} \otimes \id_{ACE}(\rho_{ABC})$, can be simulated by applying a unitary operation on $BB'(U_{BB'\to \widetilde{B}\widetilde{B'}})$ and tracing out $\widetilde{B'}$ following Stinespring dilation, as
\begin{align}
&\sigma_{A\widetilde{B}CE}\nonumber\\ &=\mathrm{Tr}_{\widetilde{B'}}\Big[U_{BB'\to \widetilde{B}\widetilde{B'}}\Big(\rho_{ABCE}\otimes\proj{0}_{B'}\Big)U_{BB'\to \widetilde{B}\widetilde{B'}}^\dagger\Big]\nonumber \\
&=\mathrm{Tr}_{\widetilde{B'}}\Big[\sigma_{A\widetilde{B}\widetilde{B'}CE}\Big].
\end{align}
Noting that QCMI remains unchanged under unitary operations on the conditioning system, we have
\begin{align}
2N_{\mathrm{sq}}(A;C|B)_{\rho}&=I(A;C|BE)_{\rho}\nonumber\\ 
&=I(A;C|\widetilde{B}\widetilde{B'}E)_{\sigma}\nonumber\\ 
& \geq 2N_{\mathrm{sq}}(A;C|\widetilde{B}\widetilde{B'})_{\sigma}\nonumber \\
& \geq 2N_{\mathrm{sq}}(A;C|\widetilde{B})_{\sigma}.
\end{align}
We use the relation~\eqref{eq:MonPtraceB} in the last inequality.

\subsection{Monotonicity under $\mathbb{LOCC}_{\mathbf{A}\leftrightarrow \mathbf{C}}$}
\label{Appendix B}
Here, we prove the monotonicity under $\mathbb{LOCC}_{\mathbf{A}\leftrightarrow \mathbf{C}}$. According to~\cite{Vidal00, Christandl04},  we need to show that $N_{\mathrm{sq}}$ is a strong monotone under local operations on Alice and Charlie, and it is convex. The latter we have demonstrated earlier.

Now, we show strong monotonicity under a local operation by Alice. Let us consider a state $\rho_{ABC}$ and any unilocal quantum instruments $\{ \Lambda^m_A \}$ applied on Alice's system $A$, where each $\Lambda^m_A$ is a completely positive map and their sum is trace-preserving. Then, the strong monotonicity implies
\begin{align}
N_{\mathrm{sq}}(A;C|B)_\rho \geq \sum_m p_m N_{\mathrm{sq}}(A;C|B)_{\sigma^m}, \label{eq:StrongMonot}  
\end{align}
where $p_m=\Tr [\Lambda^m_A(\rho_{ABC})]$ and $\sigma_{ABC}^m=\Lambda^m_A(\rho_{ABC})/{p_m}$. 

To prove the above, (i) we consider the state $\rho_{ABCE}$, serving as a non-signaling extension of the initial state $\rho_{ABC}$, for which $2N_{\mathrm{sq}}(A;C|B)_\rho = I(A;C|BE)_\rho$. (ii) Now we introduce two ancillary systems $R$ and $R'$ and both are initialized in the $\ket{0}_{R/R'}$ state. Thus, the composite becomes $\rho_{ARR'BCE}=\rho_{ABCE}\otimes\proj{0}_{R}\otimes\proj{0}_{R'}$. (iii) Next, we apply a global unitary operation $U_{ARR'}$ on $ARR'$. The rest of the systems will go through an identity channel. As a result, we have $\sigma_{ARR'BCE}=U_{ARR'}(\rho_{ARR'BCE})U_{ARR'}^\dag$.
(iv) Following this, we trace out the ancillary system $R'$. It can be written as, 
\begin{align}
    \sigma_{ARBCE}&=\sum_m \  \proj{m}_{R}\otimes \Lambda^m_A(\rho_{ABCE}) \nonumber \\  
    &= \sum_m \ p_m \ \proj{m}_{R}\otimes \sigma_{ABCE}^m, \label{Eq:QuantumInstru}
\end{align}
where $p_m=\Tr [\Lambda^m_A(\rho_{ABCE})]$ is the probability, $\sigma_{ABCE}^m=\Lambda^m_A(\rho_{ABCE})/{p_m}$ is the normalized state after the instrument $\Lambda^m_A$ applied by Alice, and $\{\ket{m}_R\}$ are the orthonormal bases on $R$. Further, we ensure that $\Tr_E[\sigma_{ABCE}^m]=\sigma_{ABC}^m$ as considered in Eq.~\eqref{eq:StrongMonot}. Our arguments proving strong monotonicity are based on QCMI, as in the following: 
\begin{align} 
2 \ N_{\mathrm{sq}}(A;C|B)_\rho & \overset{(i)}{=} I(A;C|BE)_\rho \nonumber \\ 
 &\overset{(ii)}{=}I(ARR';C|BE)_\rho \nonumber \\ 
&\overset{(iii)}{=}I(ARR';C|BE)_{\sigma}\nonumber \\
 &\overset{(iv)}{\geq}I(AR;C|BE)_{\sigma}\nonumber \\
&\overset{(v)}{=}I(R;C|BE)_{\sigma}+I(A;C|BER)_{\sigma}\nonumber \\
&\overset{(vi)}{\geq}I(A;C|BER)_{\sigma}\nonumber \\
&\overset{(vii)}{=}\sum_m \  p_m \ I(A;C|BE)_{\sigma^m} \nonumber \\
&\overset{(viii)}{\geq} 2\sum_m \  p_m  \ N_{\mathrm{sq}}(A;C|B)_{\sigma^m}. \label{Eq:B4}
\end{align}
The equality (i) is due to the choice of the extension $\rho_{ABCE}$ of the initial state $\rho_{ABC}$ as mentioned earlier. Equalities (ii) and (iii) are because the QCMI does not alter upon attaching (uncorrelated) ancilla $RR'$ with Alice and a unitary $U_{ARR'}$ respectively. The inequality (iv) follows from the data-processing inequality of QCMI under partial tracing of $R'$. The equality (v) and the inequality are due to the chain rule and inequality (vi) due to the non-negativity of QCMI ($I(R;C|BE)_\sigma\geq 0$). For the state~\eqref{Eq:QuantumInstru}, the QCMI satisfies
\begin{align*}
I(A;C|BER)_{\sigma} = \sum_m \  p_m \ I(A;C|BE)_{\sigma^m},    
\end{align*}
as represented by the equality (vii). The inequality (viii) is the result of the fact that $I(A;C|BE)_{\sigma^m} \geq 2 N_{\mathrm{sq}}(A;C|B)_{\sigma^m}$ for all $m$. Thus, we prove strong monotonicity of $N_{\mathrm{sq}}$ under local operation by Alice. Similar steps can show the strong monotonicity of local operations on Charlie.

As a corollary, we also prove the monotonicity of any trace-preserving local operation on Alice (also on Charlie)
\begin{align}
&N_{\mathrm{sq}}(A;C|B)_\rho\nonumber\\
&\hspace{0.5cm}\geq \sum_m p_m N_{\mathrm{sq}}(A;C|B)_{\sigma^m} \overset{(ix)}{\geq} N_{\mathrm{sq}}(A;C|B)_\sigma. \label{eq:WeakMonot}    
\end{align}
The inequality (ix) results from the convexity of the measure $N_{\mathrm{sq}}$.

\subsection{Monotonicity under $\mathbb{LOQC}_{\mathbf{A}\to \mathbf{B}}$ and $\mathbb{LOQC}_{\mathbf{C}\to \mathbf{B}}$}
We first demonstrate the monotonicity of one-way local operation and quantum communication from Alice to Bob ($\mathbb{LOQC}_{\mathbf{A}\to \mathbf{B}}$). Consider (i) a state $\rho_{ABC}$ and its extension $\rho_{ABCE}$ satisfying $2N_{\mathrm{sq}}(A;C|B)_{\rho}=I(A;C|BE)_\rho$. (ii) Now a local operation $\Lambda_A: A \to \widetilde{A}A'$ performed by Alice, leading to $\sigma_{\widetilde{A}A'BCE}=\Lambda_A \otimes \id_{BCE}(\rho_{ABCE})$. After this, (iii) the quantum information $A'$, initially possessed by Alice, is transferred to Bob via the ideal quantum channel, $\id_{A'\to B'}$, i.e., $\widetilde{A}A'|BE|C \to \widetilde{A}|BB'E|C$. Finally, (iv) a local operation $\mathcal{E}_{BB'}: BB' \to \widetilde{B}$ is applied on Bob's systems, as a result we have $\tau_{\widetilde{A}\widetilde{B}CE} = \mathcal{E}_{BB'} \otimes \id_{\widetilde{A}CE} (\sigma_{\widetilde{A}BB'CE})$. Note, these are the generic steps to execute an arbitrary $\mathbb{LOQC}_{\mathbf{A}\to \mathbf{B}}$. To prove the monotonicity, we employ the following arguments based on the properties of QCMI:
\begin{align}
 2 \ N_{\mathrm{sq}}(A;C|B)_{\rho} & \overset{(i)}{=} I(A;C|BE)_\rho \nonumber \\ 
 & \overset{(ii)}{\geq} I(\widetilde{A}A';C|BE)_\sigma \nonumber\\
 & \overset{(iii)}{\geq} I(\widetilde{A};C|BB'E)_\sigma \nonumber  \\
 & \overset{(iv)}{\geq} 2 \ N_{\mathrm{sq}}(\widetilde{A};C|BB')_\sigma \nonumber \\ 
 &\overset{(v)}{\geq} 2 \ N_{\mathrm{sq}}(\widetilde{A};C|\widetilde{B})_\tau. 
\end{align}
The equality (i) is due to the choice of the extension $\rho_{ABCE}$ of the state $\rho_{ABC}$, as mentioned above. The inequality (ii) is due to the data-processing inequality of QCMI. The inequality (iii) can be seen from the chain rule $I(\widetilde{A}A';C|BE)_\sigma = I(A';C|BE)_\sigma + I(\widetilde{A};C|BB'E)_\sigma$, where $B'\simeq A'$ and $I(A';C|BE)_\sigma \geq 0$. The inequalities (iv) and (v) are due to the definition of $N_{\mathrm{sq}}$ and monotonicity of $N_{\mathrm{sq}}$ under local trace-preserving operation (CPTP map) by Bob (see Section~\ref{App:MonoOnB}). The monotonicity under $\mathbb{LOQC}_{\mathbf{C}\to \mathbf{B}}$ can be shown similarly.

\section{State transformations}
\label{App:State_transformation}

\noindent 
\subsection{Proof of the converse bound in Theorem~\ref{Lem:interconversion:necessary}.} \label{subsec:converse}

Here, we present the maximum possible rate $R_C$ at which one quantum state $(\rho_{1})_{ABC}^{\otimes n}$ can be converted to the other $(\rho_{2})_{ABC}^{\otimes m_n}$ using the free operations defined by our resource theory. The proof is similar to the converse of the entanglement manipulation theorem shown in Ref.~\cite{Wilde_2016}. Let us consider there exists a free operation $\Lambda_{\rm{free}}$ such that 
\begin{align}
    \frac{1}{2}\norm{\Lambda_{\rm{free}} (\rho_1^{\otimes n}) - \rho_2^{\otimes m_n} }_1\le \varepsilon.
\end{align}
Using our continuity relation~\eqref{eq:app-Lem:continuity}, we have 
\begin{align}
    &N_{\mathrm{sq}}(A;C|B)_{\Lambda_{\rm{free}}(\rho_1^{\otimes n})}\geq N_{\mathrm{sq}}(A;C|B)_{\rho_2^{\otimes m_n}} - \frac{f(\varepsilon)}{2}\nonumber \\ 
    &= m_n N_{\mathrm{sq}}(A;C|B)_{\rho_2}
    -2m_n\sqrt{\varepsilon}\log_2 \min \{\abs{A},\abs{C}\} \nonumber \\
    &\qquad  {+}  (1{+}2\sqrt{\varepsilon}) \ h_2\left(\frac{2\sqrt{\varepsilon}}{1{+}2\sqrt{\varepsilon}}\right).\label{Eq:continuity2}
\end{align}
The second equality is due to the additivity property of the measure $N_{\mathrm{sq}}$, and we expand $f(\varepsilon)$ in the continuity relation~\eqref{eq:app-Lem:continuity}. Now, we can write down the following sets of equations:
\begin{align}
    &nN_{\mathrm{sq}}(A;C|B)_{\rho_1} = N_{\mathrm{sq}}(A;C|B)_{\rho_1^{\otimes n}}\nonumber\\
    &{\geq}N_{\mathrm{sq}}(A;C|B)_{\Lambda_{\rm{free}}(\rho_1^{\otimes n})} \nonumber \\
    &\geq m_n N_{\mathrm{sq}}(A;C|B)_{\rho_2} 
    -2m_n\sqrt{\varepsilon}\log_2 \min \{\abs{A},\abs{C}\} \nonumber \\&\hspace{0.5cm} {+}  (1{+}2\sqrt{\varepsilon}) \ h_2\left(\frac{2\sqrt{\varepsilon}}{1{+}2\sqrt{\varepsilon}}\right).  
\end{align}
The first equality is due to the additivity of the quantum non-Markovianity measure. The first inequality is due to the monotonicity of the quantum non-Markovianity measure under free operation. The second inequality is due to~\eqref{Eq:continuity2}. Now dividing both sides by ${nN_{\mathrm{sq}}(A;C|B)_{\rho_2}}$ and rearranging the equations, we have 
\begin{align}
   & \frac{m_n}{n} \Bigg(1-\frac{2\sqrt{\varepsilon}\log_2 \min \{\abs{A},\abs{C}\}}{N_{\mathrm{sq}}(A;C|B)_{\rho_2}}\Bigg)\nonumber \\
    & \leq \frac{N_{\mathrm{sq}}(A;C|B)_{\rho_1}}{N_{\mathrm{sq}}(A;C|B)_{\rho_2}} - \frac{  (1{+}2\sqrt{\varepsilon})}{nN_{\mathrm{sq}}(A;C|B)_{\rho_2}} \ h_2\left(\frac{2\sqrt{\varepsilon}}{1{+}2\sqrt{\varepsilon}}\right).
\end{align}
Taking limit at $n \to \infty$ and $\varepsilon \to 0$, we have 
\begin{align}
    R = \lim_{\varepsilon \to 0}\lim_{n \to \infty} \frac{m_n}{n} \leq\frac{N_{\mathrm{sq}}(A;C|B)_{\rho_1}}{N_{\mathrm{sq}}(A;C|B)_{\rho_2}}=R_C.
\end{align}

\subsection{Proof of the achievability bound in Theorem~\ref{Lem:interconversion:sufficient}} \label{subsec:achievable}

In this section, we establish a sufficient condition to convert a state $(\rho_1)_{ABC}$ to $(\rho_2)_{ABC}$ using free operations, i.e., local trace-preserving operations by all parties, $\mathbb{LOCC}_{\mathbf{A}\rightarrow \mathbf{C}}$, $\mathbb{LOCC}_{\mathbf{C}\rightarrow \mathbf{A}}$, and $\mathbb{QC}_{\mathbf{A} \to \mathbf{B}}$ or $\mathbb{QC}_{\mathbf{C} \to \mathbf{B}}$. Note that this protocol fails when $\max\Big\{I(A\rangle C)_{\rho_1},I(C\rangle A)_{\rho_1}\Big\}{<}0$. The overall strategy follows three steps. 
    
Firstly, Alice and Charlie want to distill as many shared ebits as possible via a state-merging protocol. The state distillation is achieved via a state-merging protocol, i.e., Alice wants to merge her system into Charlie's lab via $\mathbb{LOCC}_{\mathbf{A}\rightarrow \mathbf{C}}$. For a large $n$, this state-merging protocol on $(\rho_1)_{ABC}^{\otimes n}$ generates ${\approx} n.I(A\rangle C)_{\rho_1}$ shared ebits between Alice and Charlie (when $I(A\rangle C)_{\rho_1}{\ge}0$). Considering the purification of $\rho_1$ to be $\ket{\psi_1}$ with the environment system R, the corresponding resource-inequality for state-merging follows (see Eq. 58 in~\cite{Devetak_2008_resource_framework})
\begin{align}
        & (\psi_1)_{ABCR} + I(A:BR)_{\psi_1}[c \to c] \nonumber \\ &	\succeq (\psi_1)_{BC'CR} + I(A\rangle C)_{\rho_1} [qq]. \label{Eq:merging1}
\end{align}
Here, $A\simeq C'$, the above resource inequality indicates Alice's system $A$ is now with Charlie which we denote it by $C'$. Note we can write $I(A\rangle C)_{\rho_1}$ and $I(A\rangle C)_{\psi_1}$ interchangeably.
Alternatively, it is also possible for Charlie to merge his system into Alice's lab via $\mathbb{LOCC}_{\mathbf{C}\to \mathbf{A}}$ and establish ${\approx} n.I(C\rangle A)_{\rho_1}$ shared ebits between them (when $I(C\rangle A)_{\rho_1}{\ge}0$). 
\begin{align}
        & (\psi_1)_{ABCR} + I(C:BR)_{\psi_1}[c \to c] \nonumber \\ &	\succeq (\psi_1)_{AA'BR} + I(C\rangle A)_{\rho_1} [qq]. \label{Eq:merging2}
    \end{align}
Here $A'\simeq C$. Depending on the amount of shared ebits, Alice and Charlie perform the state-merging in Eq.~\eqref{Eq:merging1} or Eq.~\eqref{Eq:merging2} to establish $\approx n.\max\Big\{I(A\rangle C)_{\rho_1}, I(C\rangle A)_{\rho_1}\Big\}$ shared ebits between them.

After the entanglement distillation, the state $\rho_1^{\otimes n}$ is discarded. Charlie now prepares the state $(\rho_2)_{C'C''C}^{\otimes m_n}$ with $C'\simeq A$ and $C''\simeq B$.  Charlie transfers the system $(C'')^{m_n}$ to Bob via quantum communication $\mathbb{QC}_{\mathbf{C} \to \mathbf{B}}$, i.e. $m_n$ copies of the map $\id_{C''\to B}$ is implemented. Now, Charlie aims to use the previously established shared ebits to perform a state-splitting protocol to transfer system $C'$ to Alice's lab, with the map $\id_{C'\to A}$. Considering the purification of $\rho_2$ to be $\psi_2$ with purifying system $R'$, and writing $C'$ and $A$ interchangeably in the entropic quantities, let us recall the relevant resource inequality for such state splitting~\cite{Yard09}:
\begin{align}
        &(\psi_2)_{BC'CR'}+ \frac{1}{2}I(A{:}BR')_{\psi_{2}}[q \to q]\nonumber\\
        &\hspace{1cm}+ \frac{1}{2}I(A;C)_{\psi_{2}} [qq]  	\succeq (\psi_2)_{ABCR'}. \label{Eq:splitting1}
    \end{align}
Although quantum communication from Charlie to Alice is not a free operation, Charlie can simulate some quantum channels with part of the shared ebits and classical communication via a teleportation protocol, i.e., he can perform the protocol by only using classical communication and consuming additional shared ebits. Replacing $[q\to q]$ with $[qq] + 2 [c\to c]$ (teleportation) in Eq:~\eqref{Eq:splitting1}, we have 
\begin{align}
        & (\psi_2)_{BC'CR'}+ \frac{1}{2}\Big(I(A{:}BR)_{\psi_{2}}{+}I(A;C)_{\psi_{2}}\Big)[qq] \nonumber \\ 	&+ I(A{:}BR)_{\psi_{2}} [c\to c]  \succeq (\psi_2)_{ABCR}. \label{Eq:splitting_CC_1}
\end{align}
Note $\frac{1}{2}\Big(I(A{:}BR)_{\psi_{2}}{+}I(A;C)_{\psi_{2}}\Big){=}S(A)_{\psi_2}$, We can write $S(A)_{\psi_2}$ and $S(A)_{\rho_2}$ interchangeably. So to transfer $m_n$ copies of the system $A$ of $\rho_2$, Alice and Charlie need ${\approx m_nS(A)_{\rho_2}}$ ebits. The protocol is successful whenever previously distilled shared entanglement, $n.\max\Big\{ I(A\rangle C)_{\rho_1},I(C\rangle A)_{\rho_1} \Big\}$ is more than $m_n S(A)_{\rho_2}$, i.e. 
\begin{align}
&n \max\Big\{ I(A\rangle C)_{\rho_1},I(C\rangle A)_{\rho_1} \Big\} \ge m_n S(A)_{\rho_2} \\
 &\implies   \frac{m_n}{n}\leq \frac{\max\Big( I(A\rangle C)_{\rho_1},I(C\rangle A)_{\rho_1} \Big)}{S(A)_{\rho_2}}. \label{Eq:condition_st_trans1}
\end{align}

Alternatively, instead of Charlie, Alice could produce the state $(\rho_2)^{\otimes m_n}_{AA''A'}$ with $A''\simeq B$ and $A'\simeq C$ (all the subsystems are with Alice), send $(A'')^{m_n}$ to Bob via $\mathbb{QC}_{\mathbf{A}{\to}\mathbf{B}}$ (implementing $m_n$ copies of the map $\id_{A''\to B}$). She can then perform the state-splitting with the corresponding resource inequality:
\begin{align}
        & (\psi_2)_{AA'BR'}+ \frac{1}{2}I(C{:}BR)_{\psi_{2}}[q \to q] \nonumber\\
        &\hspace{1cm}+ \frac{1}{2}I(A;C)_{\psi_{2}} [qq]  	\succeq (\psi_2)_{ABCR'}. \label{Eq:splitting2}
\end{align}
Simulating the $[q\to q]$ channel via teleportation, we have the state-splitting protocol:
\begin{align}
        & (\psi_2)_{AA'BR'}+ \frac{1}{2}\Big(I(C{:}BR)_{\psi_{2}} + I(A;C)_{\psi_{2}}\Big)[qq] \nonumber \\ 	
        &\hspace{1cm} + I(C{:}BR)_{\psi_{2}} [c\to c]   	\succeq (\psi_2)_{ABCR'}. \label{Eq:splitting_CC_2}
    \end{align}
Considering
\begin{align*}
    \frac{1}{2}\Big(I(C{:}BR)_{\psi_{2}}{+}I(A;C)_{\psi_{2}}\Big){=}S(C)_{\psi_2}{=}S(C)_{\rho_2},
\end{align*} 
In this case, the available shared ebits are $n.\max\Big\{ I(A\rangle C)_{\rho_1},I(C\rangle A)_{\rho_1}\Big\}$ need to exceed $m_n.S(C)_{\rho_2}$ to make the state-splitting in Eq.~\eqref{Eq:splitting_CC_2} successful:
\begin{align}
&n \max\Big\{ I(A\rangle C)_{\rho_1},I(C\rangle A)_{\rho_1} \Big\} \ge m_n H(C)_{\rho_2} \\
&\implies
    \frac{m_n}{n}\le \frac{\max\Big\{ I(A\rangle C)_{\rho_1},I(C\rangle A)_{\rho_1} \Big\}}{S(C)_{\rho_2}}. \label{Eq:condition_st_trans2}
\end{align}
Combining Eq.~\eqref{Eq:condition_st_trans1} or Eq.~\eqref{Eq:condition_st_trans2}, the maximum possible rate of inter-conversion is upper bounded by
\begin{align}
   &R=\lim_{\varepsilon \to 0}\lim_{n\to \infty}  \frac{m_n}{n} \nonumber \\
   &\hspace{1cm}\leq \frac{\max\Big\{ I(A\rangle C)_{\rho_1},I(C\rangle A)_{\rho_1} \Big\}} {\min \Big \{S(A)_{\rho_2},S(C)_{\rho_2}\Big\}}=R_{A}. \label{Eq:condition_suff}
\end{align}

\subsection{Proof of Corollary~\ref{Cor:nec_vs_suf}} \label{subsec:Cor:nec_vs_suf}
Let us assume for $\rho_{ABC}$, $E$ is the extension that gives the quantum non-Markovianity measure, $N_{\mathrm{sq}}(A;C|B)_{\rho}=I(A;C|BE)_{\rho}/2$. Also, we assume $R$ is an extension such that $\rho_{ABCER}$ is a pure state. Now we have
\begin{align}
    &N_{\mathrm{sq}}(A;C|B)_{\rho}=\frac{1}{2} I(A;C|BE)_{\rho} \nonumber \\
    &\hspace{0.5cm}=\frac{1}{2}(S(A|BE)_{\rho}-S(A|BCE)_{\rho}) \nonumber \\
    &\hspace{0.5cm}=\frac{1}{2}(S(A|BE)_{\rho}+S(A|R)_{\rho})\le S(A)_{\rho}.
\end{align}
By symmetry, we also have $N_{\mathrm{sq}}(A;C|B)_{\rho}{\le}S(C)_\rho$, hence 
\begin{align}
N_{\mathrm{sq}}(A;C|B)_{\rho} \le \min\Big\{S(A)_{\rho},S(C)_\rho\Big\}.
\end{align}
Now, we prove the lower bound. Consider $I(A{\rangle}C)_{\rho}{=}-S(A|C)_{\rho}{=}S(A|BER)_{\rho}{\le}S(A|BE)_{\rho}$, also we have  $I(A{\rangle}C)_{\rho}{=}S(A|BER)_{\rho}{\le}S(A|R)_{\rho}$. The last inequality is because conditioning reduces entropy. Adding the inequalities and dividing by 2, we have 
\begin{align}
N_{\mathrm{sq}}(A;C|B)_{\rho} &= \frac{1}{2}(S(A|BE)_{\rho}+S(A|R)_{\rho})\nonumber\\
& \ge S(A|BER)_{\rho}=I(A{\rangle}C)_{\rho}.
\end{align}
Similarly, we can show $N_{\mathrm{sq}}(A;C|B)_{\rho}{\ge}I(C{\rangle}A)_{\rho}$, hence we have
\begin{align}
    {N_{\mathrm{sq}}(A;C|B)_{\rho}} \ge \max\Big\{ I(A\rangle C)_{\rho},I(C\rangle A)_{\rho} \Big\}.
\end{align}

\bibliographystyle{unsrtnat}
\bibliography{QNM}

\begin{thebibliography}{117}
\providecommand{\natexlab}[1]{#1}
\providecommand{\url}[1]{\texttt{#1}}
\expandafter\ifx\csname urlstyle\endcsname\relax
  \providecommand{\doi}[1]{doi: #1}\else
  \providecommand{\doi}{doi: \begingroup \urlstyle{rm}\Url}\fi

\bibitem[Einstein et~al.(1935)Einstein, Podolsky, and Rosen]{EPR35}
A.~Einstein, B.~Podolsky, and N.~Rosen.
\newblock Can quantum-mechanical description of physical reality be considered
  complete?
\newblock \emph{Phys. Rev.}, 47:\penalty0 777--780, May 1935.
\newblock \doi{10.1103/PhysRev.47.777}.
\newblock URL \url{https://doi.org/10.1103/PhysRev.47.777}.

\bibitem[Bell and Aspect(2011)]{Bbook04}
J.~S. Bell and A.~Aspect.
\newblock \emph{Speakable and Unspeakable in Quantum Mechanics: Collected
  Papers on Quantum Philosophy}.
\newblock Cambridge University Press, Apr 2011.
\newblock ISBN 9780521523387.
\newblock \doi{10.1017/CBO9780511815676}.
\newblock URL \url{https://doi.org/10.1017/CBO9780511815676}.

\bibitem[DiVincenzo(1995)]{DiV95}
D.~P. DiVincenzo.
\newblock Quantum computation.
\newblock \emph{Science}, 270\penalty0 (5234):\penalty0 255--261, Oct 1995.
\newblock \doi{10.1126/science.270.5234.255}.
\newblock URL \url{https://doi.org/10.1126/science.270.5234.255}.

\bibitem[MacFarlane et~al.(2003)MacFarlane, Dowling, and Milburn]{DM03}
A.~G.~J. MacFarlane, J.~P. Dowling, and G.~J. Milburn.
\newblock Quantum technology: the second quantum revolution.
\newblock \emph{Philos. Trans. Royal Soc. A .}, 361\penalty0 (1809):\penalty0
  1655--1674, Jun 2003.
\newblock \doi{10.1098/rsta.2003.1227}.
\newblock URL \url{https://doi.org/10.1098/rsta.2003.1227}.

\bibitem[Coles et~al.(2017)Coles, Berta, Tomamichel, and Wehner]{CBTW17}
P.~J. Coles, M.~Berta, M.~Tomamichel, and S.~Wehner.
\newblock Entropic uncertainty relations and their applications.
\newblock \emph{Rev. Mod. Phys.}, 89:\penalty0 015002, Feb 2017.
\newblock \doi{10.1103/RevModPhys.89.015002}.
\newblock URL \url{https://doi.org/10.1103/RevModPhys.89.015002}.

\bibitem[Das et~al.(2021)Das, B\"auml, Winczewski, and Horodecki]{DBWH21}
S.~Das, S.~B\"auml, M.~Winczewski, and K.~Horodecki.
\newblock Universal limitations on quantum key distribution over a network.
\newblock \emph{Phys. Rev. X}, 11:\penalty0 041016, Oct 2021.
\newblock \doi{10.1103/PhysRevX.11.041016}.
\newblock URL \url{https://doi.org/10.1103/PhysRevX.11.041016}.

\bibitem[Bennett et~al.(1993)Bennett, Brassard, Cr\'epeau, Jozsa, Peres, and
  Wootters]{BBC+93}
C.~H. Bennett, G.~Brassard, C.~Cr\'epeau, R.~Jozsa, A.~Peres, and W.~K.
  Wootters.
\newblock Teleporting an unknown quantum state via dual classical and
  {Einstein-Podolsky-Rosen} channels.
\newblock \emph{Phys. Rev. Lett.}, 70:\penalty0 1895--1899, Mar 1993.
\newblock \doi{10.1103/PhysRevLett.70.1895}.
\newblock URL \url{https://doi.org/10.1103/PhysRevLett.70.1895}.

\bibitem[Furusawa et~al.(1998)Furusawa, Sørensen, Braunstein, Fuchs, Kimble,
  and Polzik]{FSB+98}
A.~Furusawa, J.~L. Sørensen, S.~L. Braunstein, C.~A. Fuchs, H.~J. Kimble, and
  E.~S. Polzik.
\newblock Unconditional quantum teleportation.
\newblock \emph{Science}, 282\penalty0 (5389):\penalty0 706--709, Oct 1998.
\newblock \doi{10.1126/science.282.5389.706}.
\newblock URL \url{https://doi.org/10.1126/science.282.5389.706}.

\bibitem[Ekert(1991)]{Eke91}
A.~K. Ekert.
\newblock Quantum cryptography based on {B}ell's theorem.
\newblock \emph{Phys. Rev. Lett.}, 67:\penalty0 661--663, Aug 1991.
\newblock \doi{10.1103/PhysRevLett.67.661}.
\newblock URL \url{https://doi.org/10.1103/PhysRevLett.67.661}.

\bibitem[Primaatmaja et~al.(2023)Primaatmaja, Goh, Tan, Khoo, Ghorai, and
  Lim]{PGT+23}
I.~W. Primaatmaja, K.~T. Goh, E.~Y.-Z. Tan, J.~T.-F. Khoo, S.~Ghorai, and
  C.~C.-W. Lim.
\newblock Security of device-independent quantum key distribution protocols: a
  review.
\newblock \emph{{Quantum}}, 7:\penalty0 932, Mar 2023.
\newblock ISSN 2521-327X.
\newblock \doi{10.22331/q-2023-03-02-932}.
\newblock URL \url{https://doi.org/10.22331/q-2023-03-02-932}.

\bibitem[Simon(1997)]{Sim97}
D.~R. Simon.
\newblock On the power of quantum computation.
\newblock \emph{SIAM J. Comput.}, 26\penalty0 (5):\penalty0 1474--1483, Oct
  1997.
\newblock \doi{10.1137/S0097539796298637}.
\newblock URL \url{https://doi.org/10.1137/S0097539796298637}.

\bibitem[Bravyi et~al.(2018)Bravyi, Gosset, and König]{BGK18}
S.~Bravyi, D.~Gosset, and R.t König.
\newblock Quantum advantage with shallow circuits.
\newblock \emph{Science}, 362\penalty0 (6412):\penalty0 308--311, Oct 2018.
\newblock \doi{10.1126/science.aar3106}.
\newblock URL \url{https://doi.org/10.1126/science.aar3106}.

\bibitem[Pironio et~al.(2010)Pironio, Ac{\'\i}n, Massar, de~La~Giroday,
  Matsukevich, Maunz, Olmschenk, Hayes, Luo, Manning, et~al.]{PAM+10}
S.~Pironio, A.~Ac{\'\i}n, S.~Massar, A~B. de~La~Giroday, D.~N Matsukevich,
  P.~Maunz, S.~Olmschenk, D.~Hayes, Le. Luo, T~A. Manning, et~al.
\newblock Random numbers certified by {B}ell’s theorem.
\newblock \emph{Nature}, 464\penalty0 (7291):\penalty0 1021--1024, Apr 2010.
\newblock \doi{10.1038/nature09008}.
\newblock URL \url{https://doi.org/10.1038/nature09008}.

\bibitem[Colbeck and Kent(2011)]{RA11}
R.~Colbeck and A.~Kent.
\newblock Private randomness expansion with untrusted devices.
\newblock \emph{J. Phys. A Math. Theor.}, 44\penalty0 (9):\penalty0 095305, Feb
  2011.
\newblock \doi{10.1088/1751-8113/44/9/095305}.
\newblock URL \url{https://doi.org/10.1088/1751-8113/44/9/095305}.

\bibitem[Bera et~al.(2017)Bera, Acín, Kuś, Mitchell, and Lewenstein]{Bera17}
M.~N. Bera, A.~Acín, M.~Kuś, M.~W. Mitchell, and M.~Lewenstein.
\newblock Randomness in quantum mechanics: philosophy, physics and technology.
\newblock \emph{Rep. Prog. Phys.}, 80\penalty0 (12):\penalty0 124001, Nov 2017.
\newblock \doi{10.1088/1361-6633/aa8731}.
\newblock URL \url{https://doi.org/10.1088/1361-6633/aa8731}.

\bibitem[Chitambar and Gour(2019)]{CG19}
E.~Chitambar and G.~Gour.
\newblock Quantum resource theories.
\newblock \emph{Rev. Mod. Phys.}, 91:\penalty0 025001, Apr 2019.
\newblock \doi{10.1103/RevModPhys.91.025001}.
\newblock URL \url{https://doi.org/10.1103/RevModPhys.91.025001}.

\bibitem[Takagi and Regula(2019)]{TR19}
R.~Takagi and B.~Regula.
\newblock General resource theories in quantum mechanics and beyond:
  Operational characterization via discrimination tasks.
\newblock \emph{Phys. Rev. X}, 9:\penalty0 031053, Sep 2019.
\newblock \doi{10.1103/PhysRevX.9.031053}.
\newblock URL \url{https://doi.org/10.1103/PhysRevX.9.031053}.

\bibitem[Lami(2024)]{Lam24}
L.~Lami.
\newblock A solution of the generalised quantum {S}tein's lemma.
\newblock \emph{arXiv.2408.06410}, Oct 2024.
\newblock \doi{10.48550/arXiv.2408.06410}.
\newblock URL \url{https://doi.org/10.48550/arXiv.2408.06410}.

\bibitem[Barrett et~al.(2005)Barrett, Linden, Massar, Pironio, Popescu, and
  Roberts]{BLM+05}
J.~Barrett, N.~Linden, S.~Massar, S.~Pironio, S.~Popescu, and D.~Roberts.
\newblock Nonlocal correlations as an information-theoretic resource.
\newblock \emph{Phys. Rev. A}, 71:\penalty0 022101, Feb 2005.
\newblock \doi{10.1103/PhysRevA.71.022101}.
\newblock URL \url{https://doi.org/10.1103/PhysRevA.71.022101}.

\bibitem[Gallego et~al.(2012)Gallego, W\"urflinger, Ac\'{\i}n, and
  Navascu\'es]{GWAN12}
R.~Gallego, L.~E. W\"urflinger, A.~Ac\'{\i}n, and M.~Navascu\'es.
\newblock Operational framework for nonlocality.
\newblock \emph{Phys. Rev. Lett.}, 109:\penalty0 070401, Aug 2012.
\newblock \doi{10.1103/PhysRevLett.109.070401}.
\newblock URL \url{https://doi.org/10.1103/PhysRevLett.109.070401}.

\bibitem[Wiseman et~al.(2007)Wiseman, Jones, and Doherty]{WJD07}
H.~M. Wiseman, S.~J. Jones, and A.~C. Doherty.
\newblock Steering, entanglement, nonlocality, and the
  {E}instein-{P}odolsky-{R}osen paradox.
\newblock \emph{Phys. Rev. Lett.}, 98:\penalty0 140402, Apr 2007.
\newblock \doi{10.1103/PhysRevLett.98.140402}.
\newblock URL \url{https://doi.org/10.1103/PhysRevLett.98.140402}.

\bibitem[Gallego and Aolita(2015)]{GA15}
R.~Gallego and L.~Aolita.
\newblock Resource theory of steering.
\newblock \emph{Phys. Rev. X}, 5:\penalty0 041008, Oct 2015.
\newblock \doi{10.1103/PhysRevX.5.041008}.
\newblock URL \url{https://doi.org/10.1103/PhysRevX.5.041008}.

\bibitem[Horodecki et~al.(2009)Horodecki, Horodecki, Horodecki, and
  Horodecki]{HHHH09}
R.~Horodecki, P.~Horodecki, M.~Horodecki, and K.~Horodecki.
\newblock Quantum entanglement.
\newblock \emph{Rev. Mod. Phys.}, 81:\penalty0 865--942, Jun 2009.
\newblock \doi{10.1103/RevModPhys.81.865}.
\newblock URL \url{https://doi.org/10.1103/RevModPhys.81.865}.

\bibitem[Streltsov et~al.(2015)Streltsov, Singh, Dhar, Bera, and
  Adesso]{SSD+15}
A.~Streltsov, U.~Singh, H.~Shekhar Dhar, M.~Nath Bera, and G.~Adesso.
\newblock Measuring quantum coherence with entanglement.
\newblock \emph{Phys. Rev. Lett.}, 115:\penalty0 020403, Jul 2015.
\newblock \doi{10.1103/PhysRevLett.115.020403}.
\newblock URL \url{https://doi.org/10.1103/PhysRevLett.115.020403}.

\bibitem[Doherty et~al.(2002)Doherty, Parrilo, and Spedalieri]{DPS02}
A.~C. Doherty, P.~A. Parrilo, and F.~M. Spedalieri.
\newblock Distinguishing separable and entangled states.
\newblock \emph{Phys. Rev. Lett.}, 88:\penalty0 187904, Apr 2002.
\newblock \doi{10.1103/PhysRevLett.88.187904}.
\newblock URL \url{https://doi.org/10.1103/PhysRevLett.88.187904}.

\bibitem[Doherty et~al.(2004)Doherty, Parrilo, and Spedalieri]{DPS04}
A.~C. Doherty, P.~A. Parrilo, and F.~M. Spedalieri.
\newblock Complete family of separability criteria.
\newblock \emph{Phys. Rev. A}, 69:\penalty0 022308, Feb 2004.
\newblock \doi{10.1103/PhysRevA.69.022308}.
\newblock URL \url{https://doi.org/10.1103/PhysRevA.69.022308}.

\bibitem[Kaur et~al.(2019)Kaur, Das, Wilde, and Winter]{KDWW19}
E.~Kaur, S.~Das, M.~M. Wilde, and A.~Winter.
\newblock Extendibility limits the performance of quantum processors.
\newblock \emph{Phys. Rev. Lett.}, 123:\penalty0 070502, Aug 2019.
\newblock \doi{10.1103/PhysRevLett.123.070502}.
\newblock URL \url{https://doi.org/10.1103/PhysRevLett.123.070502}.

\bibitem[Kaur et~al.(2021)Kaur, Das, Wilde, and Winter]{KDWW21}
E.~Kaur, S.~Das, M.~M. Wilde, and A.~Winter.
\newblock Resource theory of unextendibility and nonasymptotic quantum
  capacity.
\newblock \emph{Phys. Rev. A}, 104:\penalty0 022401, Aug 2021.
\newblock \doi{10.1103/PhysRevA.104.022401}.
\newblock URL \url{https://doi.org/10.1103/PhysRevA.104.022401}.

\bibitem[Hayden et~al.(2004)Hayden, Jozsa, Petz, and Winter]{hayden2004}
P.~Hayden, R.~Jozsa, D.~Petz, and A.~Winter.
\newblock Structure of states which satisfy strong subadditivity of quantum
  entropy with equality.
\newblock \emph{Commun. Math. Phys.}, 246:\penalty0 359--374, Feb 2004.
\newblock \doi{10.1007/s00220-004-1049-z}.
\newblock URL \url{https://doi.org/10.1007/s00220-004-1049-z}.

\bibitem[Fawzi and Renner(2015)]{FR15}
O.~Fawzi and R.~Renner.
\newblock Quantum conditional mutual information and approximate {M}arkov
  chains.
\newblock \emph{Commun. Math. Phys.}, 340\penalty0 (2):\penalty0 575--611, Sep
  2015.
\newblock \doi{10.1007/s00220-015-2466-x}.
\newblock URL \url{https://doi.org/10.1007%2Fs00220-015-2466-x}.

\bibitem[Ruskai(2002)]{Rus02}
M.B. Ruskai.
\newblock Inequalities for quantum entropy: A review with conditions for
  equality.
\newblock \emph{J. Math. Phys.}, 43\penalty0 (9):\penalty0 4358–4375, Aug
  2002.
\newblock ISSN 1089-7658.
\newblock \doi{10.1063/1.1497701}.
\newblock URL \url{https://doi.org/10.1063/1.1497701}.

\bibitem[Lieb and Ruskai(2003)]{Lieb03}
E.~H. Lieb and M.~B. Ruskai.
\newblock {Proof of the strong subadditivity of quantum‐mechanical entropy}.
\newblock \emph{J. Math. Phys.}, 14\penalty0 (12):\penalty0 1938--1941, Nov
  2003.
\newblock ISSN 0022-2488.
\newblock \doi{10.1063/1.1666274}.
\newblock URL \url{https://doi.org/10.1063/1.1666274}.

\bibitem[Nielsen and Petz(2005)]{NP05}
M.~A. Nielsen and D.~Petz.
\newblock A simple proof of the strong subadditivity inequality.
\newblock \emph{Quantum Info. Comput.}, 5\penalty0 (6):\penalty0 507–513, Sep
  2005.
\newblock ISSN 1533-7146.
\newblock \doi{10.1007/s11128-011-0238-x}.
\newblock URL \url{https://doi.org/10.5555/2011670.2011678}.

\bibitem[Umegaki(1962)]{Ume62}
H.~Umegaki.
\newblock Conditional expectation in an operator algebra. iv. entropy and
  information.
\newblock \emph{Kodai Math. Semin. rep.}, 14\penalty0 (2):\penalty0 59 -- 85,
  Feb 1962.
\newblock \doi{10.2996/kmj/1138844604}.
\newblock URL \url{https://doi.org/10.2996/kmj/1138844604}.

\bibitem[Hiai and Petz(1991)]{HP91}
F.~Hiai and D.~Petz.
\newblock The proper formula for relative entropy and its asymptotics in
  quantum probability.
\newblock \emph{Commun. Math. Phys.}, 143\penalty0 (1):\penalty0 99--114, Dec
  1991.
\newblock \doi{10.1007/BF02100287}.
\newblock URL \url{https://doi.org/10.1007/BF02100287}.

\bibitem[Devetak and Yard(2008)]{Devetak08}
I.~Devetak and J.~Yard.
\newblock Exact cost of redistributing multipartite quantum states.
\newblock \emph{Phys. Rev. Lett.}, 100:\penalty0 230501, Jun 2008.
\newblock \doi{10.1103/PhysRevLett.100.230501}.
\newblock URL \url{https://doi.org/10.1103/PhysRevLett.100.230501}.

\bibitem[Yard and Devetak(2009)]{Yard09}
J.~T. Yard and I.~Devetak.
\newblock Optimal quantum source coding with quantum side information at the
  encoder and decoder.
\newblock \emph{IEEE Trans. Inf. Theory.}, 55\penalty0 (11):\penalty0
  5339--5351, Nov 2009.
\newblock ISSN 1557-9654.
\newblock \doi{10.1109/TIT.2009.2030494}.
\newblock URL \url{https://doi.org/10.1109/TIT.2009.2030494}.

\bibitem[Berta et~al.(2018{\natexlab{a}})Berta, Brand\~ao, Majenz, and
  Wilde]{Berta_2018_PRL}
M.~Berta, F.~G. S.~L. Brand\~ao, C.~Majenz, and M.~M. Wilde.
\newblock Conditional decoupling of quantum information.
\newblock \emph{Phys. Rev. Lett.}, 121:\penalty0 040504, Jul
  2018{\natexlab{a}}.
\newblock \doi{10.1103/PhysRevLett.121.040504}.
\newblock URL \url{https://doi.org/10.1103/PhysRevLett.121.040504}.

\bibitem[{Berta}(2009)]{berta2009singleshot}
M.~{Berta}.
\newblock Single-shot quantum state merging.
\newblock \emph{arXiv:0912.4495}, Dec 2009.
\newblock \doi{10.48550/arXiv.0912.4495}.
\newblock URL \url{https://doi.org/10.48550/arXiv.0912.4495}.

\bibitem[Sharma et~al.(2020)Sharma, Wakakuwa, and Wilde]{Conditional_otp}
K.~Sharma, E.~Wakakuwa, and M.M. Wilde.
\newblock Conditional quantum one-time pad.
\newblock \emph{Phys. Rev. Lett.}, 124:\penalty0 050503, Feb 2020.
\newblock \doi{10.1103/PhysRevLett.124.050503}.
\newblock URL \url{https://doi.org/10.1103/PhysRevLett.124.050503}.

\bibitem[Brandão et~al.(2011)Brandão, Christandl, and Yard]{Brandao11}
Fernando G. S.~L. Brandão, M.~Christandl, and J.~Yard.
\newblock Faithful squashed entanglement.
\newblock \emph{Commun. Math. Phys.}, 306:\penalty0 805, Aug 2011.
\newblock \doi{10.1007/s00220-011-1302-1}.
\newblock URL \url{https://doi.org/10.1007/s00220-011-1302-1}.

\bibitem[Ding et~al.(2016)Ding, Hayden, and Walter]{DHW16}
D.~Ding, P.~Hayden, and M.~Walter.
\newblock Conditional mutual information of bipartite unitaries and scrambling.
\newblock \emph{J. High Energy Phys.}, 12:\penalty0 145, Dec 2016.
\newblock \doi{10.1007/JHEP12(2016)145}.
\newblock URL \url{https://doi.org/10.1007/JHEP12(2016)145}.

\bibitem[Kaur et~al.(2017)Kaur, Wang, and Wilde]{Kaur17}
E.~Kaur, X.~Wang, and Mark~M. Wilde.
\newblock Conditional mutual information and quantum steering.
\newblock \emph{Phys. Rev. A}, 96:\penalty0 022332, Aug 2017.
\newblock \doi{10.1103/PhysRevA.96.022332}.
\newblock URL \url{https://doi.org/10.1103/PhysRevA.96.022332}.

\bibitem[Kim(2012)]{Kim12a}
I.~H. Kim.
\newblock Perturbative analysis of topological entanglement entropy from
  conditional independence.
\newblock \emph{Phys. Rev. B}, 86:\penalty0 245116, Dec 2012.
\newblock \doi{10.1103/PhysRevB.86.245116}.
\newblock URL \url{https://doi.org/10.1103/PhysRevB.86.245116}.

\bibitem[Renner and Wolf(2003)]{RW03}
R.~Renner and S.~Wolf.
\newblock New bounds in secret-key agreement: The gap between formation and
  secrecy extraction.
\newblock \emph{Advances in Cryptology — EUROCRYPT 2003}, pages 562--577, Jan
  2003.
\newblock \doi{10.1007/3-540-39200-9_35}.
\newblock URL \url{https://doi.org/10.1007/3-540-39200-9_35}.

\bibitem[Devetak and Winter(2005)]{DW05}
I.~Devetak and A.~Winter.
\newblock Distillation of secret key and entanglement from quantum states.
\newblock \emph{Proc. Math. Phys. Eng. Sci.}, 461\penalty0 (2053):\penalty0
  207--235, Jan 2005.
\newblock \doi{10.1098/rspa.2004.1372}.
\newblock URL \url{https://doi.org/10.1098/rspa.2004.1372}.

\bibitem[Winczewski et~al.(2022)Winczewski, Das, and Horodecki]{WDH22}
M.~Winczewski, T.~Das, and K.~Horodecki.
\newblock Limitations on a device-independent key secure against a
  non-signaling adversary via squashed nonlocality.
\newblock \emph{Phys. Rev. A}, 106:\penalty0 052612, Nov 2022.
\newblock \doi{10.1103/PhysRevA.106.052612}.
\newblock URL \url{https://doi.org/10.1103/PhysRevA.106.052612}.

\bibitem[Kaur et~al.(2022)Kaur, Horodecki, and Das]{KHD22}
E.~Kaur, K.~Horodecki, and S.~Das.
\newblock Upper bounds on device-independent quantum key distribution rates in
  static and dynamic scenarios.
\newblock \emph{Phys. Rev. Appl.}, 18:\penalty0 054033, Nov 2022.
\newblock \doi{10.1103/PhysRevApplied.18.054033}.
\newblock URL \url{https://doi.org/10.1103/PhysRevApplied.18.054033}.

\bibitem[Cerf and Cleve(1997)]{CC97}
N.~J. Cerf and R.~Cleve.
\newblock Information-theoretic interpretation of quantum error-correcting
  codes.
\newblock \emph{Phys. Rev. A}, 56:\penalty0 1721--1732, Sep 1997.
\newblock \doi{10.1103/PhysRevA.56.1721}.
\newblock URL \url{https://doi.org/10.1103/PhysRevA.56.1721}.

\bibitem[Pastawski et~al.(2017)Pastawski, Eisert, and Wilming]{PEW17}
F.~Pastawski, J.~Eisert, and H.~Wilming.
\newblock Towards holography via quantum source-channel codes.
\newblock \emph{Phys. Rev. Lett.}, 119:\penalty0 020501, Jul 2017.
\newblock \doi{10.1103/PhysRevLett.119.020501}.
\newblock URL \url{https://doi.org/10.1103/PhysRevLett.119.020501}.

\bibitem[Buscemi(2014)]{Bus14}
F.~Buscemi.
\newblock Complete positivity, {M}arkovianity, and the quantum data-processing
  inequality, in the presence of initial system-environment correlations.
\newblock \emph{Phys. Rev. Lett.}, 113:\penalty0 140502, Oct 2014.
\newblock \doi{10.1103/PhysRevLett.113.140502}.
\newblock URL \url{https://doi.org/10.1103/PhysRevLett.113.140502}.

\bibitem[Buscemi et~al.(2016)Buscemi, Das, and Wilde]{BDW16}
F.~Buscemi, S.~Das, and M.~M. Wilde.
\newblock Approximate reversibility in the context of entropy gain, information
  gain, and complete positivity.
\newblock \emph{Phys. Rev. A}, 93:\penalty0 062314, Jun 2016.
\newblock \doi{10.1103/PhysRevA.93.062314}.
\newblock URL \url{https://doi.org/10.1103/PhysRevA.93.062314}.

\bibitem[Huang and Guo(2022)]{HG22}
Z.~Huang and X.~Guo.
\newblock Classical and quantum parts of conditional mutual information for
  open quantum systems.
\newblock \emph{Phys. Rev. A}, 106:\penalty0 042412, Oct 2022.
\newblock \doi{10.1103/PhysRevA.106.042412}.
\newblock URL \url{https://doi.org/10.1103/PhysRevA.106.042412}.

\bibitem[Pollock et~al.(2018{\natexlab{a}})Pollock, Rodr\'{\i}guez-Rosario,
  Frauenheim, Paternostro, and Modi]{PRF+18}
Felix~A. Pollock, C\'esar Rodr\'{\i}guez-Rosario, Thomas Frauenheim, Mauro
  Paternostro, and Kavan Modi.
\newblock Non-markovian quantum processes: Complete framework and efficient
  characterization.
\newblock \emph{Physical Review A}, 97:\penalty0 012127, Jan
  2018{\natexlab{a}}.
\newblock \doi{10.1103/PhysRevA.97.012127}.
\newblock URL \url{https://doi.org/10.1103/PhysRevA.97.012127}.

\bibitem[Kuwahara et~al.(2020)Kuwahara, Kato, and Brand\~ao]{KKB20}
T.~Kuwahara, K.~Kato, and F.~G. S.~L. Brand\~ao.
\newblock Clustering of conditional mutual information for quantum {G}ibbs
  states above a threshold temperature.
\newblock \emph{Phys. Rev. Lett.}, 124:\penalty0 220601, Jun 2020.
\newblock \doi{10.1103/PhysRevLett.124.220601}.
\newblock URL \url{https://doi.org/10.1103/PhysRevLett.124.220601}.

\bibitem[Kim et~al.(2022)Kim, Shi, Kato, and Albert]{KSKA22}
I.~H. Kim, B.~Shi, K.~Kato, and V.~V. Albert.
\newblock Modular commutator in gapped quantum many-body systems.
\newblock \emph{Phys. Rev. B}, 106:\penalty0 075147, Aug 2022.
\newblock \doi{10.1103/PhysRevB.106.075147}.
\newblock URL \url{https://doi.org/10.1103/PhysRevB.106.075147}.

\bibitem[Sutter et~al.(2016)Sutter, Fawzi, and Renner]{SFR16}
D.~Sutter, O.~Fawzi, and R.~Renner.
\newblock Universal recovery map for approximate {M}arkov chains.
\newblock \emph{Proc. Math. Phys. Eng. Sci.}, 472\penalty0 (2186):\penalty0
  20150623, Feb 2016.
\newblock \doi{10.1098/rspa.2015.0623}.
\newblock URL \url{https://doi.org/10.1098%2Frspa.2015.0623}.

\bibitem[Rivas et~al.(2014)Rivas, Huelga, and Plenio]{Rivas2014}
Á. Rivas, S.~F. Huelga, and M.~B. Plenio.
\newblock Quantum non-{M}arkovianity: characterization, quantification and
  detection.
\newblock \emph{Rep. Prog. Phys.}, 77\penalty0 (9):\penalty0 094001, Aug 2014.
\newblock \doi{10.1088/0034-4885/77/9/094001}.
\newblock URL \url{https://doi.org/10.1088/0034-4885/77/9/094001}.

\bibitem[Allen et~al.(2017)Allen, Barrett, Horsman, Lee, and Spekkens]{Allen17}
J.~A. Allen, J.~Barrett, D.~C. Horsman, C.~M. Lee, and R.~W. Spekkens.
\newblock Quantum common causes and quantum causal models.
\newblock \emph{Phys. Rev. X}, 7:\penalty0 031021, Jul 2017.
\newblock \doi{10.1103/PhysRevX.7.031021}.
\newblock URL \url{https://doi.org/10.1103/PhysRevX.7.031021}.

\bibitem[Buscemi et~al.(2024)Buscemi, Gangwar, Goswami, Badhani, Pandit, Mohan,
  Das, and Bera]{BGG+24}
F.~Buscemi, R.~Gangwar, K.~Goswami, H.~Badhani, T.~Pandit, B.~Mohan, S.~Das,
  and M.~N. Bera.
\newblock Information revival without backflow: non-causal explanations of
  non-{M}arkovianity.
\newblock \emph{arXiv:2405.05326}, May 2024.
\newblock \doi{10.48550/arXiv.2405.05326}.
\newblock URL \url{https://doi.org/10.48550/arXiv.2405.05326}.

\bibitem[Wakakuwa(2017)]{Wakakuwa2017}
E.~Wakakuwa.
\newblock Operational resource theory of non-{M}arkovianity.
\newblock \emph{arXiv:1709.07248}, Oct 2017.
\newblock \doi{10.48550/arXiv.1709.07248}.
\newblock URL \url{https://doi.org/10.48550/arXiv.1709.07248}.

\bibitem[Wakakuwa(2021)]{Wakakuwa21}
E.~Wakakuwa.
\newblock Communication cost for non-{M}arkovianity of tripartite quantum
  states: A resource theoretic approach.
\newblock \emph{IEEE Trans. Inf. Theory}, 67\penalty0 (1):\penalty0 433--451,
  Jan 2021.
\newblock \doi{10.1109/TIT.2020.3028837}.
\newblock URL \url{https://doi.org/10.1109/TIT.2020.3028837}.

\bibitem[Christandl and Winter(2004)]{Christandl04}
M.~Christandl and A.~Winter.
\newblock Squashed entanglement: An additive entanglement measure.
\newblock \emph{J. Math. Phys.}, 45\penalty0 (3):\penalty0 829--840, Feb 2004.
\newblock ISSN 0022-2488.
\newblock \doi{10.1063/1.1643788}.
\newblock URL \url{https://doi.org/10.1063/1.1643788}.

\bibitem[Berta et~al.(2018{\natexlab{b}})Berta, Brand\~ao, Majenz, and
  Wilde]{Berta_2018}
M.~Berta, F.~G. S.~L. Brand\~ao, C.~Majenz, and M.~M. Wilde.
\newblock Deconstruction and conditional erasure of quantum correlations.
\newblock \emph{Phys. Rev. A}, 98:\penalty0 042320, Oct 2018{\natexlab{b}}.
\newblock \doi{10.1103/PhysRevA.98.042320}.
\newblock URL \url{https://doi.org/10.1103/PhysRevA.98.042320}.

\bibitem[Junge et~al.(2018)Junge, Renner, Sutter, Wilde, and Winter]{JRS+18}
M.~Junge, R.~Renner, D.~Sutter, M.~M. Wilde, and A.~Winter.
\newblock Universal recovery maps and approximate sufficiency of quantum
  relative entropy.
\newblock \emph{Ann. Henri Poincare}, 19\penalty0 (10):\penalty0 2955--2978,
  Aug 2018.
\newblock \doi{10.1007/s00023-018-0716-0}.
\newblock URL \url{https://doi.org/10.1007%2Fs00023-018-0716-0}.

\bibitem[Beigi and Gohari(2014)]{BG14}
S.~Beigi and ABD16 Gohari.
\newblock On dimension bounds for auxiliary quantum systems.
\newblock \emph{IEEE Transactions on Information Theory}, 60\penalty0
  (1):\penalty0 368--387, Jan 2014.
\newblock \doi{10.1109/TIT.2013.2286079}.
\newblock URL \url{https://doi.org/10.1109/TIT.2013.2286079}.

\bibitem[Greenberger et~al.(1989)Greenberger, Horne, and Zeilinger]{GHZ89}
D.~M. Greenberger, M.~A. Horne, and A.~Zeilinger.
\newblock Going beyond {B}ell's theorem.
\newblock In \emph{Bell's {{Theorem}}, {{Quantum Theory}} and {{Conceptions}}
  of the {{Universe}}}, pages 69--72. {Springer Netherlands}, 1989.
\newblock \doi{10.1007/978-94-017-0849-4_10}.
\newblock URL \url{https://doi.org/10.1007/978-94-017-0849-4_10}.

\bibitem[Mermin(1990)]{mermin90}
N.~David Mermin.
\newblock Quantum mysteries revisited.
\newblock \emph{Am. J. Phys.}, 58\penalty0 (8):\penalty0 731--734, Aug 1990.
\newblock ISSN 0002-9505, 1943-2909.
\newblock \doi{10.1119/1.16503}.
\newblock URL \url{https://doi.org/10.1119/1.16503}.

\bibitem[D\"ur et~al.(2000)D\"ur, Vidal, and Cirac]{DVC00}
W.~D\"ur, G.~Vidal, and J.~I. Cirac.
\newblock Three qubits can be entangled in two inequivalent ways.
\newblock \emph{Phys. Rev. A}, 62:\penalty0 062314, Nov 2000.
\newblock \doi{10.1103/PhysRevA.62.062314}.
\newblock URL \url{https://doi.org/10.1103/PhysRevA.62.062314}.

\bibitem[Li and Winter(2018)]{Li18}
K.~Li and A.~Winter.
\newblock Squashed entanglement, k-extendibility, quantum {M}arkov chains, and
  recovery maps.
\newblock \emph{Found. Phys.}, 48:\penalty0 910, Feb 2018.
\newblock \doi{10.1007/s10701-018-0143-6}.
\newblock URL \url{https://doi.org/10.1007/s10701-018-0143-6}.

\bibitem[Chitambar et~al.(2014)Chitambar, Leung, Mančinska, Ozols, and
  Winter]{CLM+14}
E.~Chitambar, D.~Leung, L.~Mančinska, M.~Ozols, and A.~Winter.
\newblock Everything you always wanted to know about {LOCC} (but were afraid to
  ask).
\newblock \emph{Commun. Math. Phys.}, 328\penalty0 (1):\penalty0 303–326, Mar
  2014.
\newblock ISSN 1432-0916.
\newblock \doi{10.1007/s00220-014-1953-9}.
\newblock URL \url{https://doi.org/10.1007/s00220-014-1953-9}.

\bibitem[Datta et~al.(2016)Datta, Hsieh, and Oppenheim]{Datta_2016}
N.~Datta, M.~Hsieh, and J.~Oppenheim.
\newblock {An upper bound on the second order asymptotic expansion for the
  quantum communication cost of state redistribution}.
\newblock \emph{J. Math. Phys.}, 57\penalty0 (5):\penalty0 052203, May 2016.
\newblock ISSN 0022-2488.
\newblock \doi{10.1063/1.4949571}.
\newblock URL \url{https://doi.org/10.1063/1.4949571}.

\bibitem[V.(2000)]{Vidal00}
Guifré V.
\newblock Entanglement monotones.
\newblock \emph{J. Mod. Opt.}, 47\penalty0 (2-3):\penalty0 355--376, Jul 2000.
\newblock \doi{10.1080/09500340008244048}.
\newblock URL \url{https://doi.org/10.1080/09500340008244048}.

\bibitem[Christandl et~al.(2007)Christandl, Ekert, Horodecki, Horodecki,
  Oppenheim, and Renner]{Christandl_2007}
M.~Christandl, A.~Ekert, M.~Horodecki, P.~Horodecki, J.~Oppenheim, and
  R.~Renner.
\newblock \emph{Unifying classical and quantum key distillation}.
\newblock Springer Berlin Heidelberg, 2007.
\newblock ISBN 978-3-540-70936-7.
\newblock \doi{10.1007/978-3-540-70936-7_25}.
\newblock URL \url{https://doi.org/10.1007/978-3-540-70936-7_25#citeas}.

\bibitem[Cerf et~al.(2002)Cerf, Massar, and Schneider]{Schneider_2002}
N.~J. Cerf, S.~Massar, and S.~Schneider.
\newblock Multipartite classical and quantum secrecy monotones.
\newblock \emph{Phys. Rev. A}, 66:\penalty0 042309, Oct 2002.
\newblock \doi{10.1103/PhysRevA.66.042309}.
\newblock URL \url{https://doi.org/10.1103/PhysRevA.66.042309}.

\bibitem[Sharma and Warsi(2013)]{SW13}
N.~Sharma and N.~A. Warsi.
\newblock Fundamental bound on the reliability of quantum information
  transmission.
\newblock \emph{Phys. Rev. Lett.}, 110:\penalty0 080501, Feb 2013.
\newblock \doi{10.1103/PhysRevLett.110.080501}.
\newblock URL \url{https://doi.org/10.1103/PhysRevLett.110.080501}.

\bibitem[Müller-Lennert et~al.(2013)Müller-Lennert, Dupuis, Szehr, Fehr, and
  Tomamichel]{MDSFT13}
M.~Müller-Lennert, F.~Dupuis, O.~Szehr, S.~Fehr, and M.~Tomamichel.
\newblock {On quantum {R}ényi entropies: A new generalization and some
  properties}.
\newblock \emph{J. Math. Phys.}, 54\penalty0 (12):\penalty0 122203, Dec 2013.
\newblock \doi{10.1063/1.4838856}.
\newblock URL \url{https://doi.org/10.1063/1.4838856}.

\bibitem[Wilde et~al.(2014)Wilde, Winter, and Yang]{WWY14}
M.~M. Wilde, A.~Winter, and D.~Yang.
\newblock Strong converse for the classical capacity of entanglement-breaking
  and {Hadamard} channels via a sandwiched {R\'enyi} relative entropy.
\newblock \emph{Commun. Math. Phys.}, 331\penalty0 (2):\penalty0 593--622, Oct
  2014.
\newblock \doi{10.1007/s00220-014-2122-x}.
\newblock URL \url{https://doi.org/10.1007/s00220-014-2122-x}.

\bibitem[Buscemi and Datta(2010)]{BD10}
F.~Buscemi and N.~Datta.
\newblock The quantum capacity of channels with arbitrarily correlated noise.
\newblock \emph{IEEE Trans. Inf. Theory}, 56\penalty0 (3):\penalty0 1447--1460,
  Mar 2010.
\newblock ISSN 0018-9448.
\newblock \doi{10.1109/TIT.2009.2039166}.
\newblock URL \url{https://doi.org/10.1109/TIT.2009.2039166}.

\bibitem[Wang and Renner(2012)]{WR12}
L.~Wang and R.~Renner.
\newblock One-shot classical-quantum capacity and hypothesis testing.
\newblock \emph{Phys. Rev. Lett.}, 108\penalty0 (20):\penalty0 200501, May
  2012.
\newblock \doi{10.1103/PhysRevLett.108.200501}.
\newblock URL \url{https://doi.org/10.1103/PhysRevLett.108.200501}.

\bibitem[Seshadreesan and Wilde(2015)]{SW15}
K.~P. Seshadreesan and Mark~M. Wilde.
\newblock Fidelity of recovery, squashed entanglement, and measurement
  recoverability.
\newblock \emph{Phys. Rev. A}, 92:\penalty0 042321, Oct 2015.
\newblock \doi{10.1103/PhysRevA.92.042321}.
\newblock URL \url{https://doi.org/10.1103/PhysRevA.92.042321}.

\bibitem[Berta and Tomamichel(2016)]{BT16}
M.~Berta and M.~Tomamichel.
\newblock The fidelity of recovery is multiplicative.
\newblock \emph{IEEE Trans. Inf. Theory}, 62\penalty0 (4):\penalty0 1758--1763,
  Apr 2016.
\newblock \doi{10.1109/tit.2016.2527683}.
\newblock URL \url{https://doi.org/10.1109%2Ftit.2016.2527683}.

\bibitem[Shirokov(2017)]{Shirokov2017}
M.~E. Shirokov.
\newblock {Tight uniform continuity bounds for the quantum conditional mutual
  information, for the {H}olevo quantity, and for capacities of quantum
  channels}.
\newblock \emph{J. Math. Phys.}, 58\penalty0 (10):\penalty0 102202, Oct 2017.
\newblock ISSN 0022-2488.
\newblock \doi{10.1063/1.4987135}.
\newblock URL \url{https://doi.org/10.1063/1.4987135}.

\bibitem[Nielsen(1999)]{Nielsen_quantum_transformation}
M.~A. Nielsen.
\newblock Conditions for a class of entanglement transformations.
\newblock \emph{Phys. Rev. Lett.}, 83:\penalty0 436--439, Jul 1999.
\newblock \doi{10.1103/PhysRevLett.83.436}.
\newblock URL \url{https://doi.org/10.1103/PhysRevLett.83.436}.

\bibitem[Yang and Eisert(2009)]{ent_combing1}
D.~Yang and J.~Eisert.
\newblock Entanglement combing.
\newblock \emph{Phys. Rev. Lett.}, 103:\penalty0 220501, Nov 2009.
\newblock \doi{10.1103/PhysRevLett.103.220501}.
\newblock URL \url{https://doi.org/10.1103/PhysRevLett.103.220501}.

\bibitem[Streltsov et~al.(2020)Streltsov, Meignant, and Eisert]{ent_combing2}
A.~Streltsov, C.~Meignant, and J.~Eisert.
\newblock Rates of multipartite entanglement transformations.
\newblock \emph{Phys. Rev. Lett.}, 125:\penalty0 080502, Aug 2020.
\newblock \doi{10.1103/PhysRevLett.125.080502}.
\newblock URL \url{https://doi.org/10.1103/PhysRevLett.125.080502}.

\bibitem[Oppenheim(2008)]{oppenheim2008paradigm}
J.~Oppenheim.
\newblock A paradigm for entanglement theory based on quantum communication.
\newblock \emph{arXiv:0801.0458}, Jan 2008.
\newblock \doi{10.48550/arXiv.0801.0458}.
\newblock URL \url{https://doi.org/10.48550/arXiv.0801.0458}.

\bibitem[Liu et~al.(2019)Liu, Bu, and Takagi]{LBT19}
Z.~Liu, K.~Bu, and R.~Takagi.
\newblock One-shot operational quantum resource theory.
\newblock \emph{Phys. Rev. Lett.}, 123:\penalty0 020401, Jul 2019.
\newblock \doi{10.1103/PhysRevLett.123.020401}.
\newblock URL \url{https://doi.org/10.1103/PhysRevLett.123.020401}.

\bibitem[Regula et~al.(2020)Regula, Bu, Takagi, and Liu]{RBTL20}
B.~Regula, K.~Bu, R.~Takagi, and Z.~Liu.
\newblock Benchmarking one-shot distillation in general quantum resource
  theories.
\newblock \emph{Phys. Rev. A}, 101:\penalty0 062315, Jun 2020.
\newblock \doi{10.1103/PhysRevA.101.062315}.
\newblock URL \url{https://doi.org/10.1103/PhysRevA.101.062315}.

\bibitem[Laine et~al.(2010)Laine, Piilo, and Breuer]{Laine10}
E.~Laine, J.~Piilo, and H.~Breuer.
\newblock Measure for the non-{M}arkovianity of quantum processes.
\newblock \emph{Phys. Rev. A}, 81:\penalty0 062115, Jun 2010.
\newblock \doi{10.1103/PhysRevA.81.062115}.
\newblock URL \url{https://doi.org/10.1103/PhysRevA.81.062115}.

\bibitem[Luo et~al.(2012)Luo, Fu, and Song]{Luo12}
S.~Luo, S.~Fu, and H.~Song.
\newblock Quantifying non-{M}arkovianity via correlations.
\newblock \emph{Phys. Rev. A}, 86:\penalty0 044101, Oct 2012.
\newblock \doi{10.1103/PhysRevA.86.044101}.
\newblock URL \url{https://doi.org/10.1103/PhysRevA.86.044101}.

\bibitem[Chruściński and Maniscalco(2014)]{Dariusz14}
D.~Chruściński and S.~Maniscalco.
\newblock Degree of non-{M}arkovianity of quantum evolution.
\newblock \emph{Phys. Rev. Lett.}, 112:\penalty0 120404, Mar 2014.
\newblock \doi{10.1103/PhysRevLett.112.120404}.
\newblock URL \url{https://doi.org/10.1103/PhysRevLett.112.120404}.

\bibitem[Buscemi and Datta(2016)]{BD16}
F.~Buscemi and N.~Datta.
\newblock Equivalence between divisibility and monotonic decrease of
  information in classical and quantum stochastic processes.
\newblock \emph{Physical Review A}, 93:\penalty0 012101, Jan 2016.
\newblock \doi{10.1103/PhysRevA.93.012101}.
\newblock URL \url{https://doi.org/10.1103/PhysRevA.93.012101}.

\bibitem[Das et~al.(2018)Das, Khatri, Siopsis, and Wilde]{Das2018}
S.~Das, S.~Khatri, G.~Siopsis, and Mark~M. Wilde.
\newblock {Fundamental limits on quantum dynamics based on entropy change}.
\newblock \emph{Journal of Mathematical Physics}, 59\penalty0 (1):\penalty0
  012205, Jan 2018.
\newblock ISSN 0022-2488.
\newblock \doi{10.1063/1.4997044}.
\newblock URL \url{https://doi.org/10.1063/1.4997044}.

\bibitem[Bhattacharya et~al.(2020)Bhattacharya, Bhattacharya, and
  Majumdar]{Bhattacharya21}
S.~Bhattacharya, B.~Bhattacharya, and A~S Majumdar.
\newblock Convex resource theory of non-{M}arkovianity.
\newblock \emph{J. Phys. A Math. Theor.}, 54\penalty0 (3):\penalty0 035302, Dec
  2020.
\newblock \doi{10.1088/1751-8121/abd191}.
\newblock URL \url{https://doi.org/10.1088/1751-8121/abd191}.

\bibitem[Kuroiwa and Yamasaki(2021)]{Kuroiwa21}
K.~Kuroiwa and H.~Yamasaki.
\newblock Asymptotically consistent measures of general quantum resources:
  Discord, non-markovianity, and non-{G}aussianity.
\newblock \emph{Phys. Rev. A}, 104:\penalty0 L020401, Aug 2021.
\newblock \doi{10.1103/PhysRevA.104.L020401}.
\newblock URL \url{https://doi.org/10.1103/PhysRevA.104.L020401}.

\bibitem[Berk et~al.(2021)Berk, Garner, Yadin, Modi, and Pollock]{Berk2021}
G.~D. Berk, J.~P.~Andrew Garner, B.~Yadin, Kavan Modi, and Felix~A. Pollock.
\newblock Resource theories of multi-time processes: {A} window into quantum
  non-{M}arkovianity.
\newblock \emph{Quantum}, 5:\penalty0 435, Apr 2021.
\newblock ISSN 2521-327X.
\newblock \doi{10.22331/q-2021-04-20-435}.
\newblock URL \url{https://doi.org/10.22331/q-2021-04-20-435}.

\bibitem[Bylicka et~al.(2013)Bylicka, Chru{\'s}ci{\'n}ski, and
  Maniscalco]{bylicka2013}
B.~Bylicka, D.~Chru{\'s}ci{\'n}ski, and S.~Maniscalco.
\newblock Non-{M}arkovianity as a resource for quantum technologies.
\newblock \emph{arXiv:1301.2585}, Jan 2013.
\newblock \doi{10.48550/arXiv.1301.2585}.
\newblock URL \url{https://doi.org/10.48550/arXiv.1301.2585}.

\bibitem[Anand and Brun(2019)]{anand2019}
N.~Anand and T.~A Brun.
\newblock Quantifying non-{M}arkovianity: a quantum resource-theoretic
  approach.
\newblock \emph{arXiv:1903.03880}, Mar 2019.
\newblock \doi{10.48550/arXiv.1903.03880}.
\newblock URL \url{https://doi.org/10.48550/arXiv.1903.03880}.

\bibitem[Costa and Shrapnel(2016)]{Costa2016}
F.~Costa and S.~Shrapnel.
\newblock Quantum causal modelling.
\newblock \emph{New Journal of Physics}, 18\penalty0 (6):\penalty0 063032, Jun
  2016.
\newblock \doi{10.1088/1367-2630/18/6/063032}.
\newblock URL \url{https://doi.org/10.1088/1367-2630/18/6/063032}.

\bibitem[Pollock et~al.(2018{\natexlab{b}})Pollock, Rodr\'{\i}guez-Rosario,
  Paternostro, and Modi]{Pollock2018}
Felix~A. Pollock, T.~Rodr\'{\i}guez-Rosario, C.and~Frauenheim, M.~Paternostro,
  and K.~Modi.
\newblock Operational markov condition for quantum processes.
\newblock \emph{Phys. Rev. Lett.}, 120:\penalty0 040405, Jan
  2018{\natexlab{b}}.
\newblock \doi{10.1103/PhysRevLett.120.040405}.
\newblock URL \url{https://doi.org/10.1103/PhysRevLett.120.040405}.

\bibitem[Pollock et~al.(2018{\natexlab{c}})Pollock, Rodr\'{\i}guez-Rosario,
  Frauenheim, Paternostro, and Modi]{Felix2018}
Felix~A. Pollock, C.~Rodr\'{\i}guez-Rosario, T.~Frauenheim, M.~Paternostro, and
  K.~Modi.
\newblock Non-markovian quantum processes: Complete framework and efficient
  characterization.
\newblock \emph{Phys. Rev. A}, 97:\penalty0 012127, Jan 2018{\natexlab{c}}.
\newblock \doi{10.1103/PhysRevA.97.012127}.
\newblock URL \url{https://doi.org/10.1103/PhysRevA.97.012127}.

\bibitem[Taranto et~al.(2019)Taranto, Milz, Pollock, and Modi]{Taranto2019}
P.~Taranto, S.~Milz, Felix~A. Pollock, and K.~Modi.
\newblock Structure of quantum stochastic processes with finite markov order.
\newblock \emph{Phys. Rev. A}, 99:\penalty0 042108, Apr 2019.
\newblock \doi{10.1103/PhysRevA.99.042108}.
\newblock URL \url{https://doi.org/10.1103/PhysRevA.99.042108}.

\bibitem[Milz et~al.(2021)Milz, Spee, Xu, Pollock, Modi, and Gühne]{Milz2021}
S.~Milz, C.~Spee, Z.P. Xu, Felix~A. Pollock, K~Modi, and O.~Gühne.
\newblock {Genuine multipartite entanglement in time}.
\newblock \emph{SciPost Phys.}, 10:\penalty0 141, Jun 2021.
\newblock \doi{10.21468/SciPostPhys.10.6.141}.
\newblock URL \url{https://doi.org/10.21468/SciPostPhys.10.6.141}.

\bibitem[Milz and Modi(2021)]{Milz2021a}
S.~Milz and K.~Modi.
\newblock Quantum stochastic processes and quantum non-markovian phenomena.
\newblock \emph{PRX Quantum}, 2:\penalty0 030201, Jul 2021.
\newblock \doi{10.1103/PRXQuantum.2.030201}.
\newblock URL \url{https://doi.org/10.1103/PRXQuantum.2.030201}.

\bibitem[Nery et~al.(2021)Nery, Quintino, G., Maciel, and Vianna]{Nery2021}
M.~Nery, M.~T{\'{u}}lio Quintino, Philippe~Allard G., Thiago~O. Maciel, and
  Reinaldo~O. Vianna.
\newblock Simple and maximally robust processes with no classical common-cause
  or direct-cause explanation.
\newblock \emph{{Quantum}}, 5:\penalty0 538, Sep 2021.
\newblock ISSN 2521-327X.
\newblock \doi{10.22331/q-2021-09-09-538}.
\newblock URL \url{https://doi.org/10.22331/q-2021-09-09-538}.

\bibitem[Fuchs and Van~DeLam24Graaf(1999)]{FV99}
C.~A. Fuchs and J.~Van~DeLam24Graaf.
\newblock Cryptographic distinguishability measures for quantum-mechanical
  states.
\newblock \emph{IEEE Transactions on Information Theory}, 45\penalty0
  (4):\penalty0 1216--1227, May 1999.
\newblock \doi{10.1109/18.761271}.
\newblock URL \url{https://doi.org/10.1109/18.761271}.

\bibitem[Uhlmann(1976)]{uhlmann_transition_1976}
A.~Uhlmann.
\newblock The “transition probability” in the state space of a
  $\star$-algebra.
\newblock \emph{Rep. Math. Phys.}, 9\penalty0 (2):\penalty0 273--279, Apr 1976.
\newblock ISSN 00344877.
\newblock \doi{10.1016/0034-4877(76)90060-4}.
\newblock URL \url{https://doi.org/10.1016/0034-4877(76)90060-4}.

\bibitem[Jozsa(1994)]{jozsa_fidelity_1994}
R.~Jozsa.
\newblock Fidelity for mixed quantum states.
\newblock \emph{J. Mod. Opt.}, 41\penalty0 (12):\penalty0 2315--2323, Dec 1994.
\newblock ISSN 0950-0340, 1362-3044.
\newblock \doi{10.1080/09500349414552171}.
\newblock URL \url{https://doi.org/10.1080/09500349414552171}.

\bibitem[Fuchs and Van De~Graaf(1999)]{fuchs_cryptographic_1999}
C.A. Fuchs and J.~Van De~Graaf.
\newblock Cryptographic distinguishability measures for quantum-mechanical
  states.
\newblock \emph{IEEE Trans. Inform. Theory}, 45\penalty0 (4):\penalty0
  1216--1227, May 1999.
\newblock ISSN 00189448.
\newblock \doi{10.1109/18.761271}.
\newblock URL \url{https://doi.org/10.1109/18.761271}.

\bibitem[Kretschmann et~al.(2008{\natexlab{a}})Kretschmann, Schlingemann, and
  Werner]{kretschmann_information-disturbance_2006}
D.~Kretschmann, D.~Schlingemann, and R.F. Werner.
\newblock The information-disturbance tradeoff and the continuity of
  {S}tinespring's representation.
\newblock \emph{IEEE Trans. Inf. Theory}, 54\penalty0 (4):\penalty0 1708--1717,
  Apr 2008{\natexlab{a}}.
\newblock ISSN 1557-9654.
\newblock \doi{10.1109/TIT.2008.917696}.
\newblock URL \url{https://doi.org/10.1109/TIT.2008.917696}.

\bibitem[Kretschmann et~al.(2008{\natexlab{b}})Kretschmann, Schlingemann, and
  Werner]{kretschmann_continuity_2007}
D.~Kretschmann, D.~Schlingemann, and R.~F. Werner.
\newblock A continuity theorem for {S}tinespring's dilation.
\newblock \emph{J. Funct. Anal.}, 255\penalty0 (8):\penalty0 1889--1904, Oct
  2008{\natexlab{b}}.
\newblock ISSN 0022-1236.
\newblock \doi{10.1016/j.jfa.2008.07.023}.
\newblock URL \url{https://doi.org/10.1016/j.jfa.2008.07.023}.

\bibitem[P{\'{e}}rez-Garc{\'{\i}}a et~al.(2006)P{\'{e}}rez-Garc{\'{\i}}a, Wolf,
  Petz, and Ruskai]{P_rez_Garc_a_2006}
D.~P{\'{e}}rez-Garc{\'{\i}}a, M.~M. Wolf, D.~Petz, and M.~B. Ruskai.
\newblock Contractivity of positive and trace-preserving maps under lp norms.
\newblock \emph{J. Math. Phys.}, 47\penalty0 (8):\penalty0 083506, Aug 2006.
\newblock \doi{10.1063/1.2218675}.
\newblock URL \url{https://doi.org/10.1063%2F1.2218675}.

\bibitem[Wilde(2016)]{Wilde_2016}
M.~M. Wilde.
\newblock Squashed entanglement and approximate private states.
\newblock \emph{Quantum Inf. Process.}, 15\penalty0 (11):\penalty0 4563--4580,
  Sep 2016.
\newblock \doi{10.1007/s11128-016-1432-7}.
\newblock URL \url{https://doi.org/10.1007%2Fs11128-016-1432-7}.

\bibitem[Alicki and Fannes(2004)]{Alicki_2004}
R.~Alicki and M.~Fannes.
\newblock Continuity of quantum conditional information.
\newblock \emph{J. Phys. A Math. Theor.}, 37\penalty0 (5):\penalty0 L55, Jan
  2004.
\newblock \doi{10.1088/0305-4470/37/5/L01}.
\newblock URL \url{https://doi.org/10.1088/0305-4470/37/5/L01}.

\bibitem[Fannes(1973)]{fannes1973continuity}
M.~Fannes.
\newblock A continuity property of the entropy density for spin-lattice
  systems.
\newblock \emph{Commun. Math. Phys.}, 31:\penalty0 291--294, 1973.
\newblock \doi{10.1007/BF01646490}.
\newblock URL \url{https://doi.org/10.1007/BF01646490}.

\bibitem[Devetak et~al.(2008)Devetak, Harrow, and
  Winter]{Devetak_2008_resource_framework}
I.~Devetak, A.~W. Harrow, and A.~J. Winter.
\newblock A resource framework for quantum {S}hannon theory.
\newblock \emph{IEEE Trans. Inf. Theory}, 54\penalty0 (10):\penalty0
  4587--4618, Oct 2008.
\newblock \doi{10.1109/tit.2008.928980}.
\newblock URL \url{https://doi.org/10.1109%2Ftit.2008.928980}.

\end{thebibliography}

\end{document}